\documentclass{theoretics}

\usepackage{float}
\floatstyle{boxed}
\restylefloat{figure}
\usepackage{graphbox}
\usepackage{subcaption}
\usepackage{breakcites}

\ThCSnewtheostd{assumption}

\numberwithin{theorem}{section}
\numberwithin{equation}{section}

\newcommand{\Z}{\mathbb{Z}} 
 
\newcommand{\R}{\mathbb{R}} 
\newcommand{\C}{\mathbb{C}} 
\newcommand{\G}{\mathbb{G}}

\newcommand{\ev}{{\mathbf{e}}}

\newcommand{\gv}{{\mathbf{g}}}
\newcommand{\hv}{{\mathbf{h}}}

\newcommand{\rv}{{\mathbf{r}}}
\newcommand{\sv}{{\mathbf{s}}}
\newcommand{\tv}{{\mathbf{t}}}
\newcommand{\uv}{{\mathbf{u}}}
\renewcommand{\vv}{{\mathbf{v}}}
\newcommand{\wv}{{\mathbf{w}}}
\newcommand{\xv}{{\mathbf{x}}}
\newcommand{\yv}{{\mathbf{y}}}

\newcommand{\Am}{{\mathbf{A}}}

\newcommand{\As}{{\mathcal{A}}}
\newcommand{\Bs}{{\mathcal{B}}}
\newcommand{\Cs}{{\mathcal{C}}}
\newcommand{\Ds}{{\mathcal{D}}}
\newcommand{\Es}{{\mathcal{E}}}
\newcommand{\Fs}{{\mathcal{F}}}
\newcommand{\Gs}{{\mathcal{G}}}
\newcommand{\Hs}{{\mathcal{H}}}

\newcommand{\Ms}{{\mathcal{M}}}

\newcommand{\Os}{{\mathcal{O}}}

\newcommand{\Ss}{{\mathcal{S}}}
\newcommand{\Ts}{{\mathcal{T}}}

\newcommand{\Xs}{{\mathcal{X}}}
\newcommand{\Ys}{{\mathcal{Y}}}

\newcommand{\negl}{{\sf negl}}

\newcommand{\setupprotocol}[2]{%
	\expandafter\newcommand\csname protocol#1\endcsname{{\sf \Pi_{#2}}}%
	\expandafter\newcommand\csname gen#1\endcsname{{\sf Gen_{#2}}}%
	\expandafter\newcommand\csname genlossy#1\endcsname{{\sf GenLossy_{#2}}}%
	\expandafter\newcommand\csname setup#1\endcsname{{\sf Setup_{#2}}}%
	\expandafter\newcommand\csname enc#1\endcsname{{\sf Enc_{#2}}}%
	\expandafter\newcommand\csname dec#1\endcsname{{\sf Dec_{#2}}}%
	\expandafter\newcommand\csname extract#1\endcsname{{\sf Extract_{#2}}}%
	\expandafter\newcommand\csname derive#1\endcsname{{\sf Derive_{#2}}}%
	\expandafter\newcommand\csname embed#1\endcsname{{\sf Embed_{#2}}}%
	\expandafter\newcommand\csname rerand#1\endcsname{{\sf ReRand_{#2}}}%
	\expandafter\newcommand\csname trace#1\endcsname{{\sf Trace_{#2}}}%
	\expandafter\newcommand\csname goodtrace#1\endcsname{{\sf GoodTr_{#2}}}%
	\expandafter\newcommand\csname badtrace#1\endcsname{{\sf BadTr_{#2}}}%
	\expandafter\newcommand\csname gooddecoder#1\endcsname{{\sf GoodDec_{#2}}}%
	\expandafter\newcommand\csname goodf#1\endcsname{{\sf Goodfunc_{#2}}}%
	\expandafter\newcommand\csname badextract#1\endcsname{{\sf BadExtr_{#2}}}%
	\expandafter\newcommand\csname adv#1\endcsname{{\As_{#2}}}%
	\expandafter\newcommand\csname sample#1\endcsname{{\sf Sample_{#2}}}%
	\expandafter\newcommand\csname diff#1\endcsname{{\sf Diff_{#2}}}%
	\expandafter\newcommand\csname identify#1\endcsname{{\sf Indentify_{#2}}}%
	\expandafter\newcommand\csname findtags#1\endcsname{{\sf FindTags_{#2}}}%
	\expandafter\newcommand\csname confirmtags#1\endcsname{{\sf ConfirmTags_{#2}}}%
	\expandafter\newcommand\csname eval#1\endcsname{{\sf Eval_{#2}}}%
	\expandafter\newcommand\csname sign#1\endcsname{{\sf Sign_{#2}}}%
	\expandafter\newcommand\csname ver#1\endcsname{{\sf Ver_{#2}}}%
	\expandafter\newcommand\csname prf#1\endcsname{{\sf F_{#2}}}%
	\expandafter\newcommand\csname Sim#1\endcsname{{\sf Sim_{#2}}}%

	\expandafter\newcommand\csname msk#1\endcsname{{\sf msk_{#2}}}%
	\expandafter\newcommand\csname sk#1\endcsname{{\sf sk_{#2}}}%
	\expandafter\newcommand\csname pk#1\endcsname{{\sf pk_{#2}}}%
}

\setupprotocol{}{}

\newcommand{\QFT}{{\sf QFT}}
\newcommand{\GGAM}{{\sf GGAM}}

\newcommand{\ignore}[1]{}

\floatstyle{plain}
\restylefloat{figure}

\title{Quantum Money from Abelian Group Actions}
\ThCSauthor[aff1]{Mark Zhandry}{mzhandry@stanford.edu}[0000-0001-7071-6272]
\ThCSaffil[aff1]{ NTT Research \& Stanford University }
\ThCSthanks{This article was invited from ITCS 2024. \cite{DBLP:conf/innovations/Zhandry24}}
\ThCSshortnames{M. Zhandry}  
\ThCSshorttitle{Quantum Money from Abelian Group Actions}
\ThCSyear{2025}
\ThCSarticlenum{18}
\ThCSreceived{Apr 6, 2024} 
\ThCSrevised{Apr 17, 2025}
\ThCSaccepted{Jun 29, 2025}
\ThCSpublished{Aug 19, 2025}
\ThCSdoicreatedtrue
\ThCSkeywords{Quantum money, group actions, isogenies}

\begin{document}

\maketitle

\begin{abstract}We give a construction of public key quantum money, and even a strengthened version called quantum lightning, from abelian group actions, which can in turn be constructed from suitable isogenies over elliptic curves. We prove security in the generic group model for group actions under a plausible computational assumption, and develop a general toolkit for proving quantum security in this model. Along the way, we explore knowledge assumptions and algebraic group actions in the quantum setting, finding significant limitations of these assumptions/models compared to generic group actions. 
\end{abstract}

\section{Introduction}\label{sec:intro}

Quantum money, first envisioned by Wiesner~\cite{Wiesner83}, is a system of money where banknotes are quantum states. By the no-cloning theorem, such banknotes cannot be copied, leading to un-counterfeitable currency. A critical goal for quantum money, identified by~\cite{CCC:Aaronson09}, is \emph{public verification}, allowing anyone to verify while only the mint can create new banknotes. Such public key quantum money is an important central object in the study of quantum protocols, but unfortunately convincing constructions have remained elusive. See Section~\ref{sec:related} for a more thorough discussion of prior work in the area.

\paragraph{This Work.} We construct public key quantum money from abelian group actions, which can be instantiated by suitable isogenies over ordinary elliptic curves. Group actions, and the isogenies they abstract, are one of the leading contenders for post-quantum secure cryptosystems. Our construction could plausibly even be quantum lightning, a strengthening of quantum money with additional applications. Our construction is arguably the first time group actions have been used to solve a classically-impossible cryptographic task that could not already be solved using other standard tools like LWE. Our construction is sketched in Section~\ref{sec:constroverview} below, and given in detail in Section~\ref{sec:constr}.

While our main construction can be instantiated on a clean abelian group action --- often referred to as an ``effective group action'' (EGA) --- many isogeny-based group actions diverge from this convenient abstraction. We therefore provide an alternative candidate scheme which can be instantiated on so-called ``restricted effective group actions'' (REGAs); see Section~\ref{sec:alternate} for details. We prove the quantum lightning security of our protocols in the generic group action model --- a black box model for group actions --- assuming a new but natural strengthening of the discrete log assumption on group actions. Note that generic group actions cannot be used to give unconditional quantum hardness results, so some additional computational assumption is necessary. In order to prove our result, we develop a new toolkit for quantum generic group action proofs; see Section~\ref{sec:ggam}. We believe ours is the first proof of security in the quantum generic group action model. 

Along the way, we explore knowledge assumptions and algebraic group actions in the quantum setting, finding significant limitations of these assumptions/models compared to generic group actions. Specifically, unlike the classical setting where knowledge assumptions typically hold unconditionally against generic attacks, we explain why such statements likely do not hold quantumly. In the specific case of group actions, we indeed show an efficient generic attack on an analog of the ``knowledge of exponent'' assumption. This potentially casts doubt on quantum knowledge assumptions in general. We do give a more complex definition that avoids our attack, but it is unclear if the assumption is sound and more analysis is needed. For completeness, we give an alternative proof of security for our construction under this new knowledge assumption, which avoids generic group actions. 

We also discuss an algebraic model for group actions, which can be seen as a variant of the knowledge of exponent assumption. Unlike the classical setting where algebraic models live ``between'' the fully generic and standard models, we find that the algebraic group action model is likely incomparable to the generic group action model, and security proofs in the model are potentially problematic. As these issues do not appear for generic group actions, we therefore propose that generic group actions are the preferred quantum idealized model for analyzing cryptosystems, instead of the algebraic group action model as argued for in~\cite{PKC:DHKKLR23}. See Section~\ref{sec:knowledge} for details.

We conclude in Section~\ref{sec:discuss} with a discussion of possible generalizations. In particular, we propose the notion of a \emph{quantum} group action where the set elements are quantum states instead of bit strings. We discuss how instantiating our scheme on quantum group actions is closely related to failed approaches for building quantum money from LWE, but different in key ways that seem to allow our scheme to remain secure while the related LWE approaches failed.

\subsection{Our Construction}\label{sec:constroverview}

\paragraph{Abelian Group Actions.} We will use additive group notation for abelian groups. An abelian group action consists of an abelian group $\G$ and a set $\Xs$, such that $\G$ ``acts'' on $\Xs$ through the efficiently computable binary relation $*:\G\times\Xs\rightarrow\Xs$ with the property that $g*(h*x)=(g+h)*x$ for all $g,h\in\G,x\in\Xs$. We will also assume a \emph{regular} group action, which means that for every $x\in\Xs$, the map $g\mapsto g*x$ is a bijection. 

The main group actions used in cryptography are those arising from isogenies over elliptic curves. For example, see~\cite{EPRINT:Couveignes06,EPRINT:RosSto06,AC:CLMPR18,AC:BeuKleVer19,PKC:DFKLMP23}. 
Group action cryptosystems rely at a minimum on the assumed hardness of discrete logarithms: given $x,y=g*x\in\Xs$, finding $g$. In other words, while the map $g\mapsto g*x$ is efficiently computable and has an inverse, the inverse is not efficiently computable. For isogeny-based actions, this corresponds to the hard problem of computing isogenies between elliptic curves. Other hard problems on group actions are possible to consider, such as analogs of computational/decisional Diffie-Hellman, and more.

\paragraph{The QFT.} Our quantum money scheme will utilize the quantum Fourier transform (QFT) over general abelian groups. This is a quantum procedure that maps
\[|g\rangle\mapsto\frac{1}{\sqrt{|\G|}}\sum_{h\in\G}\chi(g,h)|h\rangle\enspace .\]
Here, $\chi$ is some potentially complex phase term. In the case of $\G$ being the additive group $\Z_N$, $\chi(g,h)$ is defined as $e^{i2\pi gh/N}$, with a slightly more complicated definition for non-cyclic groups\footnote{Remember that the group operation is $+$, so $gh$ in the exponent is not the group operation, but instead multiplication in the ring $\Z_N$.}. The main property we utilize from $\chi$ (besides making the QFT unitary) is that it is \emph{bilinear}, in the sense that $\chi(g,h_1+ h_2)=\chi(g,h_1)\cdot\chi(g,h_2)$. It is also symmetric: $\chi(g,h)=\chi(h,g)$.

\paragraph{Our Quantum Money Scheme.} Our quantum money scheme is as follows; see Section~\ref{sec:constr} for additional details.

\begin{itemize}
	\item $\gen$: initialize a register in the state $\frac{1}{\sqrt{|\G|}}\sum_{g\in\G}|g\rangle$, which can be computed by applying the QFT to $|0\rangle$. Let $x\in\Xs$ be arbitrary. Then by computing the group action in superposition, compute $\frac{1}{\sqrt{|\G|}}\sum_{h\in\G}|g\rangle|g*x\rangle$. Next, apply the QFT over $\G$ to the first register. The result is:
	\[\frac{1}{|\G|}\sum_{g,h\in\G} \chi(g,h)|h\rangle|g*x\rangle=\frac{1}{\sqrt{|\G|}}\sum_h |h\rangle|\G^h*x\rangle\]
	Here, $|\G^h*x\rangle$ is the state $\frac{1}{\sqrt{|\G|}}\sum_{g\in\G}\chi(g,h)|g*x\rangle$. Note that $|\G^h*x\rangle$ is, up to an overall phase, independent of $x$. 
	
	Now measure $h$, in which case the second register collapses to $|\G^h*x\rangle$. Output $h$ as the serial number, and $|\G^h*x\rangle$ as the money state.
	\item To verify a banknote $\$$, we do the following\footnote{In an initial version of this work, we had a more complicated verification. The simplified version here was pointed out to us by Jake Doliskani.}: Initialize a new register in the state $|\phi\rangle:=\frac{1}{\sqrt{|\G|}}\sum_{u\in\G}|u\rangle$. Then apply the map $(u,y)\mapsto(u,(-u)*y)$ to the joint system $|\phi\rangle\times \$$\enspace\footnote{Note that we used the ``minimal'' oracle here for the group action computation, having $(-u)*y$ replace $y$, instead of being written to a response register as in the standard quantum oracle. However, since the computation $y\mapsto (-u)*y$ is efficiently reversible given $u$ (by $y\mapsto u*y$), we can easily implement the minimal oracle efficiently by first computing $|(-u)*y\rangle$, then uncomputing $|y\rangle$ using the efficient inverse, and finally swapping in $|(-u)*y\rangle$.}. In the case where $\$$ is the honest banknote $|\G^h*x\rangle$, the result is
	\begin{align*}\frac{1}{|\G|}\sum_{u\in\G}|u\rangle \sum_{g\in\G}\chi(g,h)|(g-u)*x\rangle&=\frac{1}{|\G|}\sum_{u\in\G}|u\rangle \sum_{g'\in\G}\chi(g'+u,h)|g'*x\rangle\\&=\frac{1}{|\G|}\sum_{u\in\G}\chi(u,h)|u\rangle\sum_{g\in\G}\chi(g',h)|g'*x\rangle\\
		&=\left(\frac{1}{\sqrt{|\G|}}\sum_{u\in\G}\chi(u,h)|u\rangle\right)|\G^h*x\rangle\end{align*}
	where we used the substitution $g'=g-u$. Thus we see that this process preserves the honest banknote state $|\G^h*x\rangle$. Moreover, if we apply the inverse QFT to the first register, the result for honest banknotes is $|h\rangle$, and for any state orthogonal to the honest banknote, the result of the inverse QFT will be something orthogonal to $|h\rangle$. Thus by measuring this register and checking if the result is $h$, we can distinguish the honest banknote state from any other state.
\end{itemize}

\paragraph{An instantiation using REGAs.} In some isogeny-based group actions such as CSIDH~\cite{AC:CLMPR18}, the operation $*$ is only efficiently computable for a very small set $S\subseteq\G$ of group elements. Such group actions are called ``restricted effective group actions'' (REGAs)~\cite{AC:ADMP20}. Above, however, we see that we need to compute the group action on all possible elements in $\G$, both for minting and for verification. We therefore give a variant of the construction above which only uses the ability to compute $*$ for elements in $S$. We show that we are still able to sample $|\G^h*x\rangle$, but now the serial number has the form $\Am^T h+\ev\bmod N$ for a known matrix $\Am$ and a ``small'' $\ev\in\Z^n$\enspace\footnote{Here, we are interpreting $h$ a vector in $\Z_N^n$ for some $n,N$, which is possible since $\G$ is abelian.}. Under plausible assumptions, the serial number actually hides $h$\enspace\footnote{This is the search Learning with Errors (search LWE) problem~\cite{DBLP:journals/jacm/Regev09} which is widely believed to be hard for \emph{random} $\Am$. In our case, $\Am$ is a fixed matrix that depends on the group action, and LWE may or may not be hard for this $\Am$. However, if LWE is easy for this $\Am$, then we in fact have a plain group action. Indeed, a variant of Regev's quantum reduction between LWE and Short Integer Solution (SIS)~\cite{DBLP:journals/jacm/Regev09}, outlined by~\cite{DBLP:journals/joc/Zhandry21}, shows that if LWE can be solved relative to $\Am$, then SIS can be solved for $\Am$ as well. It is straightforward to adapt this reduction to solve the Inhomogeneous SIS (ISIS) problem, which then allows for computing the group action for all of $\G$. In this case we would have a clean group action and would not need this alternate construction.}. We nevertheless show that we can use such a noisy serial number for verification. For details, see Section~\ref{sec:alternate}. The security of our alternate scheme is essentially equivalent to the main scheme.

\subsection{The security of our scheme}\label{sec:introsec}

We do not know how to base the security of our schemes on any standard assumptions on isogenies. However, we are able to prove the security of our scheme in a black box model for group actions called the generic group action model (GGAM), an analog of the generic group model~\cite{EC:Shoup97,IMA:Maurer05} adapted to group actions. Generic models for group actions have been considered previously~\cite{DBLP:journals/joc/MontgomeryZ24,EC:BonGuaZha23,EPRINT:OrsZan23,PKC:DHKKLR23}. While the model is motivated by post-quantum security, to the best of our knowledge ours is the first time the model has been used to actually prove security against quantum attacks.

The challenge with the quantum GGAM is that the query complexity of computing discrete logarithms is actually polynomial (this follows from~\cite{EttHoy00}; see Section~\ref{sec:ggam} for an explanation). This means we cannot rely on query complexity alone to justify hardness, and must additionally make computational assumptions. This is in contrast to the classical setting, where the generic group (action) model allows for unconditional proofs of security by analyzing query complexity alone. In fact, most if not all generic group model proofs from the classical setting are unconditional query complexity proofs. This means that proofs in the quantum GGAM will look very different than classical proofs in the GGM/GGAM; in particular, proofs will still require a reduction from an underlying hard computational problem. At the same time, in order to take advantage of the generic oracle setting, it would seem that quantum query complexity arguments are still needed. But a priori, it may not be obvious how to leverage query complexity in any useful way, given the preceding discussion.

\paragraph{Our Framework.} In Section~\ref{sec:ggam}, we develop a new framework to help in the task of proving quantum hardness results relative to generic group actions. To illustrate our ideas, we consider the following warm-up task. An important feature in some isogeny-based group actions are twists, which allow for computing ``negations'': computing $(-g)*x$ from $g*x$. An interesting question is whether this additional structure makes computing discrete logarithms easier. Here, we show that for generic group actions, such negations are unlikely to make discrete logarithms any easier than in group actions without negations. Concretely, we will show that discrete logarithms are generically hard, assuming a plausible computational assumption on some group action where such negation queries are \emph{not} permitted. 

Suppose toward contradiction that there was a generic adversary which could utilize negation queries to solve discrete logarithms. Let $(*,\G,\Xs)$ be a plain group action where negation queries are not allowed. We will define a new group action $(\star,\G,\Xs')$ as follows. First sample a random injection $\Pi:\Xs^2\rightarrow\{0,1\}^m$ whose inputs are \emph{pairs} of set elements. Then define $\Xs'$ as the image under $\Pi$ of pairs of the form $(g*x,(-g)*x)$. $\star$ acts in the natural way: $g\star \Pi(y,z)=\Pi(g*y,(-g)*z)$.

Our reduction will sample a $\Pi$\enspace\footnote{A random injection is exponentially large and cannot be sampled efficiently. Instead, the reduction will actually efficiently simulate a random injection $\Pi$ using known techniques. For the purposes of our discussion here, we can ignore this issue.} and run the generic adversary on the new group action, using its knowledge of $\Pi$ and its inverse to implement the action $\star$. Notice now that our reduction also has the ability to compute negations: given $\Pi(y,z)$ where $y=g*x$ and $z=(-g)*x$, the negation of $\Pi(y,z)$ is exactly the element $\Pi(z,y)$ obtained by swapping $y$ and $z$. Thus, our reduction is able to simulate the negation queries, even though the underlying group action does not support efficient negations. This is our main idea, though there are a couple lingering issues to sort out:
\begin{itemize}
	\item The reduction cannot perfectly simulate $(\star,\G,\Xs')$. The issue is that there are elements $\Pi(y,z)$ where $y,z$ do not have the form $y=g*x,z=(-g)*x$ for some $g$. In the group action $(\star,\G,\Xs')$, these elements will be identified as invalid set elements. On the other hand, while our reduction can carry out the correct computation on $y,z$ of the correct form, it will be unable to distinguish such $y,z$ from ones of the incorrect form, and will act on these elements even though they are incorrect. As such, there will be elements that are not in $\Xs'$ that the reduction will nevertheless falsely identify as valid set elements. We resolve this problem by choosing the images of $\Pi$ to be somewhat sparse, by setting the output length $m$ sufficiently large. Our reduction only provides the adversary elements corresponding to valid $y,z$, and we can show, roughly, that the adversary has a negligible chance of computing elements in the image of $\Pi$ that correspond to invalid $y,z$. This follows from standard query complexity arguments. Thus, we are able to simulate with negligible error the correct group action $(\star,\G,\Xs')$.
	\item We have not yet specified what problem the reduction actually solves. The problem we would like to solve is the plain discrete logarithm on $(*,\G,\Xs)$, where the reduction is given $g*x$, and must compute $g$. However, it is unclear what challenge the reduction should give to the adversary. The natural approach is to try to give the adversary $\Pi(g*x,(-g)*x)$, which is just the discrete log instance relative to $(\star,\G,\Xs')$ with the same solution $g$; the reduction can then simply output whatever the adversary outputs. However, this requires the reduction to know $(-g)*x$, which is presumably hard to compute given just $g*x$ (remember that negation queries are not allowed on $(*,\G,\Xs)$). Our solution is to simply use a slight strengthening of discrete logarithms, where the adversary is given $(g*x,(-g)*x)$ and must compute $g$. Under the assumed hardness of this strengthened discrete log problem (again, in ordinary group actions where negations are presumed hard), we can complete the reduction and prove the generic hardness of discrete logarithms in the presence of negation queries.
\end{itemize}

\paragraph{The security of our money scheme.} We now turn to using our framework to prove the security of our quantum money scheme in the GGAM. Inspired by our negation example above, we will simulate a generic group action $(\star,\G,\Xs')$ using an injection $\Pi$ applied to a vector of set elements. Our goal will be to use two banknotes with the same serial number relative to $(\star,\G,\Xs')$ in order to break some distinguishing problem relative to $(*,\G,\Xs)$. Any quantum money adversary yields such a pair of banknotes, and so if the distinguishing problem is hard, then there can be no such efficient quantum money adversary. This argument in fact shows the scheme attains the stronger notion of quantum lightning~\cite{DBLP:journals/joc/Zhandry21}, which has additional applications.

Concretely, our starting assumption gives the adversary $y=u*x$ for a random $u$, and then allows the adversary a single quantum query to $z\mapsto v*z$ for an unknown $v$, where either $v$ is random or $v=2u$. The adversary then has to tell whether $v=2u$ or not. It is straightforward to prove this assumption is true in the classical GGAM. In fact, it is a quantum analog of the classical group-based problem of distinguishing $g^a,g^b$ from $g^a,g^{a^2}$ for a group generator $g$, a widely used Diffie-Hellman-like assumption. Under this analogy, $g$ plays the role of $x$, $a$ plays the role of $u$, and $b$ plays the role of $v$. The main difference from the classical assumption (besides being over group actions instead of groups) is that, instead of receiving $g^b$ or $g^{a^2}$, the adversary receives $h^b$ or $h^{a^2}$ for an adversarially chosen $h$, and we allow the adversary's $h$ to be in superposition.

Our idea is to have $\Xs'$ be elements of the form $\Pi(g*x,g*y)$ where $y=u*x$ is the challenge given by the assumption. Let $X=\Pi(x,y)\in\Xs'$. Now consider the output of a successful adversary, which is two copies of the banknote $|\G^h\star X\rangle$ relative to $(\star,\G,\Xs')$ for some serial number $h$. Now consider applying the following process to, say, the first copy: map any element $\Pi(z_1,z_2)$ in the range of $\Pi$ to $\Pi(z_2,v*z_1)$, where we compute $v*z_1$ from $z_1$ using the challenge oracle. We then observe that if $v=2u$, this process preserves the banknote:
\begin{align*}
	|\G^h\star X'\rangle&=\frac{1}{\sqrt{|\G|}}\sum_{g\in\G}\chi(g,h)|g\star \Pi(x,y)\rangle=\frac{1}{\sqrt{|\G|}}\sum_{g\in\G}\chi(g,h)|\Pi(g*x,g*y)\rangle\\
	&\mapsto\frac{1}{\sqrt{|\G|}}\sum_{g\in\G}\chi(g,h)|\Pi(g*y,(g+2u)*x)\rangle\\
	&=\frac{1}{\sqrt{|\G|}}\sum_{g\in\G}\chi(g,h)|\Pi((g+u)*x,(g+2u)*x)\rangle\\
	&=\chi(-u,h)\frac{1}{\sqrt{|\G|}}\sum_{g'\in\G}\chi(g',h)|\Pi(g'*x,g'*y)\rangle=\chi(-u,h)|\G^h\star X'\rangle\enspace .
\end{align*}
Above, we used the substitution $g'=g+u$.

On the other hand, if $v\neq 2u$, then the transformation will produce a state whose support is not even on $\Xs'$. In particular, the transformed state would be orthogonal to the original state. So our reduction will apply the above transformation to one copy of $|\G^h\star X\rangle$, leaving the other as is. Then it will perform the SWAP test on the two states. If $v=2u$, the states will be identical and the SWAP test will accept. If $v\neq 2u$, the states will be orthogonal, and the swap test will accept only with probability $1/2$. Thus, we achieve a distinguishing advantage between the two cases, contradicting the assumption.

\medskip

We believe our proof gives convincing evidence that our scheme should be secure on a suitable group action, perhaps even those based on isogenies over elliptic curves. However, our underlying assumption is new, and needs further cryptanalysis. One limitation of our assumption is that it is interactive, requiring a (quantum) oracle query to the challenger. One may hope instead to use a non-interactive assumption. We do not know how to make non-interactive assumptions work, in general. In particular, if we do not have an oracle that can transform the input for us, it seems like we are limited to strategies that only permute the inputs to $\Pi$, like in our negation-query example. But since the scheme has to be efficient, the inputs to $\Pi$ can only consist of polynomial-length vectors of set elements. Any permutation on a polynomial-length set must have smooth order. On the other hand, the only permutations on $\Xs'$ which preserve $|\G^h\star X\rangle$ seem to have order that divides $|\G|$. Thus, if, say, the order of $\G$ were a large prime, it does not seem that permuting the inputs to $\Pi$ alone will be able to preserve $|\G^h\star X\rangle$.

\subsection{On Knowledge Assumptions and Algebraic Group Actions} 

In Section~\ref{sec:knowledge}, we show a different approach to justifying the security of our scheme, by adapting certain knowledge assumptions~\cite{EC:LiuMonZha23} to the setting of group actions. Despite some high-level similarities to~\cite{EC:LiuMonZha23}, the underlying details are somewhat different. The advantage of this route is that it gives a standard-model security proof (albeit, using a non-standard knowledge definition) rather than a generic model proof.

However, we find significant issues with using knowledge assumptions quantumly, that appear not to have been observed before. In particular, the straightforward way to adapt the knowledge assumptions of~\cite{EC:LiuMonZha23} to group actions actually results in \emph{false} assumptions, as we demonstrate. Interestingly, our attack on the assumption is entirely generic. This is quite surprising, as in the classical setting, knowledge assumptions generally trivially hold against generic attacks.

Concretely, we show how to construct a superposition over $\Xs$ where the underlying discrete logarithms are hidden, even to the algorithm creating the superposition. To accomplish this, we observe that any set element $x$ can be seen as a superposition over all possible banknotes $|\G^h*x\rangle$; the superposition is uniform up to individual phases. Then we show a procedure to compute, given $|\G^h*x\rangle$, the serial number $h$. This allows us to apply individual phases to the various banknotes in the superposition. Certain phases will simply map $x$ to another set element $y$. But other phases will map $x$ to a uniform superposition (up to phases) over $\Xs$. Call this state $|\psi\rangle$.

Any meaningful knowledge assumption, and in particular the result of adapting~\cite{EC:LiuMonZha23} to group actions, would imply that if we were to measure $|\psi\rangle$ to get a set element $y$, then we must also ``know'' $g$ such that $y=g*x$. However, measuring $|\psi\rangle$ simply gives a uniform set element, importantly without any side information about $y$. As such, under the discrete log assumption, computing such a $g$ is hard.

We resolve this particular problem by re-framing knowledge assumptions as follows: instead of saying that any algorithm $A$ which produces a set element $y$ must know $g$ such that $y=g*x$, we say that for any such $A$ solving some task $T$, there is another algorithm $B$ that also solves $T$ such that $B$ knows $g$, even if $A$ would not. Thus, even if the original $A$ is constructed in such a way that it does not know $g$, at least $B$ does, and we can apply any security arguments to $B$ instead of $A$. We demonstrate that this assumption, together with an appropriate generalization of the discrete log assumption, are enough to prove the security of our scheme. However, we are somewhat skeptical of our new knowledge assumption, and it certainly needs more cryptanalysis.

\paragraph{Algebraic Group Actions.} The Algebraic Group Model (AGM)~\cite{C:FucKilLos18} is an important classical model for studying group-based cryptosystems. It is considered a refinement of the generic group model, meaning that a proof in the model is ``at least as'' convincing as a proof in the generic group model\footnote{There are some caveats to this classical claim; see~\cite{C:Zhandry22b} for discussion.}, potentially even more convincing. A couple of recent works~\cite{PKC:DHKKLR23,EPRINT:OrsZan23} have considered the group action analog, the Algebraic Group Action Model (AGAM). Here, any time an adversary outputs a set element $y$, it must ``explain'' $y$ in terms of one of its input set elements $x_1,\dots,x_n$ by providing a group element $g$ such that $y=g*x_i$.

The AGM can be seen as an idealized model version of the knowledge of exponent assumption, and likewise the AGAM can be seen as an idealized model version of an appropriate knowledge assumption on group actions. After all, a knowledge assumption would say that any time the adversary outputs a $y$, it must ``know'' how it derived $y$ from its inputs. The AGM/AGAM simply require the adversary to actually output this knowledge.

In Section~\ref{sec:knowledge}, we explore the AGAM in the presence of quantum attackers. We do not prove any formal results, but discuss why, unfortunately, the quantum AGAM appears problematic. For starters, given our attack on quantum knowledge assumptions, we are skeptical about the soundness of the quantum AGAM. In particular, our attack indicates that it is unlikely that the AGAM is a refinement of the generic group action model; rather they are likely incomparable.

Another problem we observe with the AGAM is that it requires the adversary to both solve some task, and also produce some extra information, namely the explanation $g$ of any output element $y$. Classically, if the adversary is able to both solve the task and produce this extra information (which would follow from an appropriate knowledge assumption), then the adversary can do both simultaneously, as required by the AGM/AGAM. However, quantumly, even if we believe the adversary can separately solve the task \emph{or} produce the extra information (provided we believe the knowledge assumption), it may be impossible to do both simultaneously, as required by the AGAM.

This issue manifests in the following way: suppose the output is actually a superposition. Then the information $g$ will be entangled with the superposition, meaning the AGAM adversary's output will actually be a different state than if it did not output $g$. For example, if an AGAM adversary had to output a banknote $|\G^h*x\rangle$ (say, as part of the quantum money/lightning experiment), then if it also ``explained'' the banknote by outputting a group element $g$, the entanglement with $g$ would actually cause the banknote state to fail verification. It therefore unclear how to interpret such an adversary. Does it actually break the scheme, even if it does not pass verification? In Section~\ref{sec:knowledge}, we go into more details about this issue as well as pointing out several other issues with the AGAM.

We note that these issues are not present in the generic group action model. Thus, despite classically being a ``worse'' model than the algebraic model, we propose for the quantum setting that the generic group action model is actually \emph{preferred} to the AGAM.

\subsection{Further Discussion}

In Section~\ref{sec:discuss}, we generalize group actions to \emph{quantum} group actions, which replace classical set elements with quantum states, but otherwise behave mostly the same as standard group actions. We give a simple quantum group action based on the Learning with Errors (LWE) problem~\cite{DBLP:journals/jacm/Regev09}, where we can actually prove that the discrete log problem is hard under LWE. Despite this promising result, we expect that the LWE-based quantum group action will be of limited use. In particular, if we instantiate our quantum money construction over this group, the construction is \emph{insecure}. The reason is that, in this group action, it is impossible to recognize the quantum states of the set. Our security proof crucially relies on such recognition in order to characterize states accepted by the verifier. Moreover, without recognition, there is an attack which fools the verifier with dishonest --- and importantly, clonable --- banknotes that are different from the honest ones, breaking security.

Interestingly, we explain that this failed instantiation is actually \emph{equivalent} to a folklore approach toward building quantum money from lattices, an approach that has been more-or-less shown impossible to make secure~\cite{C:LiuZha19,EC:LiuMonZha23}. The \emph{only} missing piece in the folklore approach has been how to efficiently verify honest banknotes. Under our equivalence, this missing piece exactly maps to the problem of recognizing set elements in our quantum group action. For details, see Section~\ref{sec:discuss}. We believe this adds to the confidence of our proposal, since in group actions based on isogenies it is possible to recognize set elements, presumably without otherwise compromising hardness.

\subsection{Related Work}\label{sec:related}

\paragraph{Public key quantum money.} In Wiesner's original scheme, the mint is required to verify banknotes, meaning the mint must be involved in any transaction. The involvement of the mint also leads to potential attacks~\cite{Lutomirski10}. Some partial solutions have been proposed, e.g.~\cite{EPRINT:BehSat20,AC:RobZha21}. The dream solution, however, is known as \emph{public key} quantum money~\cite{CCC:Aaronson09}. Here, anyone can verify the banknote, while only the mint can create them.

Unlike Wiesner's scheme, which is well-understood, secure public key quantum money has remained elusive. While there have been many proposals for public key quantum money~\cite{CCC:Aaronson09,DBLP:journals/toc/AaronsonC13,ITCS:FGHLS12,Kane18,DBLP:journals/joc/Zhandry21,EPRINT:KanShaSil21,KLS22,EC:LiuMonZha23}, they mostly either (1) have been subsequently broken (e.g.~\cite{CCC:Aaronson09,DBLP:journals/toc/AaronsonC13,DBLP:journals/joc/Zhandry21,KLS22} which were broken by~\cite{ITCS:LAFGKH10,PDFHP19,EC:Roberts21,EC:LiuMonZha23}), or (2) rely on new cryptographic building blocks that have received little attention from the cryptographic community (e.g.~\cite{ITCS:FGHLS12,Kane18,EPRINT:KanShaSil21} from problems on knots or quaternion algebras). The two exceptions are:
\begin{itemize}
	\item Building on a suggestion of~\cite{BenSat16}, \cite{DBLP:journals/joc/Zhandry21} proved that quantum money can be built from post-quantum indistinguishability obfuscation (iO). iO has received considerable attention and even has a convincing \emph{pre-quantum} instantiation~\cite{DBLP:journals/cacm/JainLS24}. Yet the post-quantum study of iO is much less thorough. While some post-quantum proposals have been made~\cite{TCC:GenGorHal15,TCC:BGMZ18,EPRINT:BDGM20b,EC:WeeWic21}, their post-quantum hardness is not well-understood.
	\item \cite{EC:LiuMonZha23} construct quantum money from isogenies over super-singular elliptic curves. While isogenies have garnered significant attention from cryptographers, there is a crucial missing piece to their proposal: generating uniform superpositions over super-singular curves, which is currently unknown how to do. This is closely related to the major open question of obliviously sampling super-singular elliptic curves.
\end{itemize}
In light of the above, the existence of public key quantum money is largely considered open.

\paragraph{Cryptography from group actions and isogenies.} Isogenies were first proposed for use in post-quantum cryptography by Couveignes~\cite{EPRINT:Couveignes06} and Rostovtsev and Stolbunov~\cite{EPRINT:RosSto06}. Isogenies give a Diffie-Hellman-like structure, but importantly are immune to Shor's algorithm for discrete logarithms~\cite{FOCS:Shor94} due to a more restricted structure. This restricted structure, while helping preserve security against quantum attacks, also makes the design of cryptosystems based on them more complex. Thus, significant effort has gone into building secure classical cryptosystems from isogenies and understanding their post-quantum security (e.g.~\cite{CJS14,DJP14,AC:CLMPR18,AC:BeuKleVer19,CD20,PKC:DeFMey20,EC:Peikert20,EC:BonSch20,AC:ADMP20,TCC:AlaMalRah22,DBLP:journals/joc/MontgomeryZ24,EPRINT:MaiMar22,EC:CasDec23,EC:BonGuaZha23,EC:Robert23}).

Certain isogenies such as the original proposals of~\cite{EPRINT:Couveignes06,EPRINT:RosSto06} as well as CSIDH and its variants~\cite{AC:CLMPR18,PKC:DFKLMP23} can be abstracted as abelian group actions. However, many other isogenies (such as SIDH~\cite{DJP14}  and OSIDH~\cite{CD20}) cannot be abstracted as abelian group actions. Even among abelian group actions, we must distinguish between ``effective group actions'' (EGAs) and \emph{restricted} EGAs (REGAs). The former satisfies the notion of a clean group action, whereas in the latter, the group action can only be efficiently computed for a certain small set of group elements. CSIDH could plausibly be a EGA at certain concrete security parameters, though asymptotically it only achieves quasi-polynomial security\footnote{With the state-of-the-art, evaluating CSIDH as an EGA would require time approximately $2^{\sqrt[3]{n}}$ on a quantum computer, while the best quantum attack is time $2^{\sqrt{n}}$. For a thorough discussion, see~\cite{Panny23}. By setting $n=\log^3(\lambda)$, one gets polynomial-time evaluation and the best attack taking time $\lambda^{\sqrt{\log(\lambda)}}$.}. Our alternate construction also works on REGAs, which can plausibly be instantiated even asymptotically by CSIDH using a quantum computer\footnote{In order for CSIDH to be a REGA, one needs to compute the structure of the group. While this is hard classically, it is easy with a quantum computer using Shor's algorithm~\cite{FOCS:Shor94}. Since we always assume a quantum computer in this work, we can therefore treat CSIDH as a REGA.}.

While some non-isogeny abelian group actions have been proposed (e.g.~\cite{Stickel05}), currently all such examples have been broken (e.g.~\cite{Shpilrain08}). For this reason, (abelian) group actions are largely considered synonymous with isogenies in the cryptography literature, though this may change if more secure group actions are found.

The vast majority of the isogeny and group action literature has focused on post-quantum cryptography --- classical protocols that are immune to quantum attacks. To the best of our knowledge, only two prior works have used isogenies/group actions to build quantum protocols for tasks that are \emph{impossible} classically. The first is~\cite{TCC:AlaMalRah22}, who build a proof of quantumness~\cite{DBLP:journals/jacm/BrakerskiCMVV21}. We note that proofs of quantumness can also be achieved under several ``standard'' cryptographic tools, such as LWE~\cite{DBLP:journals/jacm/BrakerskiCMVV21} or under certain assumptions on hash functions~\cite{DBLP:journals/jacm/YamakawaZ24}. In contrast, no prior quantum money protocol could be based on similar standard building blocks. We also note that ~\cite{TCC:AlaMalRah22} currently has no known asymptotic instantiation with better-than-quasi-polynomial security, as it requires a clean group action (EGA). The second quantum protocol based on isogenies is that of~\cite{EC:LiuMonZha23}, who build quantum money from walkable invariants, and propose an instantiation using isogenies over super-singular elliptic curves. However, such isogenies cannot be described as abelian group actions, and even more importantly their proposal is incomplete, as discussed above. Thus, ours is arguably the first application of group actions or isogenies to obtain classically impossible tasks that could not already be achieved under standard tools.

\paragraph{Relation to~\cite{EC:LiuMonZha23}.} Aside from using isogenies, our construction has some conceptual similarities to~\cite{EC:LiuMonZha23}, though also crucial differences that allow us to specify a complete protocol, and our idealized-model analysis is completely new. Here, we give a brief overview of the similarities and differences.

The walkable invariant framework of~\cite{EC:LiuMonZha23} is very general, but here we describe a special case of it that would apply to certain group actions, in order to illustrate the differences with our scheme. Consider a group action that is \emph{not} regular, so that the set $\Xs$ is partitioned into many distinct orbits. For $x,y$ in the same orbit there will exist a unique $g$ such that $y=g*x$, but for $x,y$ in different orbits, there will not exist any group element mapping between them. We will also assume the ability to generate a uniform superposition over $\Xs$. We finally assume an ``invariant'', a unique label for each orbit which can be efficiently computed from any element in the orbit.

The minting process generates the uniform superposition over $\Xs$, and then measures the invariant, which becomes the serial number. The state then collapses to a uniform superposition over a single orbit, which becomes the banknote. This superposition can then be verified as follows. First check that the banknote has support on the right orbit by re-computing the invariant. Then check that the state is in uniform superposition by checking that the state is preserved under action by random group elements; this is accomplished using an analog of the swap test.~\cite{EC:LiuMonZha23} prove the security of their scheme under the certain assumptions which, when mapped to the group action setting above, correspond to the discrete log assumption and a knowledge assumption very similar to ours.

Unfortunately, when~\cite{EC:LiuMonZha23} was first published, there were no known instantiations of their scheme from isogenies. One possibility is to use the set of ordinary elliptic curves as the set, the number of points on the curve as the invariant, and orbits being sets of curves with the same number of points. Isogenies between curves are then the action\footnote{It is not a proper group action since different orbits will be acted on by different groups.}, which do not change the number of points on the curve. The problem is that in general curves, it is not possible to efficiently compute the action, since the degree of the isogenies will be too high. The action \emph{can} be computed on smooth-degree isogenies, but these are rare and there is no known way to compute a uniform superposition over curves supporting smooth-degree isogenies. For reasons we will not get into here,~\cite{EC:LiuMonZha23} propose using instead supersingular curves with non-smooth order, but again these are rare and there is no known way to generate a uniform superposition over such curves.

We resolve the issues with instantiating~\cite{EC:LiuMonZha23}, without needing the ability to compute uniform superpositions over the set. Our key insight is that, if we can compute the group action efficiently (say because we are using isogenies of smooth degree), then this is enough to sample states that \emph{are} uniform over a given orbit, except for certain phase terms: namely the states $|\G^h*x\rangle$ for uniform $h$. Then, rather than the serial number indicating which orbit we are in (which is now useless since we are in a single orbit), the serial number is a description of the phase terms, namely $h$. Note that subsequent to our work,~\cite{AC:MonSha24} successfully instantiate the walkable invariant approach using isogenies by developing new algorithms for working with isogenies.

\paragraph{Subsequent work.} In~\cite{C:MutZha25}, quantum state group actions are further explored. Most relevant to our work, they show a quantum state group action based on hash functions such that, when used to instantiate our quantum lightning scheme, one recovers exactly the quantum lightning scheme from non-collapsing hash functions explained in~\cite{DBLP:journals/joc/Zhandry21}. This generalizes our observations regarding the equivalence between lattice-based quantum money attempts discussed above.

In the case the group has a smooth order (no large prime factors), a very recent preprint~\cite{EPRINT:Doliskani25} shows that our scheme is a secure quantum money scheme, in the generic model but under the standard discrete-log assumption. Interestingly, this argument is fundamentally limited to proving quantum money, and does not extend to the stronger notion of quantum lightning. In another recent preprint,~\cite{DolMirMou25} generalize our scheme to work with the Hartley transform instead of the QFT.

In~\cite{STOC:BosNehZha25}, our scheme is generalized to work with non-abelian group actions, where a non-abelian group $\G$ (written multiplicatively) acts on a set $\Xs$ satisfying $g*(h*x)=(gh)*x$. They give a candidate instantiation based on the McEliece cryptosystem. Additionally, they consider a new concrete assumption related to the hardness of computing ``pre-actions'': that is, computing $(hg)*x)$ from $g$ and $h*x$. Note the order $hg$ makes computing pre-actions non-trivial in non-abelian groups, whereas computing $(gh)*x$ from $g$ and $h*x$ follows from the functionality of the group action. On the other hand, pre-actions and actions coincide for abelian group actions and so both are easy, meaning pre-action hardness only makes sense in the non-abelian setting. Under such a pre-action hardness assumption, they can prove the security of their scheme in the standard model.

\subsection*{Acknowledgments} We thank Hart Montgomery for many helpful discussions about isogenies. We thank Jake Doliskani for pointing out a simpler procedure for verifying banknotes and computing the serial number of banknotes.

\section{Preliminaries}\label{sec:prelim}

Here we give our notation and definitions. We assume the reader is familiar with the basics of quantum computation.

\subsection{Quantum Fourier Transform over Abelian Groups} 

Let $\G$ be an abelian group, which we will denote additively. We here define our notation for the quantum Fourier transform over $\G$. Write $\G=\Z_{n_1}\times \Z_{n_2}\times\cdots\times\Z_{n_k}$ where $\Z_{n_j}$ are the additive cyclic groups on $n_j$ elements, and associate elements $g\in\G$ with tuples $g=(g_1,\dots,g_k)$ where $g_j\in\Z_{n_j}$. Then define $\chi:\G^2\rightarrow\C$ by
\[\chi_\G(g,h)=\prod_{j=1}^k e^{i2\pi g_j h_j/n_j}\]
Observe the following:
\begin{align*}
	\chi_\G(g,h)&=\chi_\G(h,g)&
	\chi_\G(g_1+g_2,h)&=\chi_\G(g_1,h)\times\chi_\G(g_2,h)\\
	\chi_\G(-g,h)&=\chi_\G(g,h)^{-1}&
	\sum_{g\in\G}\chi_\G(g,h)&=\begin{cases}|\G|&\text{ if }h=0_\G\\0&\text{ if }h\neq 0_\G\end{cases}
\end{align*}

The quantum Fourier transform (QFT) over $\G$ is the unitary $\QFT_\G$ defined as 
\[\QFT_\G|g\rangle=\frac{1}{\sqrt{|\G|}}\sum_{h\in\G}\chi(g,h)|h\rangle\enspace.\]
Observe that $\QFT_\G=\QFT_{\Z_{n_1}}\otimes\cdots\otimes\QFT_{\Z_{n_k}}$. Therefore, since the standard QFT corresponds to $\QFT_{\Z_{n_j}}$ and can be implemented efficiently, so can $\QFT_\G$.

From this point on, we will only work with a single group, so we will drop the sub-script and simply write $\chi(g,h),\QFT$, etc.

\subsection{Public Key Quantum Money and Quantum Lightning} 

Here we define public key quantum money and quantum lightning. In the case of quantum money, we focus on public key \emph{mini-schemes}~\cite{DBLP:journals/toc/AaronsonC13}, which are essentially the setting where there is only ever a single valid banknote produced by the mint. As shown in~\cite{DBLP:journals/toc/AaronsonC13}, such mini-schemes can be upgraded generically to full quantum money schemes using digital signatures. We will generally drop the term ``public key'', since we exclusively consider this variant of quantum money. 

\paragraph{Syntax.} Both quantum money mini-schemes and quantum lightning share the same syntax:
\begin{itemize}
	\item $\gen(1^\lambda)$ is a quantum polynomial-time (QPT) algorithm that takes as input the security parameter (written in unary) which samples a classical serial number $\sigma$ and quantum banknote $\$$. 
	\item $\ver(\sigma,\$)$ takes as input the serial number and a supposed banknote, and either accepts or rejects, denoted by $1$ and $0$ respectively.
\end{itemize}

\paragraph{Correctness.} Both quantum money mini-schemes and quantum lightning have the same correctness requirement, namely that valid banknotes produced by $\gen$ are accepted by $\ver$. Concretely, there exists a negligible function $\negl(\lambda)$ such that
\[\Pr[\ver(\sigma,\$)=1:(\sigma,\$)\gets\gen(1^\lambda)]\geq 1-\negl(\lambda)\enspace.\]

\paragraph{Security.} We now discuss the security requirements, which differ between quantum money and quantum lightning.

\begin{definition}Consider a QPT adversary $\As$, which takes as input a serial number $\sigma$ and banknote $\$$, and outputs two potentially entangled states $\$_1,\$_2$, which it tries to pass off as two banknotes. $(\gen,\ver)$ is a secure \emph{quantum money mini-scheme} if, for all such $\As$, there exists a negligible $\negl(\lambda)$ such that the following holds:
	\[\Pr\left[\ver(\sigma,\$_1)=\ver(\sigma,\$_2)=1:\substack{(\sigma,\$)\gets\gen(1^\lambda)\\(\$_1,\$_2)\gets\As(\sigma,\$)}\right]\leq\negl(\lambda)\enspace .\]
\end{definition}
\begin{definition}Consider a QPT adversary $\Bs$, which takes as input the security parameter $\lambda$, and outputs a serial number $\sigma$ and two potentially entangled states $\$_1,\$_2$, which it tries to pass off as two banknotes. $(\gen,\ver)$ is a secure \emph{quantum lightning} scheme if, for all such $\Bs$, there exists a negligible $\negl(\lambda)$ such that the following holds:
	\[\Pr\left[\ver(\sigma,\$_1)=\ver(\sigma,\$_2)=1:(\sigma,\$_1,\$_2)\gets\Bs(1^\lambda)\right]\leq\negl(\lambda)\enspace .\]
\end{definition}
Quantum lightning trivially implies quantum money: any quantum money adversary $\As$ can be converted into a quantum lightning adversary $\Bs$ by having $\Bs$ run both $\gen$ and $\As$. But quantum lightning is potentially stronger, as it means that even if the serial number is chosen adversarially, it remains hard to devise two valid banknotes. This in particular means there is some security against the mint, which yields a number of additional applications, as discussed by~\cite{DBLP:journals/joc/Zhandry21}.

\begin{remark}One limitation of quantum lightning as defined above is that it cannot be secure against non-uniform attackers with quantum advice, as such attackers could have $\sigma,\$_1,\$_2$ hard-coded in their advice. The situation is analogous to the case of collision resistance, where unkeyed hash functions cannot be secure against non-uniform attackers. This limitation can be remedied by either insisting on only uniform attackers or attackers with classical advice. Alternatively, one can work in a trusted setup model, where a trusted third party generates a common reference string that is then inputted into $\gen,\ver$. A third option is to use the ``human ignorance'' approach~\cite{VIETCRYPT:Rogaway06}, in which we would formalize security proofs as explicitly transforming a quantum lightning adversary into an adversary for some other task, the latter adversary existing but is presumably unknown to human knowledge. We will largely ignore these issues throughout this work, but occasionally make brief remarks about what the various approaches would look like.
\end{remark}

\subsection{Group Actions}\label{sec:groupactions}

An (abelian) group action consists of a family of (abelian) groups $\G=(\G_\lambda)_\lambda$ (written additively), a family of sets $\Xs=(\Xs_\lambda)_\lambda$, and a binary operation $*:\G_\lambda\times \Xs_\lambda\rightarrow \Xs_\lambda$ satisfying the following properties:
\begin{itemize}
	\item {\bf Identity:} If $0\in \G_\lambda$ is the identity element, then $0*x=x$ for any $x\in \Xs_\lambda$.
	\item {\bf Compatibility:} For all $g,h\in \G_\lambda$ and $x\in \Xs_\lambda$, $(g+h)*x=g*(h*x)$.
\end{itemize}
We will additionally require the following properties:
\begin{itemize}
	\item {\bf Efficient Group Operation:} Each group $\G_\lambda$ has a polynomial-sized (in $\lambda$) description $\langle \G_\lambda\rangle$, which consists of a polynomial-sized set of generators $g_1,\cdots$ as well as a polynomial-sized (potentially quantum) circuit ${\sf Add}_\lambda$ such that ${\sf Add}_\lambda(g,h)=g+h$ for $g,h\in\G_\lambda$. 
	\item {\bf Efficient Group Action:} For each $\lambda$, there is a polynomial-sized (potentially quantum) circuit ${\sf Act}_\lambda$ such that ${\sf Act}_\lambda(g,x)=g*x$ for $g\in\G_\lambda,x\in\Xs_\lambda$.
	\item {\bf Efficiently Recognizable Set:} For each $\lambda$, there is a polynomial-sized (potentially quantum) circuit ${\sf Recog}_\lambda$ which recognizes elements in $\Xs_\lambda$. That is, for any $\lambda$ and any string $y$ (not necessarily in $\Xs_\lambda$), ${\sf Recog}_\lambda(y)$ accepts $y$ with overwhelming probability if $y\in\Xs_\lambda$, and rejects with overwhelming probability if $y\notin\Xs_\lambda$.
	\item {\bf Efficient Setup:} There is a QPT procedure ${\sf Construct}$ which, on input $1^\lambda$, outputs the description $\langle \G_\lambda\rangle$, an element $x_\lambda\in\Xs_\lambda$, the circuit ${\sf Act}_\lambda$, and the circuit ${\sf Recog}_\lambda$. We denote this collection by $\langle\G_\lambda,\Xs_\lambda\rangle$, which we call the description of the group action. We will assume ${\sf Construct}$ is \emph{pseudo-deterministic}, meaning that repeated runs of ${\sf Construct}(1^\lambda)$ will output the same values with overwhelming probability.
	\item {\bf Regular:} For every $y\in\Xs_\lambda$, there is exactly one $g\in\G_\lambda$ such that $y=g*x_\lambda$. 
\end{itemize}

\paragraph{Structure of the group $\G_\lambda$.} Given a quantum computer, we can always assume without loss of generality that $\G_\lambda$ has the form $\Z_{n_1}\times \Z_{n_2}\times\cdots\Z_{n_k}$ for some known integers $n_1,n_2,\cdots,n_k$. What we describe here is at least folklore: this is mostly worked out in~\cite{CheMos01} except for a minor detail. While we do not know of any reference for a full description, the following is well-known.

The key observation is that $\G_\lambda$ is isomorphic so such a group (since it is abelian), and moreover, the bijection with the group is efficiently computable using standard quantum period-finding techniques. The basic idea is that, for element $g_i$ of the generating set of the group $\G_\lambda$, we can compute its order, $o_i$, using Shor's period-finding algorithm~\cite{FOCS:Shor94}. Then the order of the group $\G_\lambda$ divides $o=\prod_i o_i$. We can also compute the prime-factorization of $o$ using Shor's algorithm. 

Now we use period-finding to find the periods of the function $\vv\in\Z_o^k\mapsto\sum_i v_i g_i =: \vv\cdot \gv$, where $\gv$ is the vector of generators and $v_i g_i$ means to add $g_i$ to itself $v_i$ times. The period-finding algorithm produces the (generators of the) subgroup $H$ of $\Z_o^k$ of vectors $\vv$ such that $\vv\cdot\gv=0$. Then $\G_\lambda$ is isomorphic to $\Z_o^k/H$, which can be easily decomposed as $\Z_{n_1}\times \Z_{n_2}\times\cdots\Z_{n_k}$ using linear algebra and knowledge of the factorization of $o$. 

The isomorphism itself is computed in one direction as $\vv\in \Z_o^k/H\mapsto \vv\cdot\gv$. Up until this point, this describes the algorithm of ~\cite{CheMos01}. To finish off the isomorphism, we explain how to compute it in the other direction: given $h\in\G_\lambda$, we use period-finding to find the period of the function $(\vv,w)\mapsto \vv\cdot\gv+w h$. This produces the (description of the) subgroup $H'$ of $\Z_o^{k+1}$ of $(\vv,w)$  such that $\vv\cdot\gv+w h=0$; in other words, $wh=-\vv\cdot\gv$. We then simply find a vector in this space with $w=-1$, so that $h=\vv\cdot \gv$ using linear algebra. Such a vector is guaranteed to exist since $\gv$ generates the entire group. We then finally output $\vv\bmod H$, again computed using linear algebra.

Once one has this isomorphism, one can derive a group action with $\Z_o^k/H$ acting on $\Xs_\lambda$ defined as $\vv\star x=(\vv\cdot\gv)*x$. As a consequence, throughout this work, we will typically assume that the group $\G_\lambda$ is explicitly given as $\Z_{n_1}\times \Z_{n_2}\times\cdots\Z_{n_k}$. In particular, this allows us to efficiently compute the QFT over $\G_\lambda$.

\paragraph{Cryptographic group actions.} A cryptographic group action is one for which some task is computationally intractable, and that this hardness is useful for cryptography. Here, we briefly define three standard assumptions that can be made on group actions; we will ultimately not use these assumptions, but they are useful as a point of comparison.
always
First, at a minimum, a cryptographically useful group action will always satisfy the following discrete log assumption:
\begin{assumption}\label{def:dlog} The \emph{discrete log assumption} (DLog) holds on a group action $(\G,\Xs,*)$ if, for all QPT adversaries $\As$, there exists a negligible $\lambda$ such that 
	\[\Pr[\As(\langle\G_\lambda,\Xs_\lambda\rangle,g*x_\lambda)=g:\langle\G_\lambda,\Xs_\lambda\rangle\gets{\sf Construct}(1^\lambda), g\gets\G_\lambda]\leq\negl(\lambda)\enspace .\]
\end{assumption}
\noindent Another pair of standard assumptions for group actions are the analogs of CDH and DDH:
\begin{assumption}\label{def:cdh} The \emph{computational Diffie-Hellman assumption} (CDH) holds on a group action $(\G,\Xs,*)$ if, for all QPT adversaries $\As$, there exists a negligible $\lambda$ such that 
 	\begin{align*}&\Pr\left[\As(\langle\G_\lambda,\Xs_\lambda\rangle,a*x_\lambda,b*x_\lambda)=(a+b)*x_\lambda: \langle\G_\lambda,\Xs_\lambda\rangle\gets{\sf Construct}(1^\lambda),a,b\gets\G_\lambda\right]\leq\negl(\lambda)\enspace .\end{align*}
\end{assumption}
\begin{assumption}\label{def:ddh} The \emph{decisional Diffie-Hellman assumption} (DDH) holds on a group action $(\G,\Xs,*)$ if, for all QPT adversaries $\As$, there exists a negligible $\lambda$ such that 
	\begin{align*}&\Big|\;\Pr\left[\As(\langle\G_\lambda,\Xs_\lambda\rangle,a*x_\lambda,b*x_\lambda,c*x_\lambda)=1: \substack{\langle\G_\lambda,\Xs_\lambda\rangle\gets{\sf Construct}(1^\lambda)\\a,b,c\gets\G_\lambda}\right]\\&\;\;\;\;\;-\Pr\left[\As(\langle\G_\lambda,\Xs_\lambda\rangle,a*x_\lambda,b*x_\lambda,(a+b)*x_\lambda)=1: \substack{\langle\G_\lambda,\Xs_\lambda\rangle\gets{\sf Construct}(1^\lambda)\\a,b\gets\G_\lambda}\right]\;\Big|\leq\negl(\lambda)\enspace .\end{align*}
\end{assumption}

\begin{remark}For simplicity, we model the group actions as being pseudo-deterministically computed from the security parameter. In the assumptions above, this means we can actually forgo giving $\langle\G_\lambda,\Xs_\lambda\rangle$ to $\As$ since $\As$ can compute them for itself using ${\sf Construct}$. This is the modeling we will use throughout the paper, so we will typically drop explicit mentions of ${\sf Construct}$ and also drop $\langle\G_\lambda,\Xs_\lambda\rangle$ from the inputs to $\As$. We could alternatively imagine the group actions being probabilistic, in which for each security parameter $\lambda$ there is a family of possible descriptions of groups and sets $\langle\G_\lambda,\Xs_\lambda\rangle$, and ${\sf Construct}(1^\lambda)$ samples from this family according to some distribution. In this case, we imagine ${\sf Construct}$ being a global setup procedure that is run to obtain a single instance $\langle\G_\lambda,\Xs_\lambda\rangle$, which is then supplied to all parties including the adversary via a common reference string. In this case, we must model $\langle\G_\lambda,\Xs_\lambda\rangle$ as being given to the adversary, as in the assumptions above.
\end{remark}

\section{Our Quantum Lightning Scheme}\label{sec:constr}

Here, we give our basic quantum lightning construction, which assumes a cryptographic group action.

\begin{construction}\label{constr:main} Let $\gen,\ver$ be the following QPT procedures:
	\begin{itemize}
		\item $\gen(1^\lambda)$: Initialize quantum registers $\Ss$ (for serial number) and $\Ms$ (for money) to states $|0\rangle_\Ss$ and $|0\rangle_\Ms$, respectively. Then do the following:
		\begin{itemize}
			\item Apply $\QFT_{\G_\lambda}$ to $\Ss$, yielding the joint state $\frac{1}{\sqrt{|\G_\lambda|}}\sum_{g\in\G_\lambda}|g\rangle_\Ss|0\rangle_\Ms$.
			\item Apply in superposition the map $|g\rangle_\Ss|y\rangle_\Ms\mapsto |g\rangle_\Ss|y\oplus (g*x_\lambda)\rangle_\Ms$. Here, $x_\lambda$ is an arbitrary set element. The joint state of the system $\Ss\otimes\Ms$ is then $\frac{1}{\sqrt{|\G_\lambda|}}\sum_{g\in\G_\lambda}|g\rangle_\Ss|g*x_\lambda\rangle_\Ms$.
			\item Apply $\QFT_{\G_\lambda}$ to $\Ss$ again, yielding $\frac{1}{|\G_\lambda|}\sum_{g,h\in\G_\lambda}\chi(g,h)|h\rangle_\Ss|g*x_\lambda\rangle_\Ms$
			\item Measure $\Ss$, giving the serial number $\sigma:=h$. The $\Ms$ register then collapses to the banknote $\$=|\G_\lambda^h*x_\lambda\rangle:=\frac{1}{\sqrt{|\Gs_\lambda|}}\sum_{g\in\G_\lambda}\chi(g,h)|g*x_\lambda\rangle_\Ms$. Output $(\sigma,\$)$. 
		\end{itemize}
		\item $\ver(\sigma,\$):$ First verify that the support of $\$$ is contained in $\Xs_\lambda$, by applying the assumed algorithm for recognizing $\Xs_\lambda$ in superposition. Then do the following:
		\begin{itemize}
			\item Initialize a new register $\Hs$ to $\frac{1}{\sqrt{|\G_\lambda|}}\sum_{u\in\G_\lambda}|u\rangle_\Hs$
			\item Apply in superposition the map $|u\rangle_\Hs|y\rangle_\Ms\mapsto |u\rangle_\Hs|(-u)*y\rangle_\Ms$\enspace \footnote{Note that we used the ``minimal'' oracle here for the group action computation, having $(-u)*y$ replace $y$, instead of being written to a response register as in the standard quantum oracle. However, since the computation $y\mapsto (-u)*y$ is efficiently reversible given $u$ (by $y\mapsto u*y$), we can easily implement the minimal oracle efficiently by first computing $|(-u)*y\rangle_{\Ms'}$ in a new register $\Ms'$, then uncomputing $|y\rangle_\Ms$ using the efficient inverse (so it now contains $|0\rangle_\Ms$), and finally swapping $\Ms'$ with $\Ms$.}. 
			\item Apply $\QFT^{-1}_{\G_\lambda}$ to $\Hs$.
			\item Measure $\Hs$, obtaining a group element $h'$. Accept if and only if $h'=h$.
		\end{itemize}
	\end{itemize}
\end{construction}

\begin{remark}If using a probabilistic setup of the group action, there are two options. The first is to have $\gen$ set up the group action, and have the parameters be included in the serial number. The second is to have a trusted third party set up the group action, and publish the parameters in a common reference string (CRS). If the goal is only quantum money security, then the former option is always possible, since the security experiment uses an honestly generated serial number. If the goal is quantum lightning security, the former option may not be possible, as the adversary computes the serial number; it may be that there are bad choices of parameters for the group action (and hence the CRS inside the serial number) which make it easy to forge banknotes. Therefore, for quantum lightning security, we would expect using a trusted setup to generate a CRS containing the group action parameters.
\end{remark}

\subsection{Accepting States of the Verifier}

Here we prove that honest banknote states are accepted by the verifier, and roughly that they are the \emph{only} states accepted by the verifier.

\begin{theorem}\label{thm:reject} Let $|\psi\rangle$ be a state over $\Ms$. Then $\Pr[\ver(h,|\psi\rangle)=1]=\|\langle\psi |\G_\lambda^h*x_\lambda\rangle \|^2$. Moreover, if verification accepts, the resulting state is exactly $|\G_\lambda^h*x_\lambda\rangle$.
\end{theorem}
In other words, we can treat $\ver(h,|\psi\rangle)$ as projecting exactly onto $|\G_\lambda^h*x_\lambda\rangle$. In particular, honest banknotes are accepted with probability 1. The remainder of this subsection is devoted to proving Theorem~\ref{thm:reject}.

\begin{lemma}For $h'\neq h$, $\langle \G_\lambda^{h'}*x_\lambda|\G_\lambda^h*x_\lambda\rangle=0$.
\end{lemma}
\begin{proof} 
	\begin{align*}
		\langle \G_\lambda^{h'}*x_\lambda|\G_\lambda^h*x_\lambda\rangle&=\frac{1}{|\G_\lambda|}\sum_{g,g'\in\G_\lambda}\chi(g',h')^{-1}\chi(g,h)\langle g'*x_\lambda|g*x_\lambda\rangle\\
		&=\frac{1}{|\G_\lambda|}\sum_{g\in\G_\lambda}\chi(g,h')^{-1}\chi(g,h)=\frac{1}{|\G_\lambda|}\sum_{g\in\G_\lambda}\chi(g,h-h')=0\enspace .\qedhere
	\end{align*}
\end{proof}

\begin{proof}[Proof of Theorem~\ref{thm:reject}]
Let $|\psi\rangle$ be a state with support on $\Xs$. Since the $|\G^{h'}*x_\lambda\rangle$ are orthogonal and the number of $h'$ equals the size of $\Xs$, the set $\{|\G_\lambda^{h'}*x_\lambda\rangle\}_{h'}$ forms an orthonormal basis for the set of states with support on $\Xs$. We can then write $|\psi\rangle=\sum_{h'}\alpha_{h'}|\G_\lambda^{h'}*x_\lambda\rangle$ where $\sum_{h'} \|\alpha_{h'}\|^2=1$. We then have $\|\alpha_h\|^2=\|\langle\psi |\G_\lambda^h*x_\lambda\rangle\|^2$.

We now consider applying the verification algorithm to the state $|\psi\rangle_\Ms$. After initializing $\Hs$ to $\frac{1}{\sqrt{|\G_\lambda|}}\sum_{u\in\G_\lambda}|u\rangle_\Hs$ and applying the map $(u,y)\mapsto (u,(-u)*y)$, the state of $\Hs,\Ms$ is:
\begin{align*}&\frac{1}{\sqrt{|\G_\lambda|}}\sum_{u}|u\rangle_\Hs \sum_{h'}\alpha_{h'} \frac{1}{\sqrt{|G_\lambda|}}\sum_g \chi(g,h') |(g-u)*x_\lambda\rangle_\Ms\\
	&=\frac{1}{\sqrt{|\G_\lambda|}}\sum_{u}|u\rangle_\Hs \sum_{h'}\alpha_{h'}\frac{1}{\sqrt{|G_\lambda|}}\sum_{g'} \chi(g'+u,h') |g'*x_\lambda\rangle_\Ms\\
	&=\sum_{h'}\alpha_{h'} \frac{1}{\sqrt{|\G_\lambda|}}\sum_u \chi(u,h')|u\rangle_\Hs \frac{1}{\sqrt{|G_\lambda|}} \sum_{g'}\chi(g',h')|g'*x_\lambda\rangle_\Ms\\
	&=\sum_{h'}\alpha_{h'} \frac{1}{\sqrt{|\G_\lambda|}}\sum_u \chi(u,h')|u\rangle_\Hs |\G^{h'}*x_\lambda\rangle_\Ms
\end{align*}
where above we used the substitution $g'=g-u$. Now when we apply the inverse $\QFT$ to $\Hs$, the resulting state is:
\[\sum_{h'}\alpha_{h'} |h'\rangle_\Hs |\G^{h'}*x_\lambda\rangle_\Ms\enspace .\]
We now measure $\Hs$, which produces outcome $h'$ with probability $|\alpha_{h'}|^2$, and the register $\Ms$ collapses to the state $|\G^{h'}*x_\lambda\rangle$. We have that the verifier accepts if $h'=h$, which occurs with probability $\|\alpha_h\|^2=\|\langle\psi |\G_\lambda^h*x_\lambda\rangle\|^2$ as desired. In this case, the register $\Ms$ collapses to $|\G^{h}*x_\lambda\rangle$, as desired. This completes the proof of Theorem~\ref{thm:reject}.
\end{proof}

\subsection{Computing the Serial Number}\label{sec:compserial}

Here, we show that, given a valid banknote $\$=|\G_\lambda^h*x_\lambda\rangle$ with unknown serial number $h$, it is possible to efficiently compute $h$. This result is not needed for understanding the construction or its security, but will be used in Section~\ref{sec:knowledge} to break a certain natural knowledge assumption.

\begin{theorem}\label{thm:computeh} There exists a QPT algorithm ${\sf Findh}$ such that, on input $|\G_\lambda^h*x_\lambda\rangle$, outputs $h$ with probability $1$.
\end{theorem}
\begin{proof} We originally had a much more complicated algorithm ${\sf Findh}$ (and one that had a negligible correctness error). We thank Jake Doliskani for pointing out a much simpler version using phase kickback. 
	
	Indeed, by simply modifying $\ver$ to output the measurement result $h'$ instead of testing whether or not $h'=h$, we immediately obtain such a {\sf Findh}. The proof of Theorem~\ref{thm:reject} shows that {\sf Findh} indeed outputs $h$ with probability 1 for the input state $|\G_\lambda^h*x_\lambda\rangle$.
\end{proof}

\section{A Quantum Toolkit for Generic Group Actions}\label{sec:ggam}

Here, we recall a definition of the generic group action model (GGAM), and show how to use it to give quantum security proofs.

\paragraph{A Shoup-style generic group action.} There have been several different proposals for how to define generic group actions~\cite{DBLP:journals/joc/MontgomeryZ24,PKC:DHKKLR23,EC:BonGuaZha23,EPRINT:OrsZan23}. Here, we give a definition in the style of Shoup~\cite{EC:Shoup97}. To help disambiguate between the different models, we will adapt terminology from~\cite{C:Zhandry22b} and refer to ours as the \emph{Random Set Representation} model. Roughly, in this model the group elements will be described as a standard-model group in the usual sense, and all parties can perform group operations for themselves. However, the set elements will be given as random strings, and the only way to perform the group action $*$ is using an oracle.

We first fix a (family of) groups $\G=(\G_\lambda)_\lambda$. This family is provided as a family standard model groups, meaning all algorithms have complete knowledge of the groups.  In particular, this means that all algorithms can perform the group operations for themselves and there is no oracle for group operations. Moreover, in this work we always consider abelian groups, meaning we can model each $\G_\lambda$ as $\Z_{n_1}\times \Z_{n_2}\times\cdots\Z_{n_k}$. This in particular means that $\G$ admits an efficient QFT. 

We also fix a length function $m:\Z\rightarrow\Z$ with the property that $m(\lambda)\geq\log_2|\G_\lambda|$. We call $m$ the \emph{label length}. In this model, for a given security parameter $\lambda$, a random injection $L:\G_\lambda\rightarrow\{0,1\}^{m}$ is chosen, where $m=m(\lambda)$. Think of $L(g)$ as representing $g*x_\lambda$; we call $L$ the labeling function. $\Xs_\lambda$ will then be the image of $\G_\lambda$ under $L$. All parties -- both the algorithms in the cryptosystem and also the adversary -- are then given the following:
\begin{itemize}
	\item As input, all parties receive the string $L(0)$, where $0\in\G_\lambda$ is the identity. $L(0)$ represents $x_\lambda$.
	\item All parties can then make ``group action'' queries. For classical algorithms, such a query takes the form $(\ell,g)\in\{0,1\}^m\times\G_\lambda$. The response to the query is $L(g+L^{-1}(\ell))$; if $\ell$ is not in the image of $L$, then the response to the query is $\bot$. For quantum algorithms, we follow the ``standard oracle'' convention for modeling superposition queries to classical functions, and have the query perform the map:
	\[\sum_{\ell,g,\ell'}\alpha_{\ell,g,\ell'}|\ell,g,\ell'\rangle\mapsto\sum_{\ell,g,\ell'}\alpha_{\ell,g,\ell'}|\ell,g,\ell'\oplus L(g+L^{-1}(\ell))\rangle\enspace .\]
	The set $\Xs_\lambda$ will be interpreted as the image of $\G_\lambda$ under $L$. Note that group action queries allow for testing membership in $\Xs_\lambda$: $\Xs_\lambda$ are exactly the set of strings where the group action query does not output $\bot$.
\end{itemize}
We call the oracle above $\GGAM_{\G,m}$. 

\paragraph{Cost Metrics.}In the classical setting, we usually consider queries to the oracle to have unit cost while computation outside the oracle queries is free. While this is technically an overly-conservative modeling, it tends to reflect the cost of actual known generic attacks. Moreover, lower-bounds (e.g.~\cite{EC:Shoup97}) in the classical setting exclusively work by lower-bounding the query complexity, and cannot say anything about the computational cost outside of the queries, so this model corresponds exactly to what the lower-bounds can show. If following this convention, our model essentially corresponds to the model considered in~\cite{PKC:DHKKLR23}. 

However, in the quantum setting, considering the query complexity alone is insufficient, as discrete logarithms can be solved in polynomial query complexity~\cite{EttHoy00}. In slightly more detail,~\cite{EttHoy00} show that the hidden subgroup problem has polynomial (in the bit-length of group elements) query complexity. This includes as a special case the dihedral hidden subgroup problem, which is known to be equivalent to the abelian hidden shift problem. Recall that in the abelian hidden shift problem one is given access to two functions $f_0$ and $f_1$ with the promise that $f_1(h)=f_0(g+h)$ for some secret $g$, where the domains of $f_0,f_1$ is an abelian group with group operation $+$. The goal is to find $g$. We see that the (abelian) group action discrete log problem is an example of an (abelian) hidden shift problem as follows. On discrete log instance $y=g*x$, let $f_0(h)=h*x$ and $f_1(h)=h*y=(g+h)*x=f_0(g+h)$. Applying the hidden shift solver to $f_0,f_1$ recovers the discrete log $g$ using polynomially-many queries.

While the above shows that the query complexity of discrete logs is polynomial, the computational cost of~\cite{EttHoy00} is exponential in the bit-length of group elements. As such, in contrast to the classical setting, it makes sense to consider the total cost of an algorithm as including both the queries (unit cost per query) and the computation outside the queries. This is the convention we follow in this work.

\paragraph{Minimal Oracles.} Note that we can equivalently model group action queries using the ``minimal'' oracle

\[\sum_{\ell,g}\alpha_{\ell,g}|\ell,g\rangle\mapsto\sum_{\ell,g}\alpha_{\ell,g}|L(g+L^{-1}(\ell)),g\rangle\enspace .\]

Observe that we can perform a standard oracle query using two calls to the minimal oracle, and a minimal oracle using two calls to a standard oracle. Indeed:
\begin{align*}
	\sum_{\ell,g,\ell'}\alpha_{\ell,g,\ell'}|\ell,g,\ell'\rangle&\mapsto \sum_{\ell,g,\ell'}\alpha_{\ell,g,\ell'}|L(g+L^{-1}(\ell)),g,\ell'\rangle&\text{(Minimal oracle)}\\
	&\mapsto \sum_{\ell,g,\ell'}\alpha_{\ell,g,\ell'}|L(g+L^{-1}(\ell)),g,\ell'\oplus L(g+L^{-1}(\ell))\rangle&\text{(CNOT)}\\
	&\mapsto \sum_{\ell,g,\ell'}\alpha_{\ell,g,\ell'}|L(g+L^{-1}(\ell)),-g,\ell'\oplus L(g+L^{-1}(\ell))\rangle&\text{(Group inversion)}\\
	&\mapsto \sum_{\ell,g,\ell'}\alpha_{\ell,g,\ell'}|L((-g)+g+L^{-1}(\ell)),-g,\ell'\oplus L(g+L^{-1}(\ell))\rangle&\text{(Minimal oracle)}\\
	&=\sum_{\ell,g,\ell'}\alpha_{\ell,g,\ell'}|\ell,-g,\ell'\oplus L(g+L^{-1}(\ell))\rangle&\\
	&\mapsto \sum_{\ell,g,\ell'}\alpha_{\ell,g,\ell'}|\ell,g,\ell'\oplus L(g+L^{-1}(\ell))\rangle\enspace ,&\text{(Group inversion)}
\end{align*}
\begin{align*}
	\sum_{\ell,g}\alpha_{\ell,g}|\ell,g\rangle&\mapsto \sum_{\ell,g}\alpha_{\ell,g,\ell'}|\ell,g,0\rangle&\text{(Initialize new register)}\\
	&\mapsto \sum_{\ell,g}\alpha_{\ell,g,\ell'}|\ell,g,L(g+L^{-1}(\ell))\rangle&\text{(Standard Oracle)}\\
	&\mapsto \sum_{\ell,g}\alpha_{\ell,g,\ell'}|L(g+L^{-1}(\ell)),g,\ell\rangle&\text{(Swap registers)}\\
	&\mapsto \sum_{\ell,g}\alpha_{\ell,g,\ell'}|L(g+L^{-1}(\ell)),-g,\ell\rangle&\text{(Group inversion)}\\
	&\mapsto \sum_{\ell,g}\alpha_{\ell,g,\ell'}|L(g+L^{-1}(\ell)),-g,\ell\oplus L((-g)+g+L^{-1}(\ell))\rangle&\text{(Standard Oracle)}\\
	&=\sum_{\ell,g}\alpha_{\ell,g,\ell'}|L(g+L^{-1}(\ell)),-g,\ell\oplus \ell\rangle&\\
	&=\sum_{\ell,g}\alpha_{\ell,g,\ell'}|L(g+L^{-1}(\ell)),-g,0\rangle&\\
	&\sum_{\ell,g}\alpha_{\ell,g,\ell'}|L(g+L^{-1}(\ell)),g,0\rangle&\text{(Group inversion)}\\
	&\sum_{\ell,g}\alpha_{\ell,g,\ell'}|L(g+L^{-1}(\ell)),g\rangle\enspace .&\text{(Discard register)}
\end{align*}
Therefore, we will allow either the minimal oracle or standard oracle when making queries to $\GGAM_{\G,m}$.

\paragraph{Verifying Membership in $\Xs_\lambda$.} Given a label $\ell\in\{0,1\}^m$, we can determine if $\ell\in\Xs_\lambda$, the image of $L$. To do so, pick an arbitrary element in $\Gs_\lambda$, say $0$, and query $\GGAM_{\G,m}(\ell,0)$. Output 1 if the result is $\ell$, and 0 if the result is $\bot$. Observe that if $\ell=L(g)\in\Xs_\lambda$, then the result of the query is $L(0+g)=L(g)=\ell$. On the other hand, if $\ell\notin\Xs_\lambda$, then the query gives $\bot$.

We can perform this membership test on superpositions of elements by implementing this classical procedure in superposition, and making a superposition query to the generic group action.

\paragraph{Our Quantum Lightning Construction in the GGAM.} For completeness, we explain how our quantum lightning construction (Construction~\ref{constr:main}) works in the GGAM.
\begin{itemize}
	\item $\gen(1^\lambda)$: Initialize quantum registers $\Ss$ (for serial number) and $\Ms$ (for money) to states $|0\rangle_\Ss$ and $|L(0)\rangle_\Ms$, respectively. Then do the following:
	\begin{itemize}
		\item Apply $\QFT_{\G_\lambda}$ to $\Ss$, yielding the joint state $\frac{1}{\sqrt{|\G_\lambda|}}\sum_{g\in\G_\lambda}|g\rangle_\Ss|L(0)\rangle_\Ms$.
		\item Apply in superposition the minimal oracle for $\GGAM$. The joint state of the system $\Ss\otimes\Ms$ is then $\frac{1}{\sqrt{|\G_\lambda|}}\sum_{g\in\G_\lambda}|g\rangle_\Ss|L(g)\rangle_\Ms$.
		\item Apply $\QFT_{\G_\lambda}$ to $\Ss$ again, yielding $\frac{1}{|\G_\lambda|}\sum_{g,h\in\G_\lambda}\chi(g,h)|h\rangle_\Ss|L(g)\rangle_\Ms$
		\item Measure $\Ss$, giving the serial number $\sigma:=h$. The $\Ms$ register then collapses to the banknote $\$=\frac{1}{\sqrt{|\Gs_\lambda|}}\sum_{g\in\G_\lambda}\chi(g,h)|L(g)\rangle_\Ms$. Output $(\sigma,\$)$. 
	\end{itemize}
	\item $\ver(\sigma,\$):$ First verify that the support of $\$$ is contained in $\Xs_\lambda$, by applying the membership testing procedure above in superposition. Then do the following:
	\begin{itemize}
		\item Initialize a new register $\Hs$ to $\frac{1}{\sqrt{|\G_\lambda|}}\sum_{u\in\G_\lambda}|u\rangle_\Hs$
		\item Apply in superposition the map $|u\rangle_\Hs|L(g)\rangle_\Ms\mapsto |u\rangle_\Hs|L((-u)+g)\rangle_\Ms$ using a query to the minimal oracle for $\GGAM$.
		\item Apply $\QFT^{-1}_{\G_\lambda}$ to $\Hs$.
		\item Measure $\Hs$, obtaining a group element $h'$. Accept if and only if $h'=h$.
	\end{itemize}
\end{itemize}

\paragraph{Assumptions in the GGAM.} We can likewise frame essentially any assumption or computational problem on group actions as an assumption/computational problem in the GGAM model, by similarly replacing the group action operation $*$ with queries to $\GGAM$ and the set elements $g*x_\lambda$ with the labels $L(g)$. For example, the discrete logarithm assumption (Assumption~\ref{def:dlog}) translates to the following:

\begin{assumption}\label{def:dlogGGAM} The \emph{discrete log assumption} (DLog) holds in $\GGAM_{\G,m}$ if, for all QPT adversaries $\As$ making polynomially-many queries, there exists a negligible $\lambda$ such that
	\[\Pr[\As^{\GGAM_{\G,m}}(L(g))=g]\leq\negl(\lambda)\enspace .\]
	Where the probability is over the choice of $g\gets\G_\lambda$ and the randomness of the labeling function $L$.
\end{assumption}
It is straightforward to adapt other standard-model assumptions to the GGAM. In our proofs below, we will consider both the standard-model assumptions such as Assumption~\ref{def:dlog} and assumptions in the GGAM such as Assumption~\ref{def:dlogGGAM}. Note that GGAM assumptions are independent of any particular group action, whereas standard-model assumptions are always with respect to a specific group action $(\G,\Xs,*)$.

\subsection{On other Styles of Generic Group Actions}

Other styles of generic group action are possible. For example,~\cite{DBLP:journals/joc/MontgomeryZ24} consider a similar model except where the group $\G$ itself is also hidden behind an oracle, meaning that group elements are random labels and an additional oracle is provided for performing group operations. We might call this the \emph{Random Group, Random Set Representation} model. Based on the discussion in Section~\ref{sec:groupactions}, in the case of abelian groups $\G$ this alternative model is equivalent to the Random Set Representation model defined above.

It is also possible to consider a version that is akin to Maurer's~\cite{IMA:Maurer05} generic group model, where instead of random labels for every element one only receives handles. This is the kind of model considered in~\cite{EC:BonGuaZha23,EPRINT:OrsZan23}. Following the terminology of~\cite{C:Zhandry22b}, this can be called the Type Safe model. We note that it does not make much sense to consider the group as an idealized object while allowing complete access to the set. Indeed, the discussion in Section~\ref{sec:groupactions} shows that any such ``Random Group Representation'' effectively is just a standard-model group action, defeating the purpose of considering an idealized model.

Here, we discuss why these alternate models come with limitations. First we observe that hiding the group behind an oracle puts more idealized constraints on the adversary. In the case of abelian groups this makes ends up making no difference, but in the case of non-abelian groups results in a model that is potentially less reflective of the real world.

Worse is the case of Type Safe models. In the classical generic group setting, as first proved in~\cite{PROVSEC:JagSch08} and clarified in~\cite{C:Zhandry22b}, when it comes to proving security, the Type Safe and Random Representation models can usually be treated as equivalent\footnote{This is not the case when using the models to prove \emph{impossibility} results, where even classically there is a major difference between the two models.}. This equivalence would carry over to the classical setting for generic group actions. However, we observe that the equivalence proved in~\cite{PROVSEC:JagSch08,C:Zhandry22b} does \emph{not} hold in the quantum setting. This observation was first made, but not elaborated on, by~\cite{DBLP:conf/crypto/HhanYY24}.

In more detail, one direction of the equivalence --- converting an adversary in the Type Safe model into one in the Random Representation model --- is trivial, both classically and quantumly. We just use the random labels from the random representation model as the handles for the Type Safe adversary. For the other direction, the classical proof will construct a Type Safe adversary out of a Random Representation adversary by choosing the random labels itself. The challenge is that the Random Representation adversary will expect identical labels on certain related queries, namely if it computes the same element $g*x_\lambda$ in multiple ways. To account for this, the Type Safe adversary maintains a table of all the queries made so far, and the labels generated for those queries. Then if it ever needs to output an element that was already produced, it can use the table to make sure it uses the same label.

In the quantum setting, maintaining this table is problematic, as it requires recording the queries made by the adversary. Quantum queries cannot be recorded without perturbing them, and if the adversary detects any disturbance it may abort and refuse to work. Such an adversary would break the classical reduction. We note that sometimes it is possible to record quantum queries~\cite{C:Zhandry19}, but the recording has to be done in careful ways that limit applications. In particular, such query recording is usually done on random oracles, and there has so far been no techniques for recording queries for complicated structures like group action oracles.

Thus, based on our current understanding, the Random Set Representation model defined above seems to be ``at least as good'' as any other model for group actions in the quantum setting, and may in fact be ``better'' than the other models. For this reason, we focus on the Random Set Representation model. We leave exploring the exact relationship between the models as an interesting open question.

\paragraph{Algebraic Group Action Model.} In Section~\ref{sec:knowledge}, we consider a different idealized model called the Algebraic Group Action Model, the quantum and group action version of the classical Algebraic Group Model (AGM)~\cite{C:FucKilLos18}. In the classical world, this model is ``between'' the Type Safe model and the standard model, in the sense that security in the algebraic model implies security in the Type Safe model (which in turn often implies security in the Random Representation model, per~\cite{C:Zhandry22b}). However, in Section~\ref{sec:knowledge}, we explain that the quantum analog of this model is actually problematic, and the proof of ``between-ness'' does not hold quantumly, for similar reasons as to why the equivalence between Random Representation and Type Safe models does not appear to hold quantumly. As such, it seems that the (Random Representation) generic group action model actually \emph{better} captures available attacks than the algebraic group action model.

\subsection{Our Framework for Quantum GGAM Security Proofs.}

\paragraph{Challenges with the quantum GGAM.} The challenge with the quantum GGAM, as observed by~\cite{PKC:DHKKLR23}, is that we cannot hope for unconditional security results, as the discrete logarithm is easy if we only count quantum query complexity.~\cite{PKC:DHKKLR23} take the approach of instead considering the Algebraic Group Action Model (AGAM). We discuss the pitfalls of this approach in Section~\ref{sec:knowledge}. Here we instead observe that we can recover a meaningful model by counting both queries and computational cost. However, because we cannot hope to prove unconditional query complexity lower bounds, we must instead resort to making computational assumptions and giving reduction-style arguments. This means arguments in the quantum GGAM will look very different that proofs in the classical GGM. To the best of our knowledge, there have been no prior security proofs in the quantum GGAM. We therefore develop some new tools and techniques for giving such proofs, including a proof of security of our quantum money scheme.

\paragraph{Our Abstract Framework.} We first give a very abstract framework, which we will then apply the framework to the GGAM.

Let $\Ys$ be a set, and $\Fs$ be a family of functions $f:\Ys\rightarrow\Ys$. Let $y_0\in\Ys$ be a specific starting element in $\Ys$. Consider a random injection $L:\Ys\rightarrow\{0,1\}^{m'}$, and consider the oracle $\Os$ which maps $\Os(L(y),f)=L(f(y))$; $\Os$ outputs $\bot$ on any string that is not in the image of $L$. We will give the adversary $L(y_0)$ and also superposition access to $\Os$.

Now consider a set $\Ys'\subset\{0,1\}^s$, and suppose we have a not-necessarily-random injection $\Gamma:\Ys\rightarrow\Ys'$ (meaning $s\geq |\Ys|$). We also have a procedure $P$ which is able to map $P(\Gamma(y),f)=\Gamma(f(y))$. However, unlike the oracle $\Os$ considered above, this procedure $P$ may output values other than $\bot$ when given inputs that are not in the image of $\Gamma$. Our goal is to, nevertheless, simulate $\Os$ using $P$.

Concretely, we will choose a random injection $\Pi:\{0,1\}^s\rightarrow\{0,1\}^{m'}$, and simulate $\Os$ with the oracle $\Os'(\Pi(z),f)=\Pi(P(z,f))$; $\Os'$ will output $\bot$ on any input not in the image of $\Pi$. We will then give the adversary $\Pi(\Gamma(y_0))$, and quantum query access to $\Os'$.

\paragraph{Application to the GGAM.} In our case, we will have $\Ys$ be a group $\G_\lambda$. $\Fs$ will include for each $h\in\G_\lambda$ the map $g\mapsto h+g$. The distinguished element $y_0$ is just $0\in\G_\lambda$. In this way, $\Os$ becomes the generic group action oracle, with labeling function $L$. However, we also include extra operations in $\Fs$, the exact operations will depend on the application.

Our goal will be to simulate $\Os$, the generic group action oracle with extra operations, using only a plain group action $(\G,\Xs,*)$. $(\G,\Xs,*)$ could be a standard-model group action, or perhaps a plain generic group action. We will assume $\Xs_\lambda\subseteq\{0,1\}^{m}$ for some polynomial $m=m(\lambda)$. This ``base'' group action will be the source of hardness. We will therefore make some hopefully simple and mild computational assumptions about $(\G,\Xs,*)$, and hope to derive useful hardness results about the expanded group action $\Os$.

To do so, we will let $\Ys'=\Xs_\lambda^{\otimes k}$ for some $k$. We will also choose some integers $c_1,\dots,c_k$ whose GCD is 1, and starting set elements $y_1,\dots,y_k$. Then define $\Gamma(g)=((c_1g)*y_1,(c_2g)*y_2,\cdots,(c_kg)*y_k)$. Since the GCD of the $c_i$ is 1, the map $\Gamma(g)$ is injective. 

For $f$ corresponding to adding group element $h$, we can set $P((z_1,\dots,z_k),h)=((c_1h)*z_1,\cdots,(c_kh)*z_k)$. Note that this will have the correct effect, as $P(\Gamma(g),h)=\Gamma(h+g)$. For simulating other functions $f\in\Fs$, we will rely on other transformations to the vector $(z_1,\dots,z_k)$, which will depend on the application.

\paragraph{Correctness of the Simulation.}

\begin{lemma}\label{lem:ggamsim} Fix $y_0,\Ys,\Ys',\Gamma,\Fs$ as above. Assume $m'\geq s+t$ for some $t$. Then consider any quantum algorithm $\As$ which makes $q$ quantum queries to its oracle. Then:
	\[\left|\Pr\left[\As^{\Os}(L(y_0))=1\right]-\Pr\left[\As^{\Os'}(\Pi(\Gamma(y_0)))=1\right]\right|<O(q\times 2^{-t/2})\enspace .\]
	Above, $L,\Pi$ are random injections, with $\Os,\Os'$ being derived from them as above. The probabilities are over the random choice of $L,\Pi$ and the randomness of $\As$. Note that our order of quantifiers allows $\As$ to depend on $y_0,\Ys,\Ys',\Gamma,\Fs$.
\end{lemma}
\begin{proof}We prove security via a sequence of hybrids.
	\paragraph{Hybrid 0.} This is the case where we run $\As^{\Os}(L(y_0))$ where $L:\Ys\rightarrow\{0,1\}^{m'}$ is uniform random injection. Let $p_0$ be the probability of outputting 1.
	\paragraph{Hybrid 1.} Here, we run $\As^{\Os}(L(y_0))$, except that we set $L$ to be the function $L(y)=\Pi(\Gamma(y))$, where $\Pi$ is a random injection. But since $\Gamma$ is an injection, this means $L$ is a random injection anyway, so the distribution of $L$ and hence $\Os$ is identical to {\bf Hybrid 0}. Therefore, if we let $p_1$ be the probability $p_0$ outputs 1 in {\bf Hybrid 1}, we have $p_1=p_0$. Observe that $L(y_0)=\Pi(\Gamma(y_0))$.

	\paragraph{Hybrid 2.} Here, we run $\As^{\Os'}(\Pi(\Gamma(y_0)))$. Let $p_2$ be the probability of outputting 1. On all points that $\Os$ accepts, $\Os'$ behaves identically. Likewise, on any point that $\Os'$ rejects, $\Os'$ rejects as well. The only difference between this and {\bf Hybrid 1} is that here, $\Os'$ may accept elements that were rejected by $\Os$, namely elements that are in the image of $\Pi$ but not in the image of $L=\Pi\circ\Gamma$. We will show that these potential changes are nevertheless undetectable except with small probability.
	
	Consider running $\As^{\Os}(L(y_0))$ where $L(y)=\Pi(\Gamma(y))$ as in {\bf Hybrid 1}. However, we only sample $\Pi$ on inputs $z$ that are in the image of $\Gamma$; for all other inputs $z$, $\Pi$ remains unspecified. Observe that {\bf Hybrid 1} never needs to evaluate $\Pi$ on $z$ outside of the image of $\Gamma$, since the oracle $\Os$ will anyway reject in these cases. Let $S\subseteq\{0,1\}^{m'}$ be the set of images of $\Pi$ sampled so far.
	
	Now imagine simulating the rest of $\Pi$. Let $T\subset\{0,1\}^{m'}$ be the set of images of $\Pi$ for $z\in \Ys'$ that are not in the image of $\Gamma$. Observe that $T$ is a random subset of size $|\Ys'|\setminus |\Ys|\leq |\Ys'|\leq 2^s$. We now observe that the only points where $\Os$ and $\Os'$ differ are on pairs $(\ell,f)$ for $\ell\in T$: for $\ell\in S$, the two faithfully compute the same function and are identical, while for $\ell\notin T\cup S$, both output $\bot$.

	From here, concluding that $p_1$ and $p_2$ are close is a standard argument. The expected total query weight in {\bf Hybrid 1} on points $(\ell,f)$ for $\ell\in T$ is at most $|T|/2^{m'}\leq 2^{-t}$. Then via standard results in quantum query complexity~\cite{BBBV97}, the difference in acceptance probabilities $|p_1-p_2|$ is at most $O(\sqrt{q^2 2^{-t}})=O(q\times 2^{-t/2})$. Thus $|p_0-p_2|\leq O(q\times 2^{-t/2})$, as desired.\end{proof}

Next, we recall a lemma that shows that random injections can be simulated quantumly:

\begin{lemma}[\cite{AC:Zhandry21}]\label{lem:injectionsim} Random injections with quantum query access can be simulated efficiently.
\end{lemma}

With Lemmas~\ref{lem:ggamsim} and~\ref{lem:injectionsim} in hand, we now turn to security proofs in the GGAM.

\subsection{Group Actions with Twists}

In group actions based on isogenies, it is possible to compute a ``twist'', which maps $g*x_\lambda\mapsto (-g)*x_\lambda$. It is straightforward to update our notion of group action and generic group action to incorporate twists. Let $\GGAM^\pm_{\G,m}$ denote the generic group action relative to group $\G$ with label length $m$. Such twists effectively allow for the dihedral group to act on the set $\Xs_\lambda$. An important question is whether having this larger (non-abelian) group act on $\Xs_\lambda$ can be damaging for security. Here, we show that, at least generically, the existence of twists plausibly has little impact on security.

\paragraph{Assumptions with Negation.} We consider variants of standard assumptions on group actions where additional ``negation'' elements are given out. For example:

\begin{assumption}\label{def:dlogpm} The \emph{discrete log assumption with negation} (DLog$^\pm$) holds on a group action $(\G,\Xs,*)$ if, for all QPT adversaries $\As$, there exists a negligible $\lambda$ such that 
	\[\Pr[\As(g*x_\lambda,(-g)*x_\lambda)=g:g\gets\G_\lambda]\leq\negl(\lambda)\enspace .\]
\end{assumption}
\begin{assumption}\label{def:cdhpm} The \emph{computational Diffie-Hellman assumption with negation} (CDH$^\pm$) holds on a group action $(\G,\Xs,*)$ if, for all QPT adversaries $\As$, there exists a negligible $\lambda$ such that 
	\[\Pr\left[\As\left(\substack{a*x_\lambda,b*x_\lambda,\\(-a)*x_\lambda,(-b)*x_\lambda}\right)=(a+b)*x_\lambda: a,b\gets\Gs_\lambda\right]\leq\negl(\lambda)\enspace .\]
\end{assumption}
\begin{assumption}\label{def:ddhpm} The \emph{decisional Diffie-Hellman assumption with negation} (DDH$^\pm$) holds on a group action $(\G,\Xs,*)$ if, for all QPT adversaries $\As$, there exists a negligible $\lambda$ such that 
	\begin{align*}&\left|\Pr\left[\As\left(\substack{a*x_\lambda,b*x_\lambda,c*x_\lambda,\\(-a)*x_\lambda,(-b)*x_\lambda,(-c)*x_\lambda}\right)=1: a,b,c\gets\Gs_\lambda\right]\right.\\&\left.-\Pr\left[\As\left(\substack{a*x_\lambda,b*x_\lambda,(a+b)*x_\lambda,\\(-a)*x_\lambda,*(-b)*x_\lambda,(-a-b)*x_\lambda}\right)=1:a,b\gets\Gs_\lambda\right]\right|\leq\negl(\lambda)\enspace .\end{align*}
\end{assumption}

Note that the $^\pm$ versions of DLog, CDH, DDH imply their ordinary counterparts (Definitions~\ref{def:dlog},~\ref{def:cdh}, and~\ref{def:ddh}, respectively). Moreover, the assumptions are \emph{equivalent} to the ordinary versions on group actions with twists. Also, note that, while~\cite{DBLP:journals/joc/MontgomeryZ24} prove the quantum equivalence of ordinary DLog and CDH, their proof does not necessarily apply to the $^\pm$ versions, and an equivalence between these versions may be incomparable since it would start from a stronger property, but also reach a stronger conclusion.

\paragraph{Our Result.} We now show that, the negation assumptions allow us to lift to security under twists, generically.

\begin{theorem}Let $(\G,\Xs,*)$ be a group with $\Xs\subseteq\{0,1\}^m$ such that DLog$^\pm$ (resp. CDH$^\pm$, DDH$^\pm$) holds. Let $m'\geq 2m+\omega(\log\lambda)$. Then DLog$^\pm$ (resp. CDH$^\pm$, DDH$^\pm$) hold in $\GGAM^{\pm}_{\G,m'}$, the GGAM with twists relative to group $\G$ and with label length $m'$.\end{theorem}
\begin{proof}We prove the case of DDH, the other proofs being nearly identical. Let $\As^{\GGAM^\pm_{\G,m}}$ be a supposed adversary for DDH$^\pm$ in $\GGAM^\pm_{\G,m}$, the GGAM with twists and with label length $m'$. Let $\epsilon$ be the distinguishing advantage of $\As$, and $q$ the polynomial number of queries. We construct a new adversary $\Bs$ for DDH$^\pm$ in the group action $(\G,\Xs,*)$ as follows.
	\begin{itemize}
		\item $\Bs$, on input $u^+,v^+,w^+,u^-,v^-,w^-$, will choose a random injective function $\Pi$ from $\{0,1\}^{2m}\rightarrow\{0,1\}^{m'}$. To make $\Bs$ efficient, we will actually use Lemma~\ref{lem:injectionsim} to efficiently simulate $\Pi$. For simplicity in the following proof, we will treat $\Bs$ as actually using a true random injection.
		\item $\Bs$ will compute $X=\Pi(x_\lambda,x_\lambda),U=\Pi(u^+,u^-),V=\Pi(v^+,v^-),W=\Pi(w^+,w^-)$.
		\item $\Bs$ will then run $\As(X,U,V,W)$\enspace\footnote{Recall that in the definition of DDH, the adversary is only given $U,V,W$. However, in the generic group action model, we additionally give all parties the starting point $X$.}, simulating its queries as follows:
		\begin{itemize}
			\item For queries to the group action $(\ell,g)$, $\Bs$ simulates the query by computing $(z_1,z_2)\gets\Pi^{-1}(\ell)$, and then returning $\Pi(g*z_1,(-g)*z_2)$. For superposition queries, $\Bs$ simply runs this computation in superposition. Note that if we let $\Gamma(g)=(g*x_\lambda,(-g)*x_\lambda)$, then $\Bs$ simulates these queries exactly as prescribed  above in our general framework, for constants $c_1=1,c_2=-1$ and $y_1=y_2=x_\lambda$.
			
			\item When $\As$ makes a twist query on label $\ell$, $\Bs$ computes $(z_1,z_2)\gets\Pi^{-1}(\ell)$, and then computes $\ell'=\Pi(z_2,z_1)$ and responds with $\ell'$. For superposition queries, $\Bs$ simply runs this computation in superposition. Observe that the twist of $\Pi(\Gamma(g))$ as computed by $\Bs$ is exactly $\Pi(\Gamma(-g))$.
		\end{itemize}
		
		\item $\Bs$ then outputs whatever $\As$ outputs.
	\end{itemize}
	We now prove security via a sequence of hybrids.
	\paragraph{Hybrid 0.} Here, we run $\As^{\Os}(X,U=a*X,V=b*X,W=c*X)$ for a random injection $L$, where $X=L(0)$, $a,b,c$ are uniform in $\G_\lambda$, and $*$ denotes the action defined by $\Os$. Let $p_0$ be the probability $\As$ outputs 1.
	\paragraph{Hybrid 1.} Here, $\Bs$ is given $u^+,v^+,w^+,u^-,v^-,w^-=a*y,b*y,c*y,(-a)^y,(-b)*y,(-c)*y$, and simulates $\As$ as described above. Let $p_1$ be the probability $\As$ (and hence $\Bs$) outputs 1. Observe that $X,U,V,W=L(0),L(a),L(b),L(c)$, where $L$ is the implicit labeling function $L(g)=\Pi(g*x_\lambda,(-g)*x_\lambda)$. Since $\Bs$ simulates twist queries by mapping $L(g)\mapsto L(-g)$, $\Bs$ correctly simulates the view of $\As$ in {\bf Hybrid 0}, except that $\Os'$ and the twist oracle operate on values $\Pi(z_1,z_2)$ that might not be in the image of $L$. But we can invoke Lemma~\ref{lem:ggamsim} to conclude that $|p_0-p_1|\leq O(q\times 2^{2m-m'})=q\times\negl(\lambda)=\negl(\lambda)$.
	\paragraph{Hybrid 2.} Here, $\Bs$ is $u^+,v^+,w^+,u^-,v^-,w^-=a*y,b*y,(a+b)*y,(-a)^y,(-b)*y,(-a-b)*y$, and simulates $\As$ as described above. Let $p_2$ be the probability $\As$ (and hence $\Bs$) outputs 1. By Assumption~\ref{def:ddhpm}, $|p_1-p_2|\leq\negl(\lambda)$.
	\paragraph{Hybrid 3.} Now we run $\As^{\Os}(X,U=a*X,V=b*X,W=(a+b)*X)$. Let $p_3$ be the probability $\As$ outputs 1. By a similar argument for going from {\bf Hybrid 0} to {\bf Hybrid 1}, we conclude that $|p_2-p_3|\leq\negl(\lambda)$ is negligible. Piecing everything together, we have that $\epsilon=|p_0-p_3|\leq\negl(\lambda)$, thereby proving DDH$^\pm$ holds in $\GGAM^\pm_{\G,m'}$.\end{proof}

\subsection{Computing Banknotes With Complementary Serial Numbers}

Here, we prove that it is hard in generic group action to compute two banknotes for our scheme with ``complementary'' serial numbers that sum to zero. 

\begin{theorem}\label{thm:almostlightning}Let $(\G,\Xs,*)$ be a group with $\Xs\subseteq\{0,1\}^m$ such that DDH holds (Assumption~\ref{def:ddh}). Let $m'\geq 4m+\omega(\log\lambda)$. Let $(\gen^{\GGAM_{\G,m'}},\ver^{\GGAM_{\G,m'}})$ be the quantum money construction from Construction~\ref{constr:main}, using the generic group action $\GGAM_{\G,m'}$. Consider a QPT adversary $\Bs^{\GGAM_{\G,m'}}$ making queries to $\GGAM_{\G,m'}$, which takes as input the security parameter $\lambda$, and outputs a serial number $h\in\G_\lambda$ and two potentially entangled states $\$_1,\$_2$, which it tries to pass off as two banknotes. For all such $\Bs$, there exists a negligible $\negl(\lambda)$ such that the following holds:
	\[\Pr\left[\ver^{\GGAM_{\G,m'}}(h,\$_1)=\ver^{\GGAM_{\G,m'}}(-h,\$_2)=1:(h,\$_1,\$_2)\gets\Bs^{\GGAM_{\G,m'}}(1^\lambda)\right]\leq\negl(\lambda)\enspace .\]
\end{theorem}
Notice that the statement above is \emph{almost} the statement that $(\gen^{\GGAM_{\G,m'}},\ver^{\GGAM_{\G,m'}})$ is a quantum lightning scheme, except that the second banknote is verified with respect to $-h$ instead of $h$. Theorem~\ref{thm:almostlightning} is therefore not quite enough to prove the security of our scheme, since it could be the case that it is possible to output many banknotes with the same serial number, even if it is impossible to output two with complementary numbers. We give a different proof below in Section~\ref{sec:qlightningsec} based on a stronger assumption which proves our scheme quantum lightning. We use the result here as a warm-up to our later result, which is based on a more complex assumption. Moreover, Theorem~\ref{thm:almostlightning} lets us prove that it is generically hard to output the uniform superposition $\frac{1}{\sqrt{|\G_\lambda|}}\sum_{g\in\G_\lambda}|L(g)\rangle$, which is just the banknote $|\G^0_\lambda*L(0)\rangle$ with serial number $0$. We state and prove this fact before proving Theorem~\ref{thm:almostlightning}.
\begin{corollary}\label{cor:nozero}Let $(\G,\Xs,*)$ be a group with $\Xs\subseteq\{0,1\}^m$ such that DDH holds (Assumption~\ref{def:ddh}). Let $m'\geq 4m+\omega(\log\lambda)$. Let $L$ be the labeling function for the generic group action $\GGAM_{\G,m'}$. Then for any QPT adversary $\As^{\GGAM_{\G,m'}}$ making queries to $\GGAM_{\G,m'}$ which outputs a state $\rho$, there exists a negligible $\negl(\lambda)$ such that $\langle\G^0_\lambda*L(0)|\rho|\G^0_\lambda*L(0)\rangle\leq\negl(\lambda)$.
\end{corollary}
\begin{proof} Consider an adversary $\As^{\GGAM_{\G,m'}}$ outputting a mixed state $\rho$ and let $\epsilon=\langle\G^0_\lambda*L(0)|\rho|\G^0_\lambda*L(0)\rangle\leq\negl(\lambda)$. Recall that our verifier from Section~\ref{sec:constr} can project exactly onto the state $|\G^0_\lambda*L(0)\rangle$. By applying this projection to $\rho$, we have that $\As^{\GGAM_{\G,m'}}$ outputs $|\G^0_\lambda*L(0)\rangle$ with probability $\epsilon$. We will therefore assume we have the state $|\G^0_\lambda*L(0)\rangle$.
	
	Apply in superposition the map $|x\rangle\mapsto|x,x\rangle$. Now we have the state 
	\[\frac{1}{\sqrt{|\G_\lambda|}}\sum_{g\in\G_\lambda}|L(g),L(g)\rangle\enspace .\]
	We can equivalently write this state as:
	\[\frac{1}{\sqrt{|\G_\lambda|}}\sum_{h\in\G_\lambda}|\G^h_\lambda*L(0)\rangle|\G^{-h}_\lambda*L(0)\rangle\enspace .\]
	We therefore apply our algorithm ${\sf Findh}$ from Theorem~\ref{thm:computeh} to the first register. The output will be a random serial number $h$, and the state will collapse to $|\G^h_\lambda*L(0)\rangle|\G^{-h}_\lambda*L(0)\rangle$. We output this, which solves the problem in Theorem~\ref{thm:almostlightning}. Thus, we conclude that $\epsilon$ must be negligible.
\end{proof}

\noindent We now turn to proving Theorem~\ref{thm:almostlightning}.
\begin{proof}[Proof of Theorem~\ref{thm:almostlightning}]Consider an adversary $\Bs^{\GGAM_{\G,m'}}$, and define: \[\epsilon:=\Pr\left[\ver^{\GGAM_{\G,m'}}(h,\$_1)=\ver^{\GGAM_{\G,m'}}(-h,\$_2)=1:(h,\$_1,\$_2)\gets\Bs^{\GGAM_{\G,m'}}(1^\lambda)\right]\enspace .\]
	Recall that $\ver^{\GGAM_{\G,m'}}(h,\$)$ projects onto the correct banknote $|\G^h_\lambda*L(0)\rangle$. Therefore, with probability $\epsilon$, $\Bs$ outputs $h$ and exactly the states $|\G^h_\lambda*L(0)\rangle,|\G^{-h}_\lambda*L(0)\rangle$.
	
	We now construct an adversary $\As$ for DDH on the group action $(\G,\Xs,*)$. $\As$, on input $(u,v,w)$, will choose a random injection $\Pi:\{0,1\}^{4m}\rightarrow\{0,1\}^{m'}$. It will then compute $X=\Pi(x_\lambda,u,v,w)$. $\As$ will then run $\Bs(X)$, simulating its queries $(\ell,g)$ to the group action as follows: compute $(z_1,z_2,z_3,z_4)\gets\Pi^{-1}(\ell)$, and then return $\Pi(g*z_1,(-g)*z_2,g*z_3,(-g)*z_4)$. For superposition queries, $\As$ simply runs this computation in superposition. Note that if we let $\Gamma(g)=(g*x_\lambda,(-g)*u,g*v,(-g)*w)$, then $\As$ simulates these queries exactly as prescribed above in our general framework, for constants $c_1=1,c_2=-1,c_3=1,c_4=-1$ and $(y_1,y_2,y_3,y_4)=(x_\lambda,u,v,w)$.
	
	Finally, when $\Bs$ produces serial number $h$ and banknotes $\$_1,\$_2$, $\As$ does the following:
	\begin{itemize}
		\item Run $\ver^{\Os'}(h,\$_1)$ and $\ver^{\Os'}(-h,\$_2)$, answering the queries of $\ver$ using the simulated group action oracle. If either run rejects, output a random bit. Otherwise, let $\$_1',\$_2'$ be the resulting states of the verifier.
		\item In superposition, it applies the following map $\ell\mapsto\ell'$ to $\$_2'$:
		\begin{itemize}
			\item First map $\ell\mapsto\Pi^{-1}(\ell)=(z_1,z_2,z_3,z_4)$
			\item Now map $(z_1,z_2,z_3,z_4)\mapsto\ell'=\Pi(z_2,z_1,z_4,z_3)$. Note that the $z_i$ inside $\Pi$ have been permuted.
		\end{itemize}
		Let $\$_2''$ be the result of this map.
		\item Apply the swap test to $\$_1',\$_2''$, outputting whatever the swap test outputs.
	\end{itemize}
	
	By applying Lemma~\ref{lem:ggamsim}, we can conclude that $\$_1,\$_2$ are actually superpositions over elements of the form $L(g)=\Pi(g*z_1,(-g)*z_2,g*z_3,(-g)*z_4)$ for varying $g$. Then using our characterization of the accepting states of $\ver$, we see that both runs of $\ver$ simultaneously accept with probability $\epsilon$, and in this case $\$_1'=|\G_\lambda^h*L(0)\rangle,\$_2'=|\G_\lambda^{-h}*L(0)\rangle$.
	
	We must analyze the effect of the map $\ell\mapsto\ell'$ on $|\G_\lambda^{-h}*L(0)\rangle$. We break into two cases:
	\begin{itemize}
		\item $u=a*x_\lambda,v=b*x_\lambda,w=(a+b)*x_\lambda$. Let $\ell=L(g)=\Pi(g*z_1,(-g)*z_2,g*z_3,(-g)*z_4)=\Pi(g*x_\lambda,(a-g)*x_\lambda,(b+g)*x_\lambda,(a+b-g)*x_\lambda)$, which maps to $\ell'=\Pi((a-g)*x_\lambda,g*x_\lambda,(a+b-g)*x_\lambda,(b+g)*x_\lambda)=L(a-g)$.
		
		Therefore, $|\G_\lambda^{-h}*L(0)\rangle$ maps to
		\begin{align*}|\G_\lambda^{-h}*L(0)\rangle&=\frac{1}{\sqrt{|\G_\lambda|}}\sum_g \chi(g,-h)|L(g)\rangle\\&\mapsto\frac{1}{\sqrt{|\G_\lambda|}}\sum_g \chi(g,-h)|L(a-g)\rangle\\&=\frac{1}{\sqrt{|\G_\lambda|}}\sum_{g'} \chi(a-g',-h)|L(g')\rangle\\&=\chi(a,-h)\frac{1}{\sqrt{|\G_\lambda|}}\sum_{g'} \chi(g',h)|L(g')\rangle\\&=\chi(a,-h)|\G_\lambda^{h}*L(0)\rangle\enspace .\end{align*}
		
		Above, we used the substitution $g'=a-g$. Thus, in this case, $\As$ obtains two copies of $|\G_\lambda^{h}*L(0)\rangle$, which the swap test will accept with probability 1. Therefore, the probability $\As$ outputs 1 is $\frac{1}{2}(1-\epsilon)+\epsilon=\frac{1+\epsilon}{2}$.
		\item $u=a*x_\lambda,v=b*x_\lambda,w=c*x_\lambda$ with $c\neq a+b$. In this case, $\ell=L(g)=\Pi(g*x_\lambda,(a-g)*x_\lambda,(b+g)*x_\lambda,(c-g)*x_\lambda)$ maps to $\ell'=\Pi((a-g)*x_\lambda,g*x_\lambda,(c-g)*x_\lambda,(b+g)*x_\lambda)$. However, $\ell'$ is \emph{not} equal to $L(g')$ for any $g'$. Indeed, in order for $\ell'=L(g')$, we get several equations:
		\[g'=a-g\enspace ,\;\;\;\;a-g'=g\enspace ,\;\;\;\;b+g'=c-g'\enspace ,\;\;\;\;c-g'=b+g\enspace .\]
		The first two equations require that $g'=a-g$, while the last two require that $g'=c-b-g\neq a-g$. Hence, the state $\$_2''$ has disjoint support from the state $|\G_\lambda^{h}*L(0)\rangle$, and hence is orthogonal to it. Therefore, the swap test will accept with probability exactly 1/2. The overall probability $\As$ outputs 1 is therefore exactly $1/2$.
	\end{itemize}
	Thus, we see that $\As$ has advantage $\epsilon/2$ in distinguishing DDH, breaking the assumption.\end{proof}

\subsection{Security of our Quantum Lightning Scheme}\label{sec:qlightningsec}

Here, we prove the generic security of our quantum lightning scheme (Construction~\ref{constr:main}). We do not know how to prove security under any standard group action-based assumption. We instead introduce a novel assumption that appears plausible, but needs extra cryptanalysis to be certain.

\paragraph{The Decisional 2x Assumption (D2X).} A classical ``Diffie-Hellman Exponent'' assumption is to distinguish $g^a,g^{a^2}$ from $g^a,g^b$ for uniform $a,b$. The group action equivalent would be to distinguish $a*x_\lambda,(2a)*x_\lambda$ from $a*x_\lambda,b*x_\lambda$ for uniform $a,b\in\G_\lambda$. Our assumption is based on this assumption. However, we need something a bit stronger. In particular, we need not just the set element $(2a)*x_\lambda$ or $b*x_\lambda$, but the ability to query on an \emph{arbitrary} set element $y$ and receive $(2a)*y$ or $b*y$. In the classical group setting, this would correspond to receiving $g^a$, and then being able to query the function $h\mapsto h^{a^2}$ or $h\mapsto h^b$. 

Note that if allowing arbitrary queries to this oracle, the problem is \emph{easy} in many cases. In particular, suppose the order of $\G_\lambda$ is odd with order $2t-1$. Then by querying the oracle $t$ times, we can compute $y_1=(2a)*x_\lambda,y_2=(2a)*y_1=(4a)*x_\lambda,\cdots$, ultimately computing $y_t=(2ta)*x_\lambda=a*x_\lambda$. On the other hand, if the oracle maps $y\mapsto b*x_\lambda$ for a random $b$, then $y_t=(tb)*x_\lambda\neq a*x_\lambda$. This allows for distinguishing the two cases.

Therefore, we only allow a \emph{single} query to the oracle. In this case, a single query does not appear sufficient for breaking the assumption. The adversary, on input $u=a*x_\lambda$, can send $u$ to the oracle, receiving $(3a)*x_\lambda$ or $(a+b)*x_\lambda$. Or it can send $x_\lambda$ to the oracle, receiving $(2a)*x_\lambda$ or $b*x_\lambda$. It can also act on these elements by known constants, computing either $(2a+c)*x_\lambda,(3a+d)*x_\lambda$, or $(b+c)*x_\lambda,(a+b+d)*x_\lambda$. It can also act on the original element $u$, and also on $x_\lambda$ by known constants, receiving $(a+e)*x_\lambda,f*x_\lambda$. Intuitively, it seems the only way the adversary can distinguish between these cases is to find constants $c,d,e,f$ that cause a collision between elements when the oracle acts by $2a$, but no collision when the oracle acts by $b$. However, for any constants $c,d,e,f$, the probability of a collision occurring in either case is negligible. Based on this intuitive argument, it is possible to prove that this assumption is generically hard against \emph{classical} algorithms. We do not, however, know if there is a clever quantum algorithm that breaks the assumption. However, it seems plausible that there is no such efficient quantum algorithm.

We will also allow the query to be quantum, and for technical reasons, we will use an \emph{in-place} (also known as \emph{minimal}) oracle, meaning it maps $\sum_g\alpha_g|g*x_\lambda\rangle\mapsto \sum_g\alpha_g|(2a+g)*x_\lambda\rangle$. This is in contrast to the usual ``standard'' oracle which maps $\sum_{g,y}\alpha_{g,y}|g*x_\lambda,y\rangle\mapsto \sum_{g,y}\alpha_{g,y}|g*x_\lambda,y\oplus |(g+2a)*x_\lambda\rangle$.

\begin{assumption}\label{def:D2X/min} The Decisional 2X Assumption with minimal oracle (D2X/min) assumption holds on a group action $(\G,\Xs,*)$ if, for all QPT adversaries $\As$, there exists a negligible $\lambda$ such that 
	\[\left|\Pr\left[\As^{M_{2a}^1}(a*x_\lambda)=1: a\gets\Gs_\lambda\right]-\Pr\left[\As^{M_b^1}(a*x_\lambda)=1: a,b\gets\Gs_\lambda\right]\right|\leq\negl(\lambda)\enspace .\]
	Above, $M_{c}$ is the in-place (or ``minimal'') oracle mapping $y\mapsto c*y$ but accessible in superposition, and $M_{c}^1$ means the adversary can make only a single quantum query to $M_c$.
\end{assumption}

\noindent If we insist on standard oracles, we can instead utilize the following assumption:
\begin{assumption}\label{def:D2X/std} The Decisional 2X Assumption with standard oracle (D2X/std) assumption holds on a group action $(\G,\Xs,*)$ if, for all QPT adversaries $\As$, there exists a negligible $\lambda$ such that 
	\[\left|\Pr\left[\As^{S_{2a}^1,S_{-2a}^1}(a*x_\lambda)=1: a\gets\Gs_\lambda\right]-\Pr\left[\As^{S_b^1,S_{-b}^1}(a*x_\lambda)=1: a,b\gets\Gs_\lambda\right]\right|\leq\negl(\lambda)\enspace .\]
	Above, $S_{c}$ is the standard oracle mapping $(y,z)\mapsto (y,z\oplus (c*y))$ but accessible in superposition, and $S_{c}^1$ means the adversary can make only a single query to $S_c$.
\end{assumption}
\noindent The following lemma is straightforward:
\begin{lemma}\label{lem:standard2min2dx} If D2X/std holds on a group action $(\G,\Xs,*)$, then so does D2X/min
\end{lemma}
\begin{proof}We simply use the oracles $S_c^1,S_{-c}^1$ to simulate the oracle $M_c^1$ in the obvious way.
\end{proof}

\paragraph{Our security proof.} We now prove the generic security of our quantum lightning scheme.

\begin{theorem}\label{thm:main} Let $(\G,\Xs,*)$ be a group with $\Xs\subseteq\{0,1\}^m$ such that D2X/min holds (Assumption~\ref{def:D2X/min}). Let $m'\geq 2m+\omega(\log\lambda)$. Let $(\gen^{\GGAM_{\G,m'}},\ver^{\GGAM_{\G,m'}})$ be the quantum money construction from Construction~\ref{constr:main}, using the generic group action $\GGAM_{\G,m'}$. Then the quantum money construction is a secure quantum lightning scheme.\end{theorem}
\begin{proof}Consider an adversary $\Bs^{\GGAM_{\G,m'}}$ for quantum lightning security, and let $\epsilon$ be the probability that $\Bs$ wins. Recall that $\ver^{\GGAM_{\G,m'}}(h,\$)$ projects onto the correct banknote $|\G^h_\lambda*L(0)\rangle$. Therefore, with probability $\epsilon$, $\Bs$ outputs $h$ and exactly two copies of the state $|\G^h_\lambda*L(0)\rangle$.
	
	We now construct an adversary $\As$ for D2X/min on the group action $(\G,\Xs,*)$. $\As$, on input $u=a*x_\lambda$, will choose a random injection $\Pi:\{0,1\}^{2m}\rightarrow\{0,1\}^{m'}$. It will then compute $X=\Pi(x_\lambda,u)$. $\As$ will then run $\Bs(X)$, simulating its queries $(\ell,g)$ to the group action as follows: compute $(z_1,z_2)\gets\Pi^{-1}(\ell)$, and then return $\Pi(g*z_1,g*z_2)$. For superposition queries, $\As$ simply runs this computation in superposition. Note that if we let $\Gamma(g)=(g*x_\lambda,g*u)$, then $\As$ simulates these queries exactly as prescribed above in our general framework, for constants $c_1=c_2=1$ and $(y_1,y_2)=(x_\lambda,u)$.
	
	Finally, when $\Bs$ produces serial number $h$ and banknotes $\$_1,\$_2$, $\As$ does the following:
	\begin{itemize}
		\item Run $\ver^{\Os'}(h,\$_1)$ and $\ver^{\Os'}(h,\$_2)$, answering the queries of $\ver$ using the simulated group action oracle. If either run rejects, output a random bit. Otherwise, let $\$_1',\$_2'$ be the resulting states of the verifier.
		\item In superposition, it applies the following map $\ell\mapsto\ell'$ to $\$_2'$:
		\begin{itemize}
			\item First map $\ell\mapsto\Pi^{-1}(\ell)=(z_1,z_2)$.
			\item Use the oracle $M_c$ from the D2X/min assumption to replace $z_1$ with $z_1'=c*z_1$, where $c=2a$ or $b$.
			\item Now map $(z_1',z_2)\mapsto\ell'=\Pi(z_2,z_1')$.
		\end{itemize}
		Let $\$_2''$ be the result of this map.
		\item Apply the swap test to $\$_1',\$_2''$, outputting whatever the swap test outputs.
	\end{itemize}
	
	By applying Lemma~\ref{lem:ggamsim}, we can conclude that $\$_1,\$_2$ are actually superpositions over elements of the form $L(g)=\Pi(g*z_1,g*z_2)$ for varying $g$. Then using our characterization of the accepting states of $\ver$, we see that both runs of $\ver$ simultaneously accept with probability $\epsilon$, and in this case $\$_1'=\$_2'=|\G_\lambda^h*L(0)\rangle,\$_2'$.
	
	We must analyze the effect of the map $\ell\mapsto\ell'$ on $|\G_\lambda^{h}*L(0)\rangle$. We break into two cases:
	\begin{itemize}
		\item $M_c$ implements the action $y\mapsto c*y$ with $c=2a$. Let $\ell=L(g)=\Pi(g*z_1,g*z_2)=\Pi(g*x_\lambda,(a+g)*x_\lambda)$, which maps to $\ell'=\Pi(g*x_\lambda,(2g)*x_\lambda)=L(a+g)$.
		
		Therefore, $|\G_\lambda^{h}*L(0)\rangle$ maps to
		\begin{align*}|\G_\lambda^{h}*L(0)\rangle&=\frac{1}{\sqrt{|\G_\lambda}}\sum_g \chi(g,h)|L(g)\rangle\\&\mapsto\frac{1}{\sqrt{|\G_\lambda|}}\sum_g \chi(g,h)|L(a+g)\rangle\\&=\frac{1}{\sqrt{|\G_\lambda|}}\sum_{g'} \chi(g'-a,h)|L(g')\rangle\\&=\chi(a,-h)|\G_\lambda^{h}*L(0)\rangle\enspace .\end{align*}
		
		Above, we used the substitution $g'=a+g$. Thus, in this case, $\As$ obtains two copies of $|\G_\lambda^{h}*L(0)\rangle$, which the swap test will accept with probability 1. Therefore, the probability $\As$ outputs 1 is $\frac{1}{2}(1-\epsilon)+\epsilon=\frac{1+\epsilon}{2}$.

		\item $M_c$ implements the action $y\mapsto c*y$ with $c=b$ for a random $b$. In this case, $\ell=L(g)=\Pi(g*x_\lambda,(a+g)*x_\lambda)$ maps to $\ell'=\Pi((a+g)*x_\lambda,(g+b)*x_\lambda)$. However, $\ell'$ is \emph{not} equal to $L(g')$ for any $g$. Indeed, in order for $\ell'=L(g')$, we get several equations:
		\[g'=a+g\enspace ,\;\;\;\;a+g'=g+b\enspace .\]
		The first equation requires that $g'=a+g$, while the last one requires that $g'=g+b-a\neq g+a$. Hence, the state $\$_2''$ has disjoint support from the state $|\G_\lambda^{h}*L(0)\rangle$, and hence is orthogonal to it. Therefore, the swap test will accept with probability exactly 1/2. The overall probability $\As$ outputs 1 is therefore exactly $1/2$.
	\end{itemize}
	Thus, we see that $\As$ has advantage $\epsilon/2$ in distinguishing DDH, breaking the assumption.
\end{proof}

\section{On Quantum Knowledge Assumptions and Algebraic Adversaries}\label{sec:knowledge}

In this section, we explore knowledge assumptions in the quantum setting, as well the algebraic model for group actions. We find significant issues with both settings. Nevertheless, we give a second security proof for our quantum lightning scheme (Construction~\ref{constr:main}), this time using knowledge assumptions.

\subsection{The Knowledge of Group Element Assumption (KGEA)}

Here, we discuss a new assumption that we define, called the Knowledge of Group Element Assumption (KGEA). This is an analog of the classical Knowledge of Exponent Assumption (KEA)~\cite{C:Damgard91}, but adapted for quantum adversaries and group actions. It can also be seen as an adaptation of the Knowledge of Path assumption of~\cite{EC:LiuMonZha23}, specialized to group actions. Despite coming from plausible origins, however, we will see that the assumption is, in fact, false. This leads to concerns over the more general Knowledge of Path assumption. We give a candidate replacement assumption that avoids our attack, but more cryptanalysis is needed to understand the new assumption.

\paragraph{The Knowledge of Group Element Assumption (KGEA).} This assumption states, informally, that any algorithm that produces a set element $y$ must ``know'' $g$ such that $y=g*x_\lambda$. Implicit in this assumption is the requirement that it is hard to obliviously sample set elements; we discuss later how to model security when oblivious sampling is possible. In the classical setting, the KGEA assumption would be formalized as follows:

\begin{assumption}\label{def:ckgea} The \emph{classical knowledge of group element assumption} (C-KGEA) holds on a group action $(\G,\Xs,*)$ if the following is true. For any probabilistic polynomial time (PPT) adversary $\As$, there exists a PPT ``extractor'' $\Es$ and a negligible $\epsilon$ such that:
	\[\Pr\left[y\in\Xs\wedge y\neq g*x_\lambda:\substack{y\gets\As(1^\lambda; r)\\g\gets\Es(1^\lambda,r)}\right]\leq\epsilon(\lambda)\enspace .\]
	Above, $r$ are the random coins given to $\As$, which are also given to $\Es$, and the probability is taken over uniform $r$ and any additional randomness of $\Es$.
\end{assumption}
In other words, if $\As$ outputs any set element, it must ``know'' how to derive that set element from $x_\lambda$, since it can compute $g$ such that $y=g*x_\lambda$ using $\Es$ and its random coins. Note that once the random coins are fixed, $\As$ is deterministic. 

As observed by~\cite{EC:LiuMonZha23}, when moving to the quantum setting, the problem with Assumption~\ref{def:ckgea} is that quantum algorithms do not have to flip random coins to generate randomness, and instead their output may be a measurement applied to a quantum state, the result being inherently randomized even if the quantum state is fixed. Thus, there is no meaningful way to give the same random coins to $\Es$.

The solution used in~\cite{EC:LiuMonZha23} is to, instead of giving $\Es$ the same inputs as $\As$, give $\Es$ the remaining state of $\As$ at the \emph{end} of the computation. This requires some care, since an algorithm can of course forget any bit of information by simply throwing it away. A more sophisticated way to lose information is to perform other measurements on the state, say measuring in the Fourier basis. The solution in~\cite{EC:LiuMonZha23} is to require that $\As$ makes no measurements at all, \emph{except} for measuring the final output. Note that the Principle of Delayed Measurement implies that it is always possible without loss of generality to move all measurements to the final output. Then $\Es$ is given both the output and the remaining quantum state of $\As$, and tries to compute $g$. Note that in the classical setting, if we restrict to \emph{reversible} $\As$, this formulation of giving $\Es$ the final state of $\As$ is equivalent to given $\Es$ the randomness, since the randomness can be computed by reversing $\As$. Similar to how we can assume a quantum $\As$ makes all its measurements at the end, in we can always assume without loss of generality that a classical $\As$ is reversible. Thus, in the classical setting these two definitions coincide. Adapting to our setting, this approach yields the following assumption:

\begin{assumption}\label{def:qkgea} The \emph{quantum knowledge of group element assumption} (Q-KGEA) holds on a group action $(\G,\Xs,*)$ if the following is true. For any quantum polynomial time (QPT) adversary $\As$ which performs no measurements except for its final output, there exists a QPT extractor $\Es$ and negligible $\epsilon$ such that 
	\[\Pr\left[y\in\Xs\wedge y\neq g*x_\lambda:\substack{(y,|\psi\rangle)\gets\As(1^\lambda)\\g\gets\Es(y,|\psi\rangle)}\right]\leq\epsilon(\lambda)\enspace .\]
\end{assumption}
Above, $y$ is considered as the output of $\As$, and the only measurements applied to $\As$ is the measurement of $y$ to obtain the output.

\paragraph{Our Attack on Q-KGEA.}\label{sec:KGEAattack} Here, we show that Q-KGEA is \emph{false}. 

\begin{theorem}\label{thm:kgeaattack} On any group action where the discrete logarithm assumption holds (Assumption~\ref{def:dlog}), Q-KGEA (Assumption~\ref{def:qkgea}) does not hold.
\end{theorem}
\begin{proof}Our proof will use the {\sf Findh} algorithm developed in Section~\ref{sec:compserial}. We first recall the functionality guaranteed by the algorithm. The algorithm takes as input the state $|\G_\lambda^h*x_\lambda\rangle=\frac{1}{\sqrt{|\G_\lambda|}}\sum_{g\in\G_\lambda}\chi(g,h)|g*x_\lambda\rangle$, and outputs $h$, while leaving $|\G_\lambda^h*x_\lambda\rangle$ intact. In other words, it maps $|\G_\lambda^h*x_\lambda\rangle\mapsto |\G_\lambda^h*x_\lambda\rangle|h\rangle$.
	
	Now, recall that the $|\G_\lambda^h*x_\lambda\rangle$ form a basis. In particular, observe that $|x_\lambda\rangle=\frac{1}{\sqrt{|\G_\lambda|}}\sum_h|\G_\lambda^h*x_\lambda\rangle$. Therefore, we have that
	\[{\sf Findh}|x_\lambda\rangle=\frac{1}{\sqrt{|\G_\lambda|}}\sum_h|\G_\lambda^h*x_\lambda\rangle|h\rangle\enspace .\]
	
	We can now apply an arbitrary $h$-dependent phase to the state, and then uncompute $h$. The result is that we have applied an arbitrary phase to whatever state we started from, but in the Fourier domain of the group. That is, let $F:\G\mapsto\R$ be an arbitrary function. We can apply the phase $|h\rangle\mapsto e^{iF(h)}|h\rangle$, and then uncompute $h$. The result is that $|x_\lambda\rangle$ maps to
	
	\begin{equation}\label{eqn:attack}\frac{1}{\sqrt{|\G_\lambda|}}\sum_h e^{iF(h)}|\G_\lambda^h*x_\lambda\rangle=\frac{1}{|\G_\lambda|}\sum_{g}|g*x_\lambda\rangle\left(\sum_h\chi(g,h)e^{iF(h)}\right)\enspace .\end{equation}
	
	Now suppose we apply Q-KGEA to the algorithm producing this state. When we measure the register, all we get is a sample of $|g*x_\lambda\rangle$ according to some distribution, with no side information. The Q-KGEA assumption then implies an algorithm $\Es$ which can recover $g$ just given $|g*x_\lambda\rangle$. Therefore, if we can guarantee that measuring the state in Equation~\ref{eqn:attack} gives a uniform choice of $g$, then $\Es$ must be solving discrete logarithms, thus breaking Assumption~\ref{def:dlog} and reaching a contradiction.
	
	It is not hard to devise a function $F$ which makes the resulting sample $g$ close uniform; a random $F$ would accomplish this, for example. With a bit more care, we can even obtain a truly uniform $g$. Indeed, suppose $\G=\Z_N$ for an odd integer $N$. Then we can let $F(h)=2\pi h^2/N$. Then the probability of observing $g$ is 
	\[\frac{1}{|\G_\lambda|^2}\times \left|\sum_he^{i2\pi (gh+h^2)/N}\right|^2=\frac{1}{|\G_\lambda|^2}\times |\G_\lambda|=\frac{1}{|\G_\lambda|}\]
	as desired, where above we used the fact about quadratic Gauss sums that $\sum_he^{i2\pi (gh+h^2)/N}$ is equal to $|\G_\lambda|^{-1/2}$, up to phase.
\end{proof}

\paragraph{Our Modified Knowledge Assumption.} We propose a simple way to circumvent the attack above. Our basic observation is that, while the attack in Theorem~\ref{thm:kgeaattack} allows for obliviously sampling elements in arbitrary group actions, it does not appear useful for actually breaking cryptosystems. After all, all the attack is doing is sampling random set elements, which can anyway be sampled easily by choosing a random group element $g$ and computing $g*x_\lambda$. Thus, while strictly speaking violating the knowledge assumption, the attack appears useless for actually breaking cryptosystems.

More generally, for ``nice'' cryptographic games (which we will define shortly), in particular games that only use the group action interface and do not themselves obliviously sample elements, it seems that giving the adversary the ability to obliviously sample elements is no help in breaking the game. We therefore postulate that, for any adversary $\As$ that wins such a nice game, there is a different adversary $\As'$ for which the KGEA assumption can be applied, yielding an extractor \emph{for that} $\As'$. Thus, even if the original $\As$ can obliviously sample elements, we essentially assume that $\As'$ cannot, and therefore $\Es$ is possible. We now make this intuition precise.

\medskip

We first introduce the notion of generic group action games. Note that we will only be interested in \emph{games} that are given by generic algorithms; we will always treat the adversary as non-generic.

Briefly, a generic group action game is given by an interactive algorithm (``challenger'') ${\sf Ch}$. ${\sf Ch}$ is limited to only performing group action computations that are ``generic'' and only interacts with the group action through oracles implementing the group action interface. Specifically, a generic algorithm is an oracle-aided algorithm $\Bs$ that has access to oracles ${\sf GA}=({\sf Start},{\sf Act},{\sf Mem})$. Here, ${\sf Start}$ is the oracle that takes as input the empty query, and outputs a string $\tilde{x}$ representing $x_\lambda$. ${\sf Act}$ is the oracle that takes as input a group element $g\in\G_\lambda$ and a string $\tilde{y}$ representing a set element $y$, and outputs a string $\tilde{z}$ representing $z=g*x$. Finally, ${\sf Mem}$ is a membership testing oracle, that tests is a given string $\tilde{x}$ represents an actual set element. From a generic game, we obtain a standard model game by implementing the oracles ${\sf Start},{\sf Act},{\sf Mem}$ with the algorithms for an actual group action: ${\sf Start}$ outputs the actual set element $x_\lambda$, ${\sf Act}$ is the group action $*$, and ${\sf Mem}$ is the membership tester for the set $\Xs_\lambda$. For a concrete group action $(\G,\Xs,*)$, we denote this standard-model game by ${\sf Ch}^{(\G,\Xs,*)}$. Note that in the quantum setting, we will allow the gave ${\sf Ch}$ to send quantum messages to and from the adversary, and make quantum queries to the oracles in ${\sf GA}$.

For any algorithm $\As$, we say the algorithm $\delta(\lambda)$-breaks ${\sf Ch}^{(\G,\Xs,*)}$ if ${\sf Ch}^{(\G,\Xs,*)}(1^\lambda)$ outputs 1 with probability at least $\delta(\lambda)$ when interacting with $\As$.

We say that ${\sf Ch}$ is one-round if it sends a single classical string to $\As$, and then receives a single quantum message from $\As$, before deciding if $\As$ wins.

\medskip

We now give our modified KGEA assumption.

\begin{assumption}\label{def:qmkgea} The \emph{quantum modified knowledge of group element assumption} (Q-mKGEA) holds on a group action $(\G,\Xs,*)$ if the following is true. Consider a one-round generic group action game ${\sf Ch}$ and any quantum polynomial time (QPT) adversary $\As$ that $1-\delta$-breaks ${\sf Ch}^{(\G,\Xs,*)}$ for a negligible $\delta$. Write the final state of $\As$ as $\rho_{1,2}$, as a joint system over two registers $1,2$, where the first register contains the state given to ${\sf Ch}^{(\G,\Xs,*)}$ and the second register contains any remaining state of $\As$. White the final state of $\As$ as $\rho_{1,2}\gets\As(1^\lambda)\Leftrightarrow {\sf Ch}^{(\G,\Xs,*)}(1^\lambda)$. Then for all such $\delta,\As,{\sf Ch}$, there exists another  negligible $\delta'$, a QPT $\As'$ that also $1-\delta'$-breaks ${\sf Ch}^{(\G,\Xs,*)}$, and moreover there exists a QPT extractor $\Es$ and negligible $\epsilon$ such that 
	\[\Pr\left[y\in\Xs\wedge y\neq g*x_\lambda:\substack{\rho_{1,2}\gets\As'(1^\lambda)\Leftrightarrow {\sf Ch}^{(\G,\Xs,*)}(1^\lambda)\\y\gets{\sf Measure}(\rho_1)\\g\gets\Es(y,\rho_2(y)}\right]\leq\epsilon(\lambda)\enspace .\]
	Above, $y\gets{\sf Measure}(\rho_1)$ means to measure $\rho_1$ (the part of $\rho_{1,2}$ contained in the first register) in the computational basis, obtaining string $y$. Then the state of the second register collapses to $\rho_2(y)$.
\end{assumption}

Intuitively, this assumption says that if $\As$ wins some game, we might not be able to apply the KGEA extractor to it. However, there is some other $\As'$ that also wins the game, and that we \emph{can} apply the KGEA extractor to.

\begin{remark}Our solution with Assumption~\ref{def:qmkgea} also resolves the problem that, for group actions based on isogenies over elliptic curves, it is \emph{classically} possible to sample certain set element obliviously, thus violating the plain KGEA assumption. A different remedy used in~\cite{EC:LiuMonZha23} explicitly assumes a probabilistic classical procedure $S()$ for obliviously sampling set elements, and modifies the KGEA assumption so that the extractor either outputs (1) an explanation relative to $x_\lambda$ \emph{or} (2) an explanation relative to some input $y$ together with the random coins $r$ that are fed into $S$ so that $y=S(r)$. This approach works, but is not robust, in the sense that if another sampling procedure is found, it would contradict even the modified assumption. Moreover, our attack in Theorem~\ref{thm:kgeaattack} shows that, when specialized to group actions, even this approach fails, since there is a quantum procedure for sampling elements that has no randomness at all, and therefore can not be explained. Our solution is robust to new sampling procedures being found as well as our quantum sampler. Nevertheless, more cryptanalysis is needed to understand if the assumption is sound.
\end{remark}

\subsection{Quantum Lightning Security Using Q-mKGEA}\label{sec:security}

Here, we give an alternative and incomparable proof of security of our quantum lightning construction to the proof given in Section~\ref{sec:ggam}. Our proof here does not require generic group actions, but instead requires our Q-mKGEA assumption. Thus, it achieves a trade-off by giving a standard-model justification, but the computational assumption is more suspect.

\paragraph{The Discrete Log Assumption, with Help.} We now define a strengthening of the Discrete Log assumption (Assumption~\ref{def:dlog}), which allows the adversary limited query access to a computational Diffie Hellman (CDH) oracle. 

\begin{assumption}\label{def:dlogminimalcdh} We say that the \emph{Discrete Log with a single minimal CDH query} assumption (DLog/1-minCDH) assumption holds if the following is true. For any QPT adversary $\As$ playing the following game, parameterized by $\lambda$, there is a negligible $\epsilon$ such that $\As$ wins with probability at most $\epsilon(\lambda)$:
	\begin{itemize}
		\item The challenger, on input $\lambda$, chooses a random $g\in\G_\lambda$. It sends $\lambda$ to $\As$
		\item $\As$ submits a superposition query $\sum_{y\in\Xs,z\in\{0,1\}^*}\alpha_{y,z}|y,z\rangle$. Here, $y$ is a set element that forms the query, and $z$ is the internal state of the adversary when making the query. The challenger responds with $\sum_{y\in\Xs,z\in\{0,1\}^*}\alpha_{y,z}|(-g)*y,z\rangle$\enspace\footnote{Note that this operation is unitary and efficiently computable since $y\mapsto(-g)*y$ is efficiently computable and efficiently reversible given $g$.}. 
		\item The challenger sends $g*x$ to $\As$.
		\item $\As$ outputs a guess $g'$ for $g$. It wins if $g'=g$.
	\end{itemize}
\end{assumption}

Note that Assumption~\ref{def:dlogminimalcdh} uses a ``minimal'' oracle for the CDH oracle, meaning it replaces $y$ with $(-g)*y$ instead of writing $(-g)*y$ to a different register. This is only a possibility because $y\mapsto(-g)*y$ is reversible; otherwise the query would not be unitary. The minimal oracle, however, is somewhat non-standard. So we here define a slightly different assumption which uses ``standard'' oracles:

\begin{assumption}\label{def:dlogstandardcdh} We say that the \emph{Discrete Log with a double standard CDH query} assumption (DLog/2-stdCDH) assumption holds if the following is true. For any QPT adversary $\As$ playing the following game, parameterized by $\lambda$, there is a negligible $\epsilon$ such that $\As$ wins with probability at most $\epsilon(\lambda)$:
	\begin{itemize}
		\item The challenger, on input $\lambda$, chooses a random $g\in\G_\lambda$. It sends $\lambda$ to $\As$.
		\item $\As$ submits a superposition query $\sum_{y\in\Xs,w,z\in\{0,1\}^*}\alpha_{y,w,z}|y,w,z\rangle$. Here, $y$ is a set element that forms the query, $w$ is a string that forms the response register, and $z$ is the internal state of the adversary when making the query. The challenger responds with $\sum_{y\in\Xs,w,z\in\{0,1\}^*}\alpha_{y,w,z}|y,w\oplus[(-g)*y],z\rangle$. 
		\item $\As$ submits a second superposition query $\sum_{y\in\Xs,w,z\in\{0,1\}^*}\alpha_{y,w,z}|y,w,z\rangle$. The challenger responds with $\sum_{y\in\Xs,w,z\in\{0,1\}^*}\alpha_{y,w,z}|y,w\oplus[g*y],z\rangle$. 
		\item The challenger sends $g*x$ to $\As$.
		\item $\As$ outputs a guess $g'$ for $g$. It wins if $g'=g$.
	\end{itemize}
\end{assumption}

\begin{lemma}\label{lem:standard2mindlog} If DLog/2-stdCDH (Assumption~\ref{def:dlogstandardcdh}) holds in a group action, then so does DLog/1-minCDH (Assumption~\ref{def:dlogminimalcdh}).
\end{lemma}
\begin{proof}Like the proof of Lemma~\ref{lem:standard2min2dx}, Lemma~\ref{lem:standard2mindlog} follows by using the two standard oracle queries to simulate a single minimal oracle query.\end{proof}

From this point forward, we will use DLog/1-minCDH as our assumption; Lemma~\ref{lem:standard2mindlog} then shows that we could have instead used DLog/2-stdCDH.

\paragraph{The security proof.} We are now ready to formally state and prove security.

\begin{theorem}\label{thm:main2} Assuming Q-mKGEA (Assumption~\ref{def:qmkgea}) and DLog/1-minCDH (Assumption~\ref{def:dlogminimalcdh}) both hold on a group action $(\G,\Xs,*)$, then Construction~\ref{constr:main} is a quantum lightning scheme.
\end{theorem}
\begin{remark}Before proving Theorem~\ref{thm:main2}, we briefly discuss how to handle the case of non-uniform attackers, since with non-uniform quantum advice quantum lightning is insecure without some modifications. Note that even against non-uniform quantum-advice attackers, DLog/1-minCDH still plausibly holds. However, Q-KGEA (Assumption~\ref{def:qkgea}) certainly does not, as a non-uniform attacker may have a $y$ hard-coded for which it does not know the discrete log with $x_\lambda$. Theorem~\ref{thm:main2} also implies that Q-mKGEA (Assumption~\ref{def:qmkgea}) does not hold in the non-uniform quantum advice setting, though this is a priori harder to see. As discussed in Section~\ref{sec:prelim}, there are several possibilities. 
	\begin{itemize}
		\item The first is to restrict to non-uniform attackers that only have classical advice. Note that Q-KGEA is still trivially false in this setting, leading to a vacuous theorem. However, Q-mKGEA may still plausibly hold.
		\item The second is to use a probabilistically generated group action, and define Q-mKGEA and DLog/1-minCDH accordingly. For quantum money security, it would suffice to have $\gen$ create the parameters of the group action and then include them in the serial number, since the serial number is generated honestly. For quantum lightning security, we would instead need the parameters to be generated by a trusted third party and then placed in a common random string (CRS). 
		\item The final option is to use the human ignorance approach~\cite{VIETCRYPT:Rogaway06}, where we explicitly state our security theorem as transforming a quantum lightning adversary into a Q-mKGEA adversary; while such Q-mKGEA adversaries exist in the non-uniform quantum advice setting without a CRS, they are presumably unknown to human knowledge. As a consequence, a quantum lightning attacker, while existing, would likewise be unknown to human knowledge. 
	\end{itemize}
	For simplicity, we state and prove Theorem~\ref{thm:main2} in the uniform setting; either probabilistically generating the group action or using human ignorance would require straightforward modifications. 
\end{remark}

\noindent We now are ready to prove Theorem~\ref{thm:main2}.

\begin{proof}[Proof of Theorem~\ref{thm:main2}]Consider a QPT quantum lightning adversary $\As'$ which breaks security with non-negligible success probability $\epsilon$. Since an adversary can always tell if it succeeded by running $\ver$, we can run $\As'$ multiple times to boost the probability of a successful break. In particular, we can run $\As'$ for $\lambda\epsilon$ times, and at except with probability $1-2^{-\Theta(\lambda)}$, at least one of the runs will succeed. This allows us to conclude without loss of generality that $\As'$ has success probability $1-2^{-\Theta(\lambda)}$. We can then invoke Q-mKGEA (Assumption~\ref{def:qmkgea}) to arrive at an adversary $\As$ which also breaks quantum lightning security with high success probability.

	By Theorem~\ref{thm:reject}, we know that if $\As$ outputs a serial number $h$, the states outputted are exponentially close to two copies of $|\G_\lambda^h*x_\lambda\rangle$.
	
	For simplicity in the following proof, we will assume the probability of passing verification is actually 1; it is straightforward to adapt the proof to the case of negligible error.
	
	Next, we purify $\As$, and assume that before measurement, $\As$ outputs a pure state $|\psi\rangle$. By our assumption that the success probability is 1, $|\psi\rangle$ will have the form
	\[|\psi\rangle=\sum_{h}\alpha_h|\phi_h\rangle|\G_\lambda^h*x_\lambda\rangle|\G_\lambda^h*x_\lambda\rangle=\frac{1}{|\G_\lambda|}\sum_{g_1,g_2,h}\alpha_h|\phi_h\rangle\chi(h,g_1+g_2)|g_1*x\rangle_{\Ms_1}|g_2*x\rangle_{\Ms_2}\enspace.\]
	Above, $|\phi_h\rangle$ are arbitrary normalized states representing whatever state the adversary contains after outputting its banknotes, and $\sum_h\|\alpha_h\|^2=1$.
	
	Now consider the adversary $\Bs$ which first constructs $|\psi\rangle$, and then measures the register $\Ms_2$ to obtain $y_2=g_2*x$. 
	\begin{claim}$g_2$ is uniform in $\G$.
	\end{claim}
	\begin{subproof}Consider additionally measuring $\Ms_1$ in the basis $\{|\G_\lambda^h*x_\lambda\rangle\}$. Since this measurement is on a different register than the measurement on $\Ms_2$, measuring $\Ms_1$ does not affect the output distribution of $\Ms_2$ (though the results may be correlated). But the measurement on $\Ms_1$ determines $h$, and conditioned on $h$, $\Ms_2$ collapses to $|\G_\lambda^h*x_\lambda\rangle$. Regardless of what $h$ is, measuring $|\G_\lambda^h*x_\lambda\rangle$ gives a uniformly random element in $\Xs$. Thus, even without measuring $\Ms_1$, the measurement of $\Ms_2$ gives a uniform element in $\Xs$.\end{subproof}
	
	Therefore, after measuring $\Ms_2$, the state $|\psi\rangle$ then collapses to \[|\psi_{g_2*x_\lambda}\rangle:=\frac{1}{\sqrt{|\G_\lambda|}}\sum_{g_1,h}\alpha_h|\phi_h\rangle\chi(h,g_1+g_2)|g_1*x\rangle_{\Ms_1}\enspace .\]
	
	\begin{claim}There is a QPT procedure ${\sf Map}$ such that ${\sf Map}(g,|\psi_{y}\rangle)=|\psi_{g*y}\rangle$.
	\end{claim}
	\begin{subproof}${\sf Map}$ simply applies the map $y\mapsto (-g)*y$ to $\Ms_1$ in superposition. Then we have that:
		\begin{align*}
			{\sf Map}(g,|\psi_{g_2*x_\lambda}\rangle)&=\frac{1}{\sqrt{|\G_\lambda|}}\sum_{g_1,h}\alpha_h|\phi_h\rangle\chi(h,g_1+g_2)|(g_1-g)*x\rangle_{\Ms_1}\\
			&=\frac{1}{\sqrt{|\G_\lambda|}}\sum_{g_1',h}\alpha_h|\phi_h\rangle\chi(h,g_1'+g+g_2)|g_1'*x\rangle_{\Ms_1}=|\psi_{(g+g_2)*y}\rangle=|\psi_{g*(g_2*y)}\rangle\enspace .
		\end{align*}
		Above we used the change of variables $g_1'=g_1-g$.
	\end{subproof}
    
	Now we invoke Q-KGEA (Assumption~\ref{def:qkgea}) on the adversary $\Bs$. Since $\Bs$ always outputs a valid set element, this means there is another QPT algorithm $\Es$ such that 
	\[\Pr[\Es(g_2*x_\lambda,|\psi_{g_2*x_\lambda}\rangle)=g_2]\geq 1-\negl(\lambda)\enspace .\]
	
	Above, the probability is over $g_2*x_\lambda$, as well as any randomness incurred when executing $\Es$. We note by a simple random self-reduction that we can insist the above probability holds for \emph{all} $g_2*x_\lambda$, where the randomness is only over $\Es$. Indeed, given $|\psi_{g_2*x_\lambda}\rangle,g_2*x_\lambda$, we can choose a random $g$ and compute $g_2'*x_\lambda$ as $g*(g_2*x_\lambda)$ where $g_2'=g+g_2$. Likewise, we can compute $|\psi_{g_2'*x_\lambda}\rangle$ as ${\sf Map}(g,|\psi_{g_2*x_\lambda}\rangle)$. This gives a random instance on which to apply $\Es$, giving $g_2'$ with probability $1-\negl(\lambda)$, regardless of $g_2$. Then we can compute $g_2=g_2'-g$. We thus compute $g_2$ with overwhelming probability, even in the worst case. We will therefore assume without loss of generality that this is the case for $\Es$.
	
	For simplicity, we will actually assume that the probability is 1; it is straightforward to handle the case the probability is negligibly close to 1. By the Gentle Measurement Lemma~\cite{Winter99}, $\Es$ can compute $g_2$ without altering the state $|\psi_{g_2*x}\rangle$. Thus, by combining $\Bs$ and $\Es$, we can compute both $|\psi_{g_2*x}\rangle$ and $g_2$ with probability 1. We can then compute ${\sf Map}(-g_2,|\psi_{g_2*x_\lambda}\rangle)=|\psi_{x_\lambda}\rangle$.
	
	We now describe a new algorithm $\Cs$ which breaks DLog/1-minCDH (Assumption~\ref{def:dlogminimalcdh}). $\Cs$ works as follows:
	\begin{itemize}
		\item It constructs $|\psi_{x_\lambda}\rangle$ as above.
		\item It makes its query to the DLog/1-minCDH challenger, setting $\Ms_1$ as the query register. This query simulates the operation ${\sf Map}(g,\cdot)$, where $g$ is the group element chosen by the challenger. Thus, at the end of the query, $\Cs$ has $|\psi_{g*x_\lambda}\rangle$.
		\item Now upon receiving $g*x_\lambda$ from the challenger, run $\Es(g*x_\lambda,|\psi_{g*x_\lambda}\rangle)$. By the guarantees of $\Es$, the output will be $g$.
	\end{itemize}
	Thus, we see that $\Cs$ breaks the DLog/1-minCDH assumption. This completes the security proof.\end{proof}

\subsection{Algebraic Group Actions.}

Next we turn to the Algebraic Group Action Model (AGAM), considered by a couple recent works~\cite{PKC:DHKKLR23,EPRINT:OrsZan23}. This is an analog of the Algebraic Group Model (AGM)~\cite{C:FucKilLos18}, adapted to group actions and quantum attackers. This model considers algebraic adversaries, which are algorithms where, any time they produce a set element output, must also ``explain'' the output in terms of the set elements the adversary saw as input. That is, if the algebraic adversary has so far been given set elements $y_1,\dots,y_\ell$, when it outputs a new element $y$, it must also output a group element $g\in \G_\lambda$ and index $i$ such that $y=g*y_i$.

In the classical world, a common refrain is that the AGM is ``between'' the generic group model and standard model. As formalized by Zhandry~\cite{C:Zhandry22b}, this is true in a particular sense: any ``nice'' security game that is secure in the standard model is also secure in the AGM, and in turn any nice security game that is secure in the AGM is also secure in the appropriate generic group model. The statements also hold true for group actions, provided we still restrict to the classical world. Here, ``nice'' comes with some important restrictions. The game must be ``single stage'', meaning there is only a single adversary interacting with the challenger. Moreover, the game must be a ``type safe'' game, which for group actions informally means the algorithms can pass set elements around and perform group action computations on them as a black box, but cannot manipulate the individual bits of the set element representations.

We might expect, therefore, that the AGAM is also ``between'' the GGAM and the standard model quantumly. However, this appears not to be the case, or at least it does not follow from any obvious adaptation of existing work. There are at least three problems.

The first is closely related to the issue with knowledge assumptions explored above. After all, the motivation for the AGAM, following the motivation from the AGM, is that we would expect the only way to output set elements is to actually derive them from existing set elements via the group action, in which case we would seem to know how to explain the new elements in terms of existing elements. In the classical setting, you can indeed show that this is true generically. However, our attack on the Q-KGEA assumption (Theorem~\ref{thm:kgeaattack}) shows that this is not true quantumly. Namely, it is possible to output a superposition of set elements where one does not ``know'' how to derive those elements from input elements.

For the second issue, consider the security game for our quantum lightning scheme. Recall that the adversary must output some $h$ along with two copies of $|\G^h*x_\lambda\rangle=\frac{1}{\sqrt{|\G_\lambda|}}\sum_g \chi(h,g)|g*x_\lambda\rangle$. An algebraic adversary would have to ``explain'' this state, meaning it must output two copies of
\[\frac{1}{\sqrt{|\G_\lambda|}}\sum_{g}\chi(h,g)|g*x_\lambda,g\rangle\enspace .\]
But here, note that if the challenger tries to verify the banknote state, the verification will actually \emph{fail}, since the state is entangled with $g$. Worse, observe that the state produced by the algebraic adversary is actually trivial to construct for any given $h$, by first constructing $\frac{1}{\sqrt{|\G_\lambda|}}\sum_{g}\chi(h,g)|g\rangle$ and then applying the group action operation. Thus, we see that the algebraic adversary can actually trivially produce two copies of the requisite state. This is in contrast to the actual banknote state $|\G^h*x_\lambda\rangle$, where it appears only possible to sample actual banknotes for a random $h$, but not produce a banknote for a given $h$; indeed the security of our scheme inherently relies on this difficulty. That is, the state required of the algebraic adversary is trivial, whereas the state required by a standard-model adversary is presumably hard to construct. This is in contrast to the classical world, where the algebraic adversary's task is always at least as hard as the real-world adversary.

The third issue is the claim that any game which is secure in the classical AGM/AGAM is also secure in the classical GGM/GGAM. This claim, or at least the classical proof of it, does not hold quantumly. This is because the proof relies on the ability to view the adversary's queries to the group/group action oracle and extract information from them. Specifically, in the classical GGM/GGAM, the only way the adversary can obtain new set elements is to act on existing elements by querying the group action. By writing down the input set and group element as well as the output group element, we can remember how we derived all set elements. Importantly, for any set element we produce, we can trace that set element back to an input set element, and see that the output element was obtained via a sequence of actions by group elements on the original input element. By multiplying these group actions together, we can explain the output element in terms of the input set element.

This strategy, however, does not work quantumly. Consider for example the hardness assumption DLog/minCDH (Assumption~\ref{def:dlogminimalcdh}). Here, the adversary can query on a superposition $\sum_y \alpha_y |y\rangle$ of set elements, and get the resulting superposition obtained by action of a secret group element $(-g)$: $\sum_y \alpha_y |(-g)*y\rangle$

In the AGAM, we would ask the adversary queries on $\sum_y \alpha_y |y,{\sf Explain}_y\rangle$, where ${\sf Explain}_y$ is an explanation of $y$ in terms of the elements the adversary has seen so far. In the case of DLog/minCDH, the only element seen by the time the adversary must make its query is $x_\lambda$, and so ${\sf Explain}_y$ is the unique $h$ such that $y=h*x_\lambda$. Thus, the adversary's query takes the form $\sum_h \alpha_{h*x_\lambda} |h*x_\lambda,h\rangle$. In response, it receives \[|\phi_{\sf AGAM}\rangle=\sum_h \alpha_{h*x_\lambda} |(h-g)*x_\lambda,h\rangle\enspace .\]

On the other hand, a generic adversary would have just \[|\phi_{\sf GGAM}\rangle=\sum_h \alpha_{h*x_\lambda} |(h-g)*x_\lambda\rangle\enspace .\]

While in the classical setting, having the extra information $h$ about $y$ does not cause problems (it can just be erased or ignored), this extra information is problematic quantumly. For example, it might be that having $|\phi_{\sf GGAM}\rangle$ allows for solving some task, whereas having $\sum_h \alpha_{h*x_\lambda} |(h-g)*x_\lambda,h\rangle$ does not. In such a case, we find that the task is hard in the AGAM, despite being easy in the GGAM and even in the standard model. In particular, if we want the AGAM to be ``between'' the GGAM and standard models, we would need to rule this situation out, meaning we would need a way to map the state $|\phi_{\sf AGAM}\rangle$ containing the explanation back to the state $|\phi_{\sf GGM}\rangle$ without the explanation. This mapping, in general, will be intractable, as it requires un-computing $h$ from $|(h-g)*x_\lambda\rangle$. 

\medskip

Based on these issues, we see that the AGAM is probably \emph{not} a reasonable model for quantum attacks, at least when the game is inherently quantum, as with the security of our quantum lightning scheme or with assumptions that allow quantum queries. On the other hand, the model might be reasonable for ``classically stated'' security games, such as ordinary discrete log or CDH. However, these problems do not arise at all for generic group actions. Therefore, based on this discussion, we posit that generic group actions should be the preferred method for analyzing cryptosystems and security games.

\section{A Construction for REGAs}\label{sec:alternate}

In this section, we give a construction for the case where the group action can only be computed efficiently for a small ``base'' set of group elements. Such group actions are known as ``restricted effective group actions'' (REGAs). 

\subsection{Some additional background}

Before giving the construction, we here provide some additional background that will be necessary for understanding the construction.

\paragraph{Groups.} Let $\G$ be a group (written additively), and $N$ an integer such that $N\times g=0$ for all $g\in\G$. $N=|\G|$ will do. Then $\G$ is a subgroup of $\Z_N^n$ for some positive integer $n$. Let $W$ be the set of vectors in $\Z_N^n$ such that $\wv\cdot g=0\bmod N$ for all $g\in\G$. $W$ is then a group, and we can therefore consider the group $(\Z_N^n)/W$ defined using the equivalence relation $\sim$, where $\uv_1\sim\uv_2$ if $\uv_1-\uv_2\in W$. $(\Z_N^n)/W$ is isomorphic to $\G$; let $\phi:\G\rightarrow(\Z_N^n)/W$ be an isomorphism. Note that for $g\in\G\subseteq\Z_N^n$ and $h\in \G$, $g\cdot\phi(h)\bmod N$ is well-defined by taking any representative $h'\in\phi(h)$ and computing $g\cdot h'\bmod N$.

Under this notation, we can re-define $\chi(g,h)$ as $e^{i2\pi g\cdot\phi(h)/N}$, which is equivalent to the definition in Section~\ref{sec:prelim}.

We associate $\Z_N$ with the interval $[-\lfloor (N-1)/2\rfloor,\lceil (N-1)/2 \rceil]$ in the obvious way, and likewise associate $\Z_N^n$ with the hypercube $[-\lfloor (N-1)/2\rfloor,\lceil (N-1)/2 \rceil]^n$. This gives rise to a notion of norm on $\Z_N^n$ by taking the norm in $\Z^n$.

\begin{lemma}\label{lem:bigelements} Let $\G$ be a subgroup of $\Z_N$. Then the number of elements $g\in\G$ such that $|g|\geq N/4$ is exactly
	$|\G|+1-2\lceil|\G|/4\rceil$.
	In particular, if $\G\neq\{0\}$, then there is at least one element $g\in\G$ has $|g|\geq N/4$.
\end{lemma}
\begin{proof}First, it suffices to consider $|\G|=N$, in other words $\G=\Z_N$: we can then lift to $N=t|\G|$, where $\G$ is embedded into $\Z_N$ by multiplying each element in $\G$ by $t$ (where multiplication is over the integers). Since $N$ is also multiplied by $t$, this preserves the number of elements with $|g|\geq N/4$.
	
	When $\G=\Z_N$, we are then simply asking for the number of elements in $[-\lfloor (|\G|-1)/2\rfloor,\lceil (|\G|-1)/2 \rceil]$ with absolute value at least $|\G|/4$. In other words, it is the combined size of the intervals $[\lceil |\G|/4\rceil,\lceil (|\G|-1)/2 \rceil]$ and $[-\lfloor (|\G|-1)/2\rfloor,-\lceil |\G|/4\rceil]$, giving a total of $\left(\lceil (|\G|-1)/2 \rceil-\lceil |\G|/4\rceil+1\right)+\left(\lfloor (|\G|-1)/2 \lfloor-\lceil |\G|/4\rceil+1\right)=|\G|+1-2\lceil|\G|/4\rceil$.
\end{proof}

\begin{lemma}\label{lem:far}Let $\Am\in\Z_N^{n\times m}$ be a matrix. Let $\G$ be the subgroup of $\Z_N^n$ generated by the columns of $\Am$. Let $B,C$ be positive integers such that $8BCm< N$. Suppose there is a distribution $\Ds$ on $[-B,B]^m$ such that $\Am\cdot \xv$ for $\xv\gets\Ds$ is negligibly close to uniform in $\G$. Then the function $f:\G\times[-C,C]\rightarrow\Z_N^m$ given by $f(g,\ev)=\Am^T\cdot\phi(g)+\ev$ is injective.
\end{lemma}
\begin{proof}Note that $\Am^T\cdot\phi(g)$ is well-defined since it is independent of the representative of $\phi(g)$. Consider a potential collision in $f$: $\Am^T\cdot\phi(g_1)+\ev_1=\Am^T\cdot\phi(g_2)+\ev_2$. By subtracting, this gives a non-zero pair $(g=g_1-g_2,\ev=\ev_1-\ev_2)$ where $\ev\in[-2C,2C]$ such that $\Am^T\cdot\phi(g)+\ev=0$ or equivalently $\Am^T\cdot\phi(g)=-\ev$. Now consider sampling $\xv\gets\Ds$, meaning $\uv=\Am\cdot\xv$ is negligibly close to uniform in $\G$. Then $\uv^T\cdot\phi(g)=\xv^T\cdot\Am^T\cdot\phi(g)=-\xv^T\cdot\ev$. On one hand, $\uv^T\cdot\phi(g)$ is statistically close to uniform in a subgroup $\G'$ of $\Z_N$, and $\G'$ is different from $\{0\}$ since $g\neq 0$. By Lemma~\ref{lem:bigelements}, the probability $|\uv^T\cdot\phi(g)|\geq N/4$ is $|\G'|+1-2\lceil|\G'|/4\rceil>0$ since $|\G'|\geq 2$. On the other hand, $|-\xv^T\cdot\ev|<2mBC\leq N/4$ always. This means the distributions of $\uv^T\cdot\phi(g)$ and $-\xv^T\cdot\ev$ must be non-negligibly far, a contradiction.
\end{proof}

\paragraph{Discrete Gaussians.} The \emph{discrete Gaussian distribution} is the distribution over $\Z$ defined as:
\[\Pr[x]=\Ds_\sigma(x):=C_\sigma e^{2\pi x^2/\sigma^2},\]
where $C_\sigma$ is the normalization constant $C_\sigma=\sum_{x\in\Z}e^{2\pi x^2/\sigma^2}$, so that $\Ds_\sigma$ defined a probability distribution. We will also define a truncated variant, denoted
\[\Ds_{\sigma,B}(x):=\begin{cases}C_{\sigma,B} e^{2\pi x^2/\sigma^2}&\text{ if }|x|\leq B\\0&\text{ otherwise}\end{cases},\]
where again $C_{\sigma,B}$ is an appropriately defined normalization constant. For large $B$, we can treat the truncated and un-truncated Gaussians as essentially the same distribution:
\begin{fact}For $\sigma\geq\omega(\sqrt{\log\lambda})$ and $B\geq \sigma\times\omega(\sqrt{\log\lambda})$, the distributions $\Ds_{\sigma}$ and $\Ds_{\sigma,B}$ are negligibly close
\end{fact}
\noindent For a vector $\rv\in\Z^m$, we write $\Ds_{\sigma,B}(\rv)=\prod_{i=1}^m\Ds_{\sigma,B}(r_i)$.

\medskip

\noindent The \emph{discrete Gaussian superposition} is the quantum state
\[|\Ds_\sigma\rangle:=\sum_{x\in\Z}\sqrt{\Ds_\sigma(x)}|x\rangle\enspace .\]
As we will generally need to restrict to finite-precision, we also consider the truncated variant
\[|\Ds_{\sigma,B}\rangle:=\sum_{x\in[-B,B]}\sqrt{\Ds_{\sigma,B}(x)}|x\rangle\enspace .\]
Again, for large enough $B$, we can treat the truncated and un-truncated Gaussian superpositions as essentially the same state:
\begin{fact}For $\sigma\geq\omega(\sqrt{\log\lambda})$ and  $B\geq \sigma\times\omega(\sqrt{\log\lambda})$, $\||\Ds_\sigma\rangle-|\Ds_{\sigma,B}\rangle\|$ is negligible.
\end{fact}
\noindent By adapting classical lattice sampling algorithms, the states $|\Ds_{\sigma,B}\rangle$ can be efficiently constructed.

\paragraph{Fourier transform pairs.} Fix an integer $N$. We will associate the set $\Z_N$ with the integers $[-\lfloor (N-1)/2\rfloor,\lceil (N-1)/2\rceil]$. Denote by $\QFT_N$ the Quantum Fourier Transform $\QFT_{\Z_N}$. We now recall some basic facts about quantum Fourier transforms. 

\begin{align*}
	\QFT_N^m\sum_{\rv\in\Z_N^m:\Am\cdot\rv=\sv}|\rv\rangle&=N^{m/2-n}\sum_{\tv\in\Z_N^n}e^{i2\pi \tv\cdot\sv/N }|\Am^T\cdot \tv\rangle\hspace{0.5em}\text{ for }\hspace{0.5em}\Am\in\Z_N^{n\times m}\\
	\QFT_N^m\sum_{\rv}\alpha_\rv\beta_\rv|\rv\rangle&=\frac{1}{N^{m/2}}\sum_{\tv,\uv}\hat{\alpha}_\tv\hat{\beta}_\uv|\uv+\tv\rangle\hspace{0.5em}\text{ for }\hspace{0.5em}\begin{subarray}{l}\sum_\tv\hat{\alpha}_\tv|\tv\rangle=\QFT_N^m\sum_\rv\alpha_\rv|\rv\rangle\\\sum_\uv\hat{\beta}_\uv|\uv\rangle=\QFT_N^m\sum_\rv\beta_\rv|\rv\rangle\end{subarray}\\
	\QFT_N |\Ds_{\sigma,\lfloor (N-1)/2\rfloor}\rangle&\approx|\Ds_{N/\sigma,\lfloor (N-1)/2\rfloor}\rangle\hspace{0.5em}\text{ for }\hspace{0.5em}\begin{subarray}{c}N\geq \sigma\times\omega(\sqrt{\log\lambda})\\ \sigma\geq\omega(\sqrt{\log\lambda})\end{subarray}
\end{align*}
Above, $\approx$ means the two states are negligibly close.

\subsection{The Construction}

Let $\G_\lambda,\Xs_\lambda,*$ be a REGA, and $\Ts=(g_1,\dots,g_m)$ a set such that $*$ can be efficiently computed for $g_i$ and $g_i^{-1}$. We can associate $\G_\lambda$ with a subgroup of $\Z_N^n$ for some integers $N,n$. We can likewise associate the list $\Ts$ with the matrix $\Am=(g_1,\cdots,g_m)\in\Z_N^{n\times m}$.

Since we can only compute the action of certain group elements, this will significantly complicate our construction. There are several issues that need to be resolved.
\begin{itemize}
	\item For both minting and verification of our original scheme, we needed the ability to apply the group action on \emph{random} group elements, which is not possible in REGAs. Our solution, following typical applications of REGAs in the literature, is to choose our random group element as a ``small'' known combination of the base group elements $g=\sum_{i=1}^m r_i g_i$ where the $r_i$ are small integers. Under mild assumptions, $g$ will be uniform, and using the representation as a small combination of the $g_i$ we can efficiently compute the action by $g$.
	\item Unfortunately, the $r_i$ are now side-information entangled with $g$ which is hard to un-compute. If the $r_i$ are left around, it they will be entangled with the banknote which will break the correctness of the scheme. Our solution is to actually treat the $r_i$ as the group element, and perform the QFT on the $r_i$ instead of on $g$. This results in a number of complications, one being that the serial number is actually now hidden, and a different quantity must be used as the serial number. This quantity also turns out to be noisy. Nevertheless, by careful analysis, we are able to show our scheme is correct, and explain how to adapt the security proof from Section~\ref{sec:qlightningsec} to our REGA scheme.
\end{itemize}

We now give the details. We will make the following assumption about the structure of $\Ts$, which is typical in the isogeny literature.

\begin{assumption}\label{assump:uniform} There is a polynomial $B$ and a distribution $\Ds^*$ on $[-B,B]^m$ such that for $\xv\gets\Ds$, $\sum_{i=1}^m x_i g_i=\Am\cdot\xv$ is statistically close to a uniform element in $\G$. \end{assumption}
Numerous examples of such $\Ds^*$ have been proposed, such as discrete Gaussians~\cite{EC:DeFGal19}, or uniform vectors in small balls relative to different norms~\cite{AC:CLMPR18,EPRINT:NOTT20}.

Let $C= N/8Bm$, which then satisfies the conditions of Lemma~\ref{lem:far}. Thus, for $\ev$ with entries in $[-C,C]^m$, the map $(g,\ev)\mapsto\Am^T\cdot\phi(g)+\ev$ is injective.

Let $\sigma\geq 16Bm/\epsilon\times\omega(\sqrt{\log\lambda})$ and $B'\geq \sigma\times\omega(\sqrt{\log\lambda})$ be polynomials. We will assume $N\geq 2B'$, which is always possible since we can take $N$ to be arbitrarily large. We will also for simplicity assume $N$ is even. This assumption is not necessary but will simplify some of the analysis, and is moreover without loss of generality since we can always make $N$ larger by multiplying it by arbitrary factors.

\begin{construction}\label{constr:alt} Construct $(\gen,\ver)$ as follows:
	\begin{itemize}
		\item $\gen(1^\lambda)$: Initialize quantum registers $\Ss$ (for serial number) and $\Ms$ (for money) to states $|\Ds_{\sigma,B'}\rangle^{\otimes m}_\Ss$ and $|0\rangle_\Ms$, respectively. Then do the following:
		\begin{itemize}
			\item Apply in superposition the map $|\rv\rangle_\Ss|y\rangle_\Ms\mapsto |\rv\rangle_\Ss|y\oplus [(\sum_{i=1}^m r_i g_i)*x_\lambda])\rangle_\Ms$. The joint state of the system $\Ss\otimes\Ms$ is then \[\sum_{\rv\in\Z_N^m}\sqrt{\Ds_{\sigma',B}(\rv)}|\rv\rangle_\Ss|(\sum_{i=1}^m r_i g_i)*x_\lambda\rangle_\Ms=\sum_{g\in\G_\lambda}\left(\sum_{\rv\in\Z_N^m:\Am\cdot\rv=g}\sqrt{\Ds_{\sigma,B'}(\rv)}|\rv\rangle_\Ss\right)|g*x_\lambda\rangle_\Ms\enspace .\]
			
			\item Apply $\QFT_{\Z_N^m}$ to $\Ss$. Using the QFT rules given above, this yields the state negligibly close to:
			\begin{align*}
				&\frac{1}{N^n}\sum_{g\in\G_\lambda}\left(\sum_{\sv,\ev\in\Z_N^n}\sqrt{\Ds_{N/\sigma,N/2-1}(\ev)}e^{i2\pi (g\cdot \sv)}|\Am^T\cdot\sv+\ev\rangle_\Ss\right)|g*x_\lambda\rangle_\Ms\\
				&=\frac{1}{|\G_\lambda|}\sum_{g\in\G_\lambda}\left(\sum_{h\in\G_\lambda,\ev\in\Z_N^n}\sqrt{\Ds_{N/\sigma,N/2-1}(\ev)}e^{i2\pi (g\cdot\phi(h))}|\Am^T\cdot\phi(h)+\ev\rangle_\Ss\right)|g*x_\lambda\rangle_\Ms\\
				&=\frac{1}{\sqrt{|\G_\lambda|}}\sum_{g\in\G_\lambda}\left(\frac{1}{\sqrt{|\G_\lambda|}}\sum_{h\in\G_\lambda,\ev\in\Z_N^n}\sqrt{\Ds_{N/\sigma,N/2-1}(\ev)}\chi(g,h)|\Am^T\cdot\phi(h)+\ev\rangle_\Ss\right)|g*x_\lambda\rangle_\Ms\enspace .
			\end{align*}
			\item Measure $\Ss$, giving the serial number $\tv:=\Am^T\cdot\phi(h)+\ev$. $\ev$ is distributed negligibly close to $\Ds_{N/\sigma}$, meaning with overwhelming probability each entry is in $[-N/16Bm,N/16Bm]= [- C/2,C/2]\subseteq [-C,C]$. This means, to within negligible error, $\tv$ uniquely determines $\phi(h)$ and hence $h$.	Therefore, the $\Ms$ register then collapses to a state negligibly close to \[\frac{1}{\sqrt{|\G_\lambda|}}\sum_{g\in\G_\lambda}\chi(g,h)|g*x_\lambda\rangle_\Ms=:|\G_\lambda^h*x_\lambda\rangle\enspace .\]
			Note that $h$ is unknown. Output $(\tv,|\G_\lambda^h*x_\lambda\rangle)$.
		\end{itemize}
		
		\item $\ver(\tv,\$):$ First verify that the support of $\$$ is contained in $\Xs_\lambda$, by applying the assumed algorithm for recognizing $\Xs_\lambda$ in superposition. Then repeat the following $\lambda$ times:
		\begin{itemize}
			\item Initialize a new register $\Hs$ to $(|0\rangle_\Hs+|1\rangle_\Hs)/\sqrt{2}$. 
			\item Choose a random element $\xv\gets\Ds^*$.
			\item Apply to $\Hs\otimes\Ms$ in superposition the map
			\[{\sf Apply}|b\rangle_\Hs|y\rangle_\Ms\mapsto \begin{cases}|0\rangle_\Hs|y\rangle_\Ms&\text{ if }b=0\enspace ,\\|1\rangle_\Hs|(-\sum_i x_i g_i)*y\rangle_\Ms\enspace&\text{ if }b=1\enspace .\end{cases}\]
			
			Since the entries of $\xv$ are bounded by $B$ which is polynomial, this step is efficient. 
			\item Measure $\Hs$ in the basis $B_{\tv,\xv}:=\{(|0\rangle_\Hs+e^{i2\pi \xv^T\cdot\tv/N}|1\rangle_\Hs)/\sqrt{2},(|0\rangle_\Hs-e^{i2\pi \xv^T\cdot\tv/N}|1\rangle_\Hs)/\sqrt{2}\}$, giving a bit $b_\xv\in\{0,1\}$. Discard the $\Hs$ register. 
			\item Accept if at least a fraction $7/8$ of the $b_\xv=0$ and the support of $\$$ is contained in $\Xs_\lambda$; otherwise reject.
		\end{itemize}
	\end{itemize}
\end{construction}

\subsection{Accepting States of the Verifier}

We now analyze the correctness of the construction.

\begin{theorem}\label{thm:rejectalternate} Let $|\psi\rangle$ be a state over $\Ms$. Then $\Pr[\ver(h,|\psi\rangle)=1]=\|\langle\psi |\G_\lambda^h*x_\lambda\rangle \|^2(1-2^{-\Omega(\sqrt{\lambda})})\pm 2^{-\Omega(\sqrt{\lambda})}$.
\end{theorem}
\begin{proof}
	For simplicity, we analyze the case of $|\psi\rangle=|\G_\lambda^{h'}*x_\lambda\rangle$, which form a basis for superpositions over $\Xs_\lambda$. In this case, Theorem~\ref{thm:rejectalternate} states that $|\G_\lambda^h*x_\lambda\rangle$ is accepted with probability $1-2^{\Omega(\sqrt{\lambda})}$, while $|\G_\lambda^{h'}*x_\lambda\rangle$ for $h'\neq h$ is accepted with probability $2^{\Omega(\sqrt{\lambda})}$. Linearity of the verifier allows us to extend to all possible states.
	
	If we let $u=\Am\cdot\xv=\sum_i x_i g_i$, then by the same analysis as in Construction~\ref{constr:main}, we have that applying ${\sf Apply}$ to the state $|\G_\lambda^{h'}*x_\lambda\rangle$ results in the state
	\begin{align*}
		&\frac{1}{\sqrt{2}}\left(|0\rangle_\Hs+\chi(u,h')|1\rangle_\Hs\right)|\G_\lambda^{h'}*x_\lambda\rangle\\
		&=\frac{1}{\sqrt{2}}\left(|0\rangle_\Hs+e^{i2\pi u\cdot\phi(h')/N}|1\rangle_\Hs\right)|\G_\lambda^{h'}*x_\lambda\rangle\\
		&=\frac{1}{\sqrt{2}}\left(|0\rangle_\Hs+e^{i2\pi \xv^T\cdot\Am^T\cdot\phi(h')/N}|1\rangle_\Hs\right)|\G_\lambda^{h'}*x_\lambda\rangle\enspace .
	\end{align*}
	Conditioned on sampling $\xv$, the probability $\Pr[b_\xv=0]$ is the inner product squared of \[\left(|0\rangle_\Hs+e^{i2\pi \xv^T\cdot\Am^T\cdot\phi(h')/N}|1\rangle_\Hs\right)/\sqrt{2}\] with the basis state \[\left(|0\rangle_\Hs+e^{i2\pi \xv\cdot\tv/N}|1\rangle_\Hs\right)/\sqrt{2}\enspace .\]
    
    This probability therefore evaluates to:
	\begin{align*}
		\Pr[b_\xv=0]&=\frac{1}{4}\left\|1+e^{i2\pi (\xv^T\cdot\Am^T\cdot\phi(h')-\xv^T\cdot\tv)/N}\right\|^2\\
		&=\frac{1}{2}\left(1+\cos\left[2\pi (\xv^T\cdot\Am^T\cdot\phi(h')-\xv^T\cdot(\Am^T\cdot\phi(h)+\ev))/N\right]\right)\\
		&=\frac{1}{2}\left(1+\cos\left[2\pi (\xv^T\cdot\Am^T\cdot\phi(h'-h)+\xv^T\cdot\ev)/N\right]\right)\enspace .
	\end{align*}

	In the case $h=h'$, $\Pr[b_\xv=0]=\frac{1}{2}\left(1+\cos\left[2\pi\xv^T\cdot\ev/N\right]\right)$. We have that $|2\pi\xv^T\cdot\ev/N|\leq \pi /8$. Using the fact that $\cos(x)\geq 1-x^2/2$, we therefore have that $\Pr[b_\xv=0]\geq 1-\pi^2/256=0.9614\ldots=7/8+\Omega(1)$. Then via standard concentration inequalities, after $\lambda$ trials, except with probability $2^{-\Omega(\sqrt{\lambda})}$, at least $7/8$ of the $b_\xv$ will be 0. Therefore, $\ver$ accepts with probability $1-2^{-\Omega(\sqrt{\lambda})}$.
	
	On the other hand, if $h\neq h'$, then $\xv^T\Am^T$ is statistically close to uniform in $\G_\lambda$, and so $\xv^T\cdot\Am^T\cdot\phi(h'-h)$ is statistically close to uniform in a non-trivial subgroup $\G'$ of $\Z_N$. By Lemma~\ref{lem:bigelements} and our assumption that $N$ is even, at least half of the elements of $\Z_N$ are at least $N/4$ in absolute value. In particular, this means $\Pr[|\xv^T\cdot\Am^T\cdot\phi(h'-h)|\geq N/4]\geq 1/2-\negl$. On the other hand, $|\xv^T\cdot\ev|\leq N/16$ always. This means $\|\xv^T\cdot\Am^T\cdot\phi(h'-h)+\xv^T\cdot\ev\|\geq N/4-N/16$ with probability at least $1/2-\negl$. In this case, we can use that $\cos(\pi/2+x)\leq |x|$ to bound $\cos\left[2\pi (\xv^T\cdot\Am^T\cdot\phi(h'-h)+\xv^T\cdot\ev)/N\right]\leq 2\pi/16=\pi/8$, meaning $\Pr[b_\xv=0]\leq 1/2+\pi/16$. Averaging over all $\xv$, we therefore have that: $\Pr[b_\xv=0]\leq \frac{3}{4}+\pi/32+\negl=0.8481\ldots=7/8-\Omega(1)$. Then via standard concentration inequalities, after $\lambda$ trials, except with probability $2^{-\Omega(\sqrt{\lambda})}$, fewer than $7/8$ of the $b_\xv$ will be 0. Therefore, $\ver$ accepts with probability $2^{-\Omega(\sqrt{\lambda})}$.
\end{proof}

\subsection{Security}

Here, we state the security of Construction~\ref{constr:alt}. 

\paragraph{Assumptions.} We first need to define slight variants of our assumptions, in order to be consistent with the more limited structure of a REGA. For example, in the ordinary Discrete Log assumption (Assumption~\ref{def:dlog}), the challenger computes $y=g*x$ for a random $g$, and adversary produces $g$. But the adversary cannot even tell if it succeeded since it cannot compute the action of $g$ in general. Instead, the adversary is required not to compute $g$, but instead to compute any short $\xv$ such that $g=\sum_i x_i g_i$. The adversary can then check that it has a solution by computing the action of $g$ using its knowledge of $\xv$. We analogously update each of our assumptions to work with the limited ability to compute the group action on REGAs.

As above, let $\G_\lambda,\Xs_\lambda,*$ be a REGA, and $\Ts=(g_1,\dots,g_m)$ a set such that $*$ can be efficiently computed for $g_i$ and $g_i^{-1}$. Let $\Ds^*,B$ be as in Assumption~\ref{assump:uniform}.

\begin{assumption}\label{def:REGAqkgea} The \emph{REGA quantum knowledge of group element assumption} (REGA-Q-KGEA) holds on a group action $(\G,\Xs,*)$ if the following is true. For any quantum polynomial time (QPT) adversary $\As$ which performs no measurements except for its final output, there exists a polynomial $C$, a QPT extractor $\Es$ with outputs in $[-C,C]^m$, and negligible $\epsilon$ such that 
	\[\Pr\left[y\in\Xs\wedge y\neq g*x_\lambda:\substack{(y,|\psi\rangle)\gets\As(1^\lambda)\\\xv\gets\Es(y,|\psi\rangle)\\g\gets\sum_i x_i g_i}\right]\leq\epsilon(\lambda)\enspace .\]
\end{assumption}

As with the non-REGA Q-KGEA assumption, we expect the REGA-Q-KGEA assumption is likely false. Certainly it is false on group actions with oblivious sampling. However, we note that it is unclear if our attack from Theorem~\ref{thm:kgeaattack} can be adapted to REGAs. Nevertheless, to mitigate any risks associated with the plain REGA-Q-KGEA assumption, we can likewise define a \emph{modified} REGA KGEA assumption (REGA-Q-mKGEA), in the same spirit as Assumption~\ref{def:qmkgea}.

We next define our REGA analog of Assumption~\ref{def:dlogminimalcdh}.

\begin{assumption}\label{def:REGAdlogminimalcdh} We say that the \emph{REGA Discrete Log with a single minimal CDH query} assumption (REGA-DLog/1-minCDH) assumption holds if the following is true. For any QPT adversary $\As$ playing the following game, parameterized by $\lambda$, there is a negligible $\epsilon$ such that $\As$ wins with probability at most $\epsilon(\lambda)$:
	\begin{itemize}
		\item The challenger, on input $\lambda$, chooses a random $g\in\G_\lambda$. It sends $\lambda$ to $\As$
		\item $\As$ submits a superposition query $\sum_{y\in\Xs,z\in\{0,1\}^*}\alpha_{y,z}|y,z\rangle$. Here, $y$ is a set element that forms the query, and $z$ is the internal state of the adversary when making the query. The challenger responds with $\sum_{y\in\Xs,z\in\{0,1\}^*}\alpha_{y,z}|(-g)*y,z\rangle$. 
		\item The challenger sends $g*x$ to $\As$.
		\item $\As$ outputs a $\xv\in\Z^m$, encoded in unary. It wins if $g=\sum_i x_i g_i$.
	\end{itemize}
\end{assumption}

Note that the challenger in Assumption~\ref{def:REGAdlogminimalcdh} is inefficient on a REGA. However, under Assumption~\ref{assump:uniform}, the challenger can be made efficient by first sampling $\yv\gets\Ds^*$ and then computing $g=\sum_i y_i g_i$.

\begin{theorem}\label{thm:alt} Assuming REGA-DLog/1-minCDH (Assumption~\ref{def:REGAdlogminimalcdh}) and REGA-Q-KGEA (Assumption~\ref{def:REGAqkgea}) (or more generally, REGA-Q-mKGEA) both hold on a group action $(\G,\Xs,*)$, then Construction~\ref{constr:alt} is a quantum lightning scheme. Alternatively, if D2X/min (Assumption~\ref{def:D2X/min}) holds on a group action with $\Xs\subseteq\{0,1\}^m$, then Construction~\ref{constr:alt} is a quantum lightning scheme in the generic group action model $\GGAM_{\G,m'}$ with label length $m'$.
\end{theorem}

We only sketch the proof. Like in the proof of Theorems~\ref{thm:main} and~\ref{thm:main2}, we can assume the adversary wins the quantum lightning experiment with probability $1-\negl(\lambda)$. In order for a supposed note $\$$ to be accepted relative to serial number $\tv$ with overwhelming probability, $\tv$ must have the form $\tv=\Am^T\cdot\phi(h)+\ev$ for ``short'' $\ev$, and $\$$ must be negligibly close to $|\G_\lambda^h*x_\lambda\rangle$. Therefore, a quantum lightning adversary outputs two copies of $|\G_\lambda^h*x_\lambda\rangle$ for some $h$. The security reduction of Theorem~\ref{thm:main} did not rely on knowing $h$, just that the adversary outputted two copies of $|\G_\lambda^h*x_\lambda\rangle$. Hence, a near-identical proof holds for Construction~\ref{constr:alt}. The only difference is that when the extractor $\Es$ outputs a group element, it instead outputs a small linear combination of the $g_i$ giving that group element, and then the DLog/1-minCDH adversary uses this small representation to compute the action by that group element.

\section{Further Discussion}\label{sec:discuss}

\subsection{Quantum Group Actions}

Here, we consider a generalization of group actions where set elements are replaced with quantum states.

A quantum (abelian) group action consists of a family of (abelian) groups $\G=(\G_\lambda)_\lambda$ (written additively), a family $\Xs=(\Xs_\lambda)_\lambda$ of sets $\Xs_\lambda$ of quantum states over a system $\Ms_\lambda$, and an operation $*$. We will require that the states in $\Xs_\lambda$ are orthogonal. $*$ is a quantum algorithm that takes as input a group element $g\in\G_\lambda$ and a quantum state $|\psi\rangle$ over $\Ms_\lambda$, and outputs another state over $\Ms_\lambda$. $*$ satisfies the following properties:
\begin{itemize}
	\item {\bf Identity:} If $0\in \G_\lambda$ is the identity element, then $|0\rangle*|\psi\rangle=|\psi\rangle$ for any $|\psi\rangle\in \Xs_\lambda$.
	\item {\bf Compatibility:} For all $g,h\in \G_\lambda$ and $|\psi\rangle\in \Xs_\lambda$, $(g+h)*|\psi\rangle=g*(h*|\psi\rangle)$.
\end{itemize}
We can also relax the above properties to only hold to within negligible error, and/or relax the orthogonality requirement to being near-orthogonal. We will additionally require the following properties:
\begin{itemize}
	\item {\bf Efficiently computable:} There is a pseudo-deterministic QPT procedure ${\sf Construct}$ which, on input $1^\lambda$, outputs a description of $\G_\lambda$ and an specific element $|\psi_\lambda\rangle\in\Xs_\lambda$. The operation $*$ is also computable by a QPT algorithm.
	\item {\bf Efficiently Recognizable:} There is a QPT procedure ${\sf Recog}$ which recognizes elements in $\Xs_\lambda$. That is, ${\sf Recog}(1^\lambda,\cdot)$ projects onto the span of the states in $\Xs_\lambda$.
	\item {\bf Regular:} For every $|\phi\rangle\in\Xs_\lambda$, there is exactly one $g\in\G_\lambda$ such that $|\phi\rangle=g*|\psi_\lambda\rangle$. 
\end{itemize}
Again, we can also relax the above properties to only hold to within negligible error.

\paragraph{Cryptographic group actions.} At a minimum, a cryptographically useful quantum group action will satisfy the following discrete log assumption:
\begin{assumption}\label{def:qdlog} The \emph{discrete log assumption} (DLog) holds on a quantum group action $(\G,\Xs,*)$ if, for all QPT adversaries $\As$, there exists a negligible $\lambda$ such that 
	\[\Pr[\As(g*|\psi_\lambda\rangle)=g:g\gets\G_\lambda]\leq\negl(\lambda)\enspace .\]
\end{assumption}

Note that if we do not insist on orthogonality of the states in $\Xs_\lambda$, then it is trivial to construct a quantum group action in which DLog holds: simply have all $|\psi\rangle\in\Xs_\lambda$ be identical, or negligibly close. Then it will be information-theoretically impossible to determine $g$. Orthogonality essentially says that the group action is classical, except that the basis for the set elements is potentially different than the computational basis.

\subsection{Quantum Group Actions From Lattices}

Here, we describe a simple quantum group action from lattices.

The group $\G_{\sf LWE,N,n,m,\sigma}$ will be set to $\Z_N^n$ for some integers $N,n$. We will fix a short wide matrix $\Am\in\Z_N^{n\times m}$; we can think of $\Am$ as being sampled randomly and included in a common reference string. Note that $\G$ is independent of $\sigma$, but we include it for notational consistency.

The set $\Xs_{\sf LWE,N,n,m,\sigma}$ will be the set of states $|\psi_\sv\rangle=\sum_{\ev\in\Z_N^n}\sqrt{\Ds_{\sigma,N/2}(\ev)}|\Am^T\cdot\sv+\ev\rangle$. In other words, we take the discrete Gaussian vector superposition of some width, and add the vector $\Am^T\cdot\sv$.

$\G_{\sf LWE,N,n,m,\sigma}$ acts on $\Xs_{\sf LWE,N,n,m,\sigma}$ in the following obvious way: $\rv*|\psi_\sv\rangle=|\psi_{\rv+\sv}\rangle$, which can be computed by simply adding $\Am^T\cdot\rv$ in superposition.

We have the following theorem:
\begin{theorem}\label{def:qdlogfromlwe} Let $\sigma,\sigma_0$ be non-negative real numbers such that $\sigma/\sigma_0$ is super-polynomial. Assuming the Learning with Errors problem is hard for noise distribution $\Ds_{\sigma_0}$, discrete logarithms are hard in the group action $(\G_{\sf LWE,N,n,m,\sigma},\Xs_{\sf LWE,N,n,m,\sigma},*)$.
\end{theorem}
\begin{proof}The learning with errors assumption states that it is hard to compute $\sv$ given $\Am^T\cdot\sv+\ev$ with $\ev$ sampled from $\Ds_{\sigma_0}$. We need to show that it is hard to compute $\sv$ given the analogous superposition over $\Am^T\cdot\sv+\ev$, where here $\ev$ comes from the Gaussian superposition $|\Ds_\sigma\rangle$. The idea is a simple application of noise flooding: given $\uv=\Am^T\cdot\sv+\ev$, compute the state $|\psi_\sv'\rangle:=\sum_{\ev'\in\Z_N^n}\sqrt{\Ds_{\sigma,N/2}(\ev')}|\Am^T\cdot\sv+\ev+\ev'\rangle$. Since $\sigma/\sigma_0$ is super-polynomial, $\ev+\ev'$ where $\ev'\gets \Ds_{\sigma,N/2}$ is negligibly close to a Gaussian centered at 0. Therefore, $|\psi_\sv'\rangle$ is negligibly close to $|\psi_\sv\rangle$. Plugging into a supposed DLog adversary then gives $\sv$, breaking LWE.
\end{proof}

Unfortunately, this LWE-based group action is missing a crucial feature: it is not possible to recognize states in $\Xs$. In particular, the states in $\Xs$ are indistinguishable from states of the form $\sum_{\ev\in\Z_N^n}\sqrt{\Ds_{\sigma,N/2}(\ev)}|\vv+\ev\rangle$, where $\vv$ is an arbitrary vector in $\Z_N^m$. As we will see in the next subsection, the inability to recognize $\Xs$ will prevent us from using this group action to instantiate our quantum money scheme.

\subsection{Relation to Quantum Money Approaches based on Lattices}

Here, we see that our quantum money scheme is conceptually related to a folklore approach to building quantum money from lattices. This approach has not been able to work; in our language, the reason is exactly due to the inability to recognize $\Xs_{\sf LWE,N,n,m,\sigma}$. 

The approach is the following. Let $\Am\in\Z_N^{n\times m}$ be a random short wide matrix over $\Z_n$. To mint a banknote, construct the discrete Gaussian superposition $|\Ds_\sigma\rangle^{\otimes m}$ in register $\Ms$. Then compute and measure $\Am\cdot \xv$ applied to $\Ms$. The result is a vector $\hv\in\Z_N^n$, which will be the serial number, and $\Ms$ collapses to a superposition $|\$_\hv\rangle\propto \sum_{\xv:\Am\cdot\xv=\hv} \sqrt{\Ds_\sigma(\xv)}|\xv\rangle$ of short vectors $\xv$ such that $\Am\cdot\xv=\hv$. This is the banknote. A simple argument shows that it is impossible to construct two copies of $|\$_\hv\rangle$ for the same $\hv$: given such a pair, measure each to get $\xv,\xv'$ such that $\Am\cdot\xv=\Am\cdot\xv'=\hv$. Then subtract to get a short vector $\xv-\xv'$ such that $\Am\cdot(\xv-\xv')=0^n$. We can conclude $\xv-\xv'$ is non-zero with overwhelming probability, since the measurement of $|\$_\hv\rangle$ has high entropy. Such a non-zero short kernel vector would solve the Short Integer Solution (SIS) problem, which is widely believed to be hard and is the foundation of lattice-based cryptography.

Unfortunately, the above approach is broken. The problem is that there is no way to actually verify banknotes. One can verify that a banknote has support on short vectors with $\Am\cdot\xv=\hv$, but it is impossible to verify that the banknote is in superposition. If one could solve the Learning with Errors (LWE) problem, it would be possible to verify banknotes as follows: first perform the QFT to the banknote state. If an honest banknote, the QFT will give a state negligibly close to
\begin{equation}\label{eqn:banknote}|\$_\hv'\rangle:=\frac{1}{N^{n/2}}\sum_{\sv,\ev\in\Z_N^n}\sqrt{\Ds_{N/\sigma}(\ev)}e^{i2\pi \hv\cdot\sv/N}|\Am^T\cdot\sv+\ev\rangle\enspace .\end{equation}
The second step is to simply apply the supposed LWE solver to this state in superposition, ensuring that the state has support on vectors of the form $\Am^T\cdot\sv+\ev$ for small $\ev$. 

Unfortunately, LWE is likely hard. In fact, it is quantumly equivalent to SIS~\cite{DBLP:journals/jacm/Regev09}, meaning if one could verify banknotes using an LWE solver, then SIS is easy. Not only does this mean we are reducing from an easy problem, but it would be possible to turn such a SIS algorithm into an attack.

Without the ability to verify that banknotes are in superposition, the attacker can simply measure a banknote to get $\xv$, and then pass off $|\xv\rangle$ as a fake banknote that will pass verification. Since $\xv$ is trivially copied, this would break security. Interestingly,~\cite{C:LiuZha19} prove that, no matter what efficient verification procedure is used, even if the verification diverged from the LWE-based approach above, this attack works.~\cite{EC:LiuMonZha23} extend this to a variety of potential schemes based on similar ideas, including a recent proposed instantiation of this approach by~\cite{KLS22}.

\medskip

We now see how the above approach is essentially equivalent to our construction of quantum money from group actions, instantiated over our LWE-based quantum group action. The inability to recognize $\Xs$ is the reason this instantiation is insecure, despite natural hardness assumptions presumably holding on the group action.

We consider the quantum group action $(\G_{\sf LWE,N,n,m,N/\sigma},\Xs_{\sf LWE,N,n,m,N/\sigma},*)$, where $\sigma$ is from the folklore construction above. When applied to $(\G_{\sf LWE,N,n,m,N/\sigma},\Xs_{\sf LWE,N,n,m,N/\sigma},*)$, a banknote in our scheme, up to negligibly error from truncating discrete Gaussians, is the state $|\$_\hv'\rangle$ from Equation~\ref{eqn:banknote} above, where the serial number is $\hv$. Thus, we see that our quantum money scheme is simply the folklore construction but moved to the Fourier domain. The attack on the folklore construction can therefore easily be mapped to an attack on our scheme: if the adversary is given $|\$_\hv'\rangle$, it measures in the Fourier domain (which is the primal domain for the folklore construction) to get a short vector $\xv$ such that $\Am\cdot\sv=\hv$. Then it switched back to the primal domain, giving the state
\[\frac{1}{N^{m/2}}\sum_\uv e^{i2\pi\ev\cdot\xv}|\xv\rangle\enspace .\]
This is a state that lies outside the span of $\Xs$. However, no efficient verification procedure can distinguish it from an honest banknote state.

\medskip

Two features distinguish isogeny-based group actions from the LWE-based action above. The first is the ability to recognize elements in $\Xs$. Suppose it were possible to recognize elements of $\Xs$ in the LWE-based action, and we had the verifier check to see if the banknote belonged to the span of the elements in $\Xs$. In the language of quantum group actions, this check would prevent the attacker from sending $\frac{1}{N^{m/2}}\sum_\uv e^{i2\pi\ev\cdot\xv}|\xv\rangle$, which lies outside the span of $\Xs$. In the language of the folklore construction, this check would correctly distinguish between an honest banknote and the easily clonable state $|\xv\rangle$ in the attack. If such a check were possible, the proof sketched above would work to base the security of the scheme on SIS. Unfortunately, such a check is computationally intractable under the decision LWE problem, which is equivalent to SIS and most likely hard.

The issue of recognizing set elements is also crucial in our security arguments. Indeed, the first step in our proof was to characterize the states accepted by the verifier, showing that only honest banknote states are accepted. This step in the proof fails in the LWE-based scheme, which would prevent the proof from going through. Thus, even though the scheme based on LWE is broken, it does not contradict our DLog/1-minCDH and Q-KGEA assumptions holding on the LWE-based group action.

The second difference, is that, with the LWE-based group action, taking the QFT of money states gives elements with meaningful structure: short vectors $\xv$ such that $\Am\cdot\xv=\hv$. This structure and its relation to the original money state are what enables the attack. In contrast, taking the QFT of money states over $\Xs$ coming from isogenies will give terms with no discernible structure.

We believe the above perspective adds to the confidence in our proposal. Indeed, in the LWE-based scheme, the key missing piece is recognizing set elements; if not for this missing piece the scheme \emph{could} be proven secure. By switching to group actions based on isogenies, we add the missing piece. The hope is that even though the source of hardness is now from hard problems on isogenies over elliptic curves instead of lattices, by adding the missing piece we can finally obtain a scheme.

\newpage

\printbibliography


\end{document}